\documentclass[12pt,reqno]{amsart}
\usepackage{amsmath,amsfonts,amssymb,amscd,amsthm,amsbsy, color}
\usepackage{graphicx}
\usepackage{verbatim}
\usepackage{color}
\usepackage{caption,subcaption}
\usepackage{url}

\usepackage{xcolor}

\theoremstyle{plain}
\newtheorem{theorem}{Theorem}[section]
\newtheorem{proposition}[theorem]{Proposition}

\newtheorem{corollary}[theorem]{Corollary}

\newtheorem{remark}[theorem]{Remark}

\theoremstyle{definition}

\title[The Spatial Hill four-body problem I]{The Spatial Hill four-body problem I - an exploration of basic invariant sets}

\author[J. Burgos-Garc\'ia]{Jaime Burgos-Garc\'ia}
\address{Facultad de Ciencias F\'isico Matem\'aticas, Universidad Aut\'onoma de Coahuila, Unidad Campo Redondo, Edificio A, 25020 Saltillo, Coahuila, M\'exico}
\email{jburgos@uadec.edu.mx }

\author[A. Bengochea]{Abimael Bengochea}
\address{Department of Mathematics, ITAM, R\'io Hondo 1, 01080 Ciudad de M\'exico, M\'exico}
\email{abimael.bengochea@itam.mx}

\author[L. Franco]{Luis Franco-P\'erez}
\address{Departamento de Matem\'aticas Aplicadas y Sistemas, Universidad Aut\'onoma Metropolitana Unidad Cuajimalpa, Mexico City 05348, Mexico}
\email{lfranco@cua.uam.mx}

\date{}

\begin{document}

\maketitle

\begin{abstract}
In this work we perform a first study of basic invariant sets of the spatial Hill's four-body problem, where we have used both analytical and numerical approaches. This system depends on a mass parameter $\mu$ in such a way that the classical Hill's problem is recovered when $\mu=0$. Regarding the numerical work, we perform a numerical continuation, for the Jacobi constant $C$ and several values of the mass parameter $\mu$ by applying a classical predictor-corrector method, together with a high-order Taylor method considering variable step and order and automatic differentiation techniques, to specific boundary value problems related with the reversing symmetries of the system. The solution of these boundary value problems defines initial conditions of symmetric periodic orbits. Some of the results were obtained departing from periodic orbits within Hill's three-body problem. The numerical explorations reveal that a second distant disturbing body has a relevant effect on the stability of the orbits and bifurcations among these families. We have also found some new families of periodic orbits that do not exist in the classical Hill's three-body problem; these families have some desirable properties from a practical point of view.
\end{abstract}

\begin{center}
\end{center}

\section{Introduction}

During the last century, with the incorporation of computers as a tool for exploring systems with a rich dynamical structure, several techniques were developed in order to compute periodic orbits in few body problems as the restricted three-body problem (from now on referred to as the R3BP). Such techniques mainly deal with the computation of normal forms, on topological and variational methods, and numerical methods. Restricting only to numerical approach, it is worth mentioning the pioneering works by E. Str\"{o}mgren and his collaborators \cite{Stromgren}, M. H\'enon \cite{HenI}, \cite{HenII}. The interested reader can consult the chapters 8 and 9 of \cite{Sz} for a detailed compilation of the early works performed on the R3BP.
\newline

Nevertheless, in some cases a three-body approach is not enough to model particular dynamics of some small bodies in the Solar System. For instance, in Celestial Mechanics, it has been customary to use the term `Trojan' to describe a small body (an asteroid or a natural satellite) in an equilateral triangle configuration, relative to a rotating frame of reference, with two other bodies. Actually, this kind of configuration is a true solution of the three-body problem and it was one of the first triumphs of the Celestial Mechanics. In our Solar System there are well known examples of Trojan asteroids being those corresponding to the Sun-Jupiter system the first ones discovered. However, Trojan asteroids have been observed in Mars-Sun and Neptune-Sun systems. Between Saturn and some of its moons there are also equilateral triangle configurations, such as Saturn--Tethys--Telesto, Saturn--Tethys--Calypso, and Saturn--Dione--Helen. Thus, celestial bodies forming equilateral triangle configurations are not at all exceptional.
\newline

In this paper we consider a small particle interacting with the gravitational forces produced by three bodies, called primaries, in an equilateral triangle configuration, and so a four-body model becomes necessary; this problem is known as the restricted four-body problem (R4BP) where there have been recent numerical explorations of periodic orbits and computer assisted proofs of these objects (see for instance \cite{PapaI}, \cite{BurgosDelgado}, \cite{BurgosLessardJames} and references therein). Nevertheless, astronomical data show that for the equilateral triangle configurations in the Solar System the mass of one of the primaries (the Trojan) is much smaller than the others and the distances of some bodies orbiting a Trojan are quite small in comparison with the distance among the primaries. Therefore, the Hill's four body problem (H4BP) \cite{BurgosGidea} is a similar approximation to the one developed by G. W. Hill in his works on the orbit of the Moon where he provides an approximation of the dynamics of the massless particle in an neighbourhood of the smallest primary.
\newline

The first results in the H4BP show the existence of equilibrium points that do not appear in the classical case, providing a richer dynamics, in particular they influence the stability of planar periodic orbits (see \cite{BurgosGidea} and \cite{Burgos}). However, from a practical point of view, it is more realistic to consider spatial dynamics with potential applications since it is well known that, for instance, the R3BP has been exploited to construct preliminary orbits for some space missions, relying on the numerical computation of invariant manifolds and continuation of periodic orbits. It is worth mentioning the works \cite{Belbruno,Broucke,Koon}, where the authors used the  R3BP  to obtain   \textit{ideal} orbits that can be used as starting points to design real trajectories. Furthermore, in recent years, the exploration of Trojan asteroids has been included in the 2013 Decadal Survey among the New Frontiers missions in the decade 2013-2022. So, the above discussion motivates the exploration of the basic families of spatial periodic orbits in the H4BP.
\newline

It is known that there exist extrasolar systems formed by a star with similar mass as the Sun and a giant gas planet as the system $HD28185$ \cite{Schwarz} where, under the hypothesis of existence of planets forming an equilateral configuration with the star and the gas giant planet, the H4BP is a natural frame for studying the dynamics around the small planet. Another potential application of this work is related to a Keplerian binary system tidally perturbed by a normally incident circularly polarized gravitational wave \cite{Chicone}. In the next Section we will see that the equations of motion for the H4BP and this problem are quite similar. Thus, the results of this work could be extended to the above mentioned model of gravitational radiation.
\newline

This paper is organized as follows: in Section \ref{sec:H4BP} we will recall the construction and some basic aspects of the H4BP. In Section \ref{sec:reversing} we present the reversing symmetries of the equations of motion, usefull for computing symmetric periodic orbits, and classify the computed periodic orbits according to their symmetries. In Section \ref{sec:invariantsets} we perform an analytical study of some invariant sets and its dynamical properties, specifically $saddle\times center\times center$ equilibrium points with the corresponding resonant periodic orbits, and the description of the invariant $z$-axis. We offer some theoretical and computational details of the implementation of the variable step and order Taylor method for the equations of motion and the variational ones in Section \ref{sec:autoDiff}, whereas in Section \ref{sec:numericalresults} we present a systematic numerical exploration of some basic families of symmetric periodic orbits and their properties. Finally, Section \ref{sec:conclusions} contains the conclusions and perspectives of the work.

\section{The Hill  approximation in the four body problem}\label{sec:H4BP}
The restricted four-body problem is a natural extension of the celebrated restricted three-body problem and, although it is true that this problem is not the main subject of study of this work, we recall some basic properties which will help us to get a better understanding of Hill's approximation. Consider three point masses moving under mutual Newtonian gravitational attraction in circular periodic orbits around their center of mass forming an equilateral triangle configuration. A fourth massless particle is moving under the gravitational attraction of the primaries without affecting their motion and the problem is to determine the dynamics of the massless particle. We will and refer to $m_1$ as the primary, $m_2$ as the secondary, and $m_3$ as the tertiary. The configuration of those particles is shown in Figure \ref{triangle}.
\newline

\subsection{Equations of motion}\label{equationsofmotion}

If we assume that the primaries move in circular orbits with constant angular velocity then the equations of motion  in dimensionless coordinates relative to a synodic frame of reference that rotates together with the point masses are:
\begin{figure}
\centering
\includegraphics[width=3.5in]{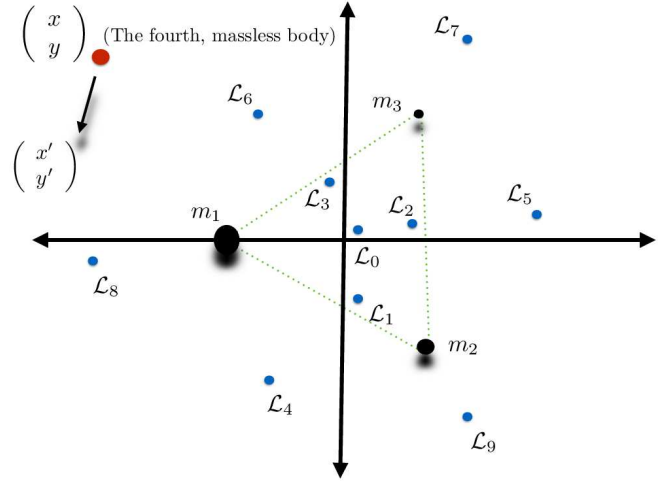}
\caption{The restricted four-body problem in a synodic frame of reference.\label{triangle}}
\end{figure}

\begin{equation}\begin{split}\label{ecuacionesfinales}
\ddot{x}-2\dot{y}&=\Omega_{x},\\
\ddot{y}+2\dot{x}&=\Omega_{y},\\
\ddot{z}&=\Omega_{z},\end{split}
\end{equation}
where $$\Omega(x,y,z)=\frac{1}{2}(x^{2}+y^{2})+\sum_{i=1}^{3}\frac{m_{i}}{r_{i}},$$
and $r_{i}=\sqrt{(x-x_{i})^{2}+(y-y_{i})^{2}+z^2}$, for $i=1,2,3$. Here $(x,y,z)$ represents the position vector of the massless particle, and $(x_i,y_i)$, $i=1,2,3$, denote the coordinates of the primaries which can be written in terms of the masses of the primaries as
\begin{equation*}
\begin{split}
x_{1}&=\frac{-\vert K\vert\sqrt{m_{2}^{2}+m_{2}m_{3}+m_{3}^{2}}}{K},\\
y_{1}&=0,\\
z_1&=0,\\
x_{2}&=\frac{\vert K\vert[(m_{2}-m_{3})m_{3}+m_{1}(2m_{2}+m_{3})]}{2K\sqrt{m_{2}^{2}+m_{2}m_{3}+m_{3}^{2}}},\\
y_{2}&=\frac{-\sqrt{3}m_{3}}{2m_{2}^{3/2}}\sqrt{\frac{m_{2}^{3}}{m_{2}^{2}+m_{2}m_{3}+m_{3}^{2}}},\\
z_2&=0,\\
x_{3}&=\frac{\vert K\vert}{2\sqrt{m_{2}^{2}+m_{2}m_{3}+m_{3}^{2}}},\\ y_{3}&=\frac{\sqrt{3}}{2\sqrt{m_{2}}}\sqrt{\frac{m_{2}^{3}}{m_{2}^{2}+m_{2}m_{3}+m_{3}^{2}}},\\
z_3&=0, \end{split}
\end{equation*}
where $K=m_{2}(m_{3}-m_{2})+m_{1}(m_{2}+2m_{3})$. For these equations the units are choosen in such a way $m_{1}+m_{2}+m_{3} = 1$. The equations of motion have a first integral (Jacobi):

\begin{equation} \label{jacobiintegral}
C=-(\dot{x}^2+\dot{y}^2+\dot{z}^2)+2\Omega.
\end{equation}

The dynamics of the R4BP is much more complex than the one associated to R3BP. Actually, there has been a considerable progress in understanding the basic but important aspects of this problem.  We refer to the interested reader to the works \cite{BurgosDelgado}, \cite{BurgosDelgadoII}, \cite{BurgosLessardJames}, \cite{PapaI}, \cite{Leandro}, \cite{KepleyII}  and references therein for a deeper discussion in the dynamics of this fascinating system.
\newline

The equations \eqref{ecuacionesfinales} are derived from the Hamiltonian
\begin{equation}
H=\frac{1}{2}(p_{x}^{2}+p_{y}^{2}+p_{z}^{2})+yp_{x}-xp_{y}-\frac{m_{1}}{r_{1}}-\frac{m_{2}}{r_{2}}-\frac{m_{3}}{r_{3}},\label{originalhamiltonian}
\end{equation}
which can be used to study dynamics around $m_3$, according to the next theorem \cite{BurgosGidea}.
\begin{theorem} \label{main theorem} After the symplectic scaling
$$(x,y,z,p_{x},p_{y},p_{z})\rightarrow m_{3}^{1/3}(x,y,z,p_{x},p_{y},p_{z}),
$$ the limit $m_{3}\rightarrow0$ of Hamiltonian (\ref{originalhamiltonian}) restricted to a neighborhood of $m_{3}$ exists and yields a new Hamiltonian
\begin{equation}\begin{split}\label{hillhamiltonian}
H=&\frac{1}{2}(p^{2}_{x}+p^{2}_{y}+p_{z}^{2})+yp_{x}-xp_{y}+\frac{1}{8}x^2-\frac{3\sqrt{3}}{4}(1-2\mu)xy-
\frac{5}{8}y^2+\frac{1}{2}z^2\\
&-\frac{1}{\sqrt{x^2+y^2+z^2}},
\end{split}\end{equation}where $m_{1}=1-\mu$ and $m_{2}=\mu$.
\end{theorem}

The resulting Hamiltonian also gives rise to a three-degree of freedom system depending on a parameter $\mu$ which represents the relative mass $m_{2}$. We can think this system as a Kepler problem, for the body with infinitesimal mass and the  tertiary, placed at origin, with additional quadratic terms produced by the gravitational influence of two large bodies placed at infinite distance from the tertiary. When we consider $\mu=0$ in the Hamiltonian (\ref{hillhamiltonian}), we recover the system of the classical Hill's problem. This will be obvious if we write the equations in a more recognizable form as we mention in the following corollary \cite{BurgosGidea}.

\begin{corollary} \label{main corollary} The equations of motion given by (\ref{hillhamiltonian}) are equivalent, via a rotation, to the system
\begin{equation}\begin{split}\label{finalequations}
\ddot{x}-2\dot{y}&=\Omega_{x},\\
\ddot{y}+2\dot{x}&=\Omega_{y},\\
\ddot{z}&=\Omega_{z},\end{split}
\end{equation}
where
\begin{equation}\label{rotatedeffective}
\Omega=\frac{1}{2}(\lambda_{2}x^2+\lambda_{1}y^2-z^2)+
\frac{1}{\sqrt{x^2+y^2+z^2}},
\end{equation}
and $\lambda_{2}=\frac{3}{2}(1+d)$, $\lambda_{1}=\frac{3}{2}(1-d)$, $d=\sqrt{1-3\mu+3\mu^2}$.
\end{corollary}

In the work \cite{BurgosGidea} it was established that the system (\ref{finalequations}) is symmetric with respect to $\mu = 1/2$, so thereafter we assume $\mu \in [0,1/2]$.
\newline

We realize that system (\ref{finalequations}) is a particular case of a Keplerian binary system under the influence of an external gravitational perturbation of linear tidal type. The equations of motion look like \cite{Chicone}
\begin{eqnarray*}
\ddot{x}-\omega\dot{y}-\frac{\omega^2}{4}x+\frac{k_0x}{r^3}&=&-(k_{11}x+k_{12}y),\\
\ddot{y}+\omega\dot{x}-\frac{\omega^2}{4}y+\frac{k_0y}{r^3}&=&-(k_{21}x+k_{22}y),\\
\ddot{z}+\frac{k_0z}{r^3}&=&-k_{33}z,
\end{eqnarray*}
where $r = \sqrt{x^2+y^2+z^2}$, $\omega$ is a frequency associated to the rotating system, and $k_0$, $k_{11}$, $k_{12}$, $k_{21}$, $k_{22}$ are real parameters. In the case of the perturbation of a normally incident circularly polarized gravitational wave $k_{12}=k_{21}=0$, $k_{11}=4\epsilon\alpha$, $k_{22}=-4\epsilon\alpha$ and $k_{33}=-2\xi_0$, where $\epsilon$ is the perturbing parameter and $\alpha$ and $\xi_0$ are real constants related to the gravitational wave. The H4BP has the same structure as the former problem with $\Omega=2$, $k_0=1$, $k_{12}=k_{21}=0$, $k_{11}=1-\lambda_2$,  $k_{22}=1-\lambda_1$ and $k_{33}=1$.
\newline

The Jacobi constant C (\ref{jacobiintegral}) also holds for (\ref{finalequations}) with $\Omega$ given in (\ref{rotatedeffective}). It is not difficult to see that $\Omega_x$, $\Omega_y$ and $\Omega_z$, fulfill the conditions
\begin{equation*}
\begin{split}
&\Omega_{x}(x,\pm y, \pm z)=  \Omega_{x}(x,y,z),\\
&\Omega_{x}(-x,\pm y, \pm z)=  -\Omega_{x}(x,y,z),\\
&\Omega_{y}(\pm x,y, \pm z)=  \Omega_{y}(x,y,z),\\
&\Omega_{y}(\pm x,-y, \pm z)=  -\Omega_{y}(x,y,z),\\
&\Omega_{z}(\pm x,\pm y,z)= \Omega_{z}(x,y,z),\\
&\Omega_{z}(\pm x,\pm y,-z)= -\Omega_{z}(x,y,z),
\end{split}
\end{equation*}
for any combination of the sign in its variables.

The above properties allow the existence of symmetries and reversing symmetries for the equations of motion. Such transformations will be made explicit in Section \ref{sec:reversing}.
\newline

The equations (\ref{finalequations}) can be derived from the following Hamiltonian, which can be obtained rewritting (\ref{hillhamiltonian}) after the rotation given in Corollary \ref{main corollary}, that is
\begin{equation}\begin{split}
H(x,y,z,p_x,p_y,p_z)=&\frac{1}{2}(p_x^2+p_y^2+p_z^2)+yp_x-xp_y+ax^2+by^2+cz^2\\&-\frac{1}{\sqrt{x^2+y^2+z^2}},
\label{rotatedspatialhamiltonian}
\end{split}
\end{equation}
where $a=(1-\lambda_{2})/2$, $b=(1-\lambda_{1})/2$ and $c=1/2$.
\newline

\section{Reversing symmetries}\label{sec:reversing}

The symmetries of the equations of motions have been used successfully for studying periodic orbits of dynamical systems. Birkhoff \cite{Birkhoff} studied the symmetric orbits of the circular restricted three-body problem using symmetry lines \cite{Vogelaere}, which consist in the factorization of the Poincar\'e map as the product of two involutions. This technique has been used in the three-body problem with an inverse square law potential \cite{Jimenez2003}, the  St{\"o}rmer problem \cite{Jimenez1990}, and the motion of rigid body \cite{Chavoya1989}, among others.
\newline

The concept of symmetry line is closely related to reversing symmetry \cite{Lamb}. It has been used for studying periodic orbits of conservative and Hamiltonian systems \cite{Munoz1, Munoz2, Vanderbauwhede}, particularly $N$--body problems \cite{Bengochea1, Bengochea2}. In the following we introduce some concepts and useful results.
\newline

Let us consider the equations of motion in the standard form
\begin{equation} \label{standard}
\frac{d{\bf u}}{dt} = F({\bf u}),
\end{equation}
where ${\bf u} = (x,y,z,\dot{x},\dot{y},\dot{z})^T$, and $F$ is the field of the system, that is
$$
F(x,y,z,\dot{x},\dot{y},\dot{z}) = (\dot{x},\dot{y},\dot{z},2 \dot{y} + \Omega_x,-2 \dot{x} + \Omega_y, \Omega_z)^T.
$$
Here ${\bf u} \in D = {\mathbb{R}^3 \setminus \{(0,0,0) \} \times \mathbb{R}^3}$.\\

A symmetry is an involution $S: D \to D$, $(S^2 = id)$, in such a way
$$
\frac{d S ({\bf u})}{dt} = F \circ S ({\bf u})
$$
holds. In a similar way, a reversing symmetry is an involution $R: D \to D$, such that
$$
\frac{d R ({\bf u})}{dt} = -F \circ R ({\bf u}).
$$

\begin{remark}
The composition of two symmetries, or two reversing symmetries, lead to a symmetry, and the composition of a symmetry and a reversing symmetry give rise to a reversing symmetry.
\label{composition}
\end{remark}

Now, we introduce a Theorem \cite{Munoz1,Vanderbauwhede} that states the relationship between reversing symmetries and symmetric periodic orbits.
\begin{theorem}
Let ${\bf u}(t)$ be a solution defined for $t \in [0,T_0]$, which passes through fixed points of the reversing symmetries $R$ and $\widetilde R$ at times $t=0$ and $T_0$, respectively. Therefore, for all $t \in \mathbb{R}$, $m \in \mathbb{Z}$ we have that
$$
\begin{array}{c}
{\bf u}(-t) = R {\bf u}(t), \\
{\bf u}(t) = \widetilde R {\bf u}(2T_0 -t), \\
{\bf u}(2mT_0 + t) = (\widetilde R \circ R)^m {\bf u}(t).
\end{array}
$$
\label{Munoz}
\end{theorem}

\begin{remark}
Notice that if $(\widetilde R \circ R)^M = id$ for some $M \in \mathbb{N}$ then ${\bf u}(t)$ is periodic, with period $T= 2MT_0$.
\label{remark2}
\end{remark}

\begin{remark}
Let us assume the existence of a solution ${\bf u}(t)$ which passes through fixed points of reversing symmetries $R$, $\widetilde R$ in such a way $(\widetilde R \circ R)^M = id$ for some $M \in \mathbb{N}$ (the orbit is periodic). Moreover, suppose that whole periodic solution does not intersect any fixed point of a given symmetry (o reversing symmetry) $P$, then $P {\bf u}(0)$ defines an initial condition of a periodic orbit different from the original.
\label{newperiodic}
\end{remark}

The Hill's problem admits the symmetries
\begin{equation}
\begin{array}{l}
S_1(x,y,z,\dot{x},\dot{y},\dot{z}) = (-x,-y,z,-\dot{x},-\dot{y},\dot{z}), \\ [6pt]
S_2(x,y,z,\dot{x},\dot{y},\dot{z}) = (x,y,-z,\dot{x},\dot{y},-\dot{z}),
\end{array}
\label{symmetries}
\end{equation}
and reversing symmetries

\begin{equation}
\begin{array}{l}
R_1(x,y,z,\dot{x},\dot{y},\dot{z}) = (x,-y,-z,-\dot{x},\dot{y},\dot{z}), \\ [6pt]
R_2(x,y,z,\dot{x},\dot{y},\dot{z}) = (-x,y,-z,\dot{x},-\dot{y},\dot{z}), \\ [6pt]
R_3(x,y,z,\dot{x},\dot{y},\dot{z}) = (x,-y,z,-\dot{x},\dot{y},-\dot{z}), \\ [6pt]
R_4(x,y,z,\dot{x},\dot{y},\dot{z}) = (-x,y,z,\dot{x},-\dot{y},-\dot{z}).
\end{array}
\label{reversors}
\end{equation}

The reversing symmetries $R_2$, $R_3$ and $R_4$, given in (\ref{reversors}), correspond to the compositions $S_1 \circ R_1$, $S_2 \circ R_1$ and $S_1 \circ S_2 \circ R_1$, respectively. There are no other reversing symmetries that can be obtained from all possible combinations of (\ref{symmetries}) and (\ref{reversors}) (see Remark \ref{composition}).
\newline

In order to apply Theorem \ref{Munoz}, we observe that
\begin{equation}
\begin{array}{l}
R_i \circ R_i = id, i=1,\cdots,4, \\ [6pt]
(R_i \circ R_j)^2 = id, \enskip i, j \in \{1,\cdots,4 \}, \enskip i \not = j.
\end{array}
\label{tableofm}
\end{equation}
The first row of (\ref{tableofm}) is clear since $R_i$, $i=1,\cdots,4$ are involutions, and the second row can be verified easily by a straightforward computation. With this we have that if ${\bf u}(t)$ intersects two fixed points of the same reversing symmetry at times $t=0$ and $t=T_0$ then ${\bf u}(t)$ is periodic with period $T=2T_0$. On the other hand, if the solution passes through fixed points of different reversing symmetries at times $t = 0$ and $t = T_0$ then the solution is periodic, with period $T = 4T_0$.
\newline

In Section \ref{classification} we classify the periodic orbits according to its symmetry, which will be required for a better understanding of the numerical results presented in Section \ref{sec:numericalresults}.

\subsection{Classification of symmetric periodic orbits}\label{classification}

In this Section we describe  some types of symmetric periodic orbits that we have computed numerically; all of them pass through two fixed points of reversing symmetries at times $t=0$ and $t=T_0$ (in some cases the corresponding reversing symmetries at $t=0$, $t=T_0$ are equal to each other). In the following we use ${\bf u}(0) =(x_0,y_0,z_0,\dot{x}_0,\dot{y}_0,\dot{z}_0)$ and ${\bf u}(T_0) =(x_1,y_1,z_1,\dot{x}_1,\dot{y}_1,\dot{z}_1)$ for representing the state vector of an arbitrary orbit at times $t=0$ and $t=T_0$, respectively.
\newline

The classification of the orbits is as follows. The type \textbf{I} is characterized by solutions ${\bf u}(t)$ which pass through fixed points of $R_1$, we mean ${\bf u}(0) =(x_0,0,0,0,\dot{y}_0,\dot{z}_0)$ and ${\bf u}(T_0) =(x_1,0,0,0,\dot{y}_1,\dot{z}_1)$. Since $R_1$ is an involution, these kind of orbits are $T=2T_0$ periodic, and symmetric with respect to, hereafter written as wrt, $x-$axis. For the type \textbf{II} we have that at $t=0$ and $T_0$ the solutions pass through fixed points of $R_3$, we mean ${\bf u}(0) = (x_0,0,z_0,0,\dot{y}_0,0)$ and ${\bf u}(T_0)=(x_1,0,z_1,0,\dot{y}_1,0)$. In this case the orbits are also $T=2T_0$ periodic, and symmetric wrt $xz-$plane. In the type \textbf{III} the solutions ${\bf u}(t)$ pass through fixed points of $R_1$ and $R_3$ at times $t=0$ and $t=T_0$, respectively, we mean ${\bf u}(0)=(x_0,0,0,0,\dot{y}_0,\dot{z}_0)$ and ${\bf u}(T_0)=(x_1,0,z_1,0,\dot{y}_1,0)$, therefore these orbits are $T = 4T_0$ periodic and symmetric with respect to both $x-$axis and $xz-$plane. Finally, the symmetric orbits of type \textbf{IV} pass  through fixed points of $R_1$ and $R_4$ at times $t=0$ and $t=T_0$, respectively, therefore ${\bf u}(0)=(x_0,0,0,0,\dot{y}_0,\dot{z}_0)$ and ${\bf u}(T_0)=(0,y_1,z_1,\dot{x}_1,0,0)$. In this case the orbits are also $T = 4T_0$ periodic, and symmetric wrt to both $x-$axis and $yz-$plane.
\newline

All computed periodic orbits in this work belong to one of the types \textbf{I}-\textbf{IV} described above. The results of the numerical study of symmetric periodic orbits are shown in Section \ref{sec:numericalresults}.

\section{Some invariant sets and dynamical properties}\label{sec:invariantsets}

One of the advantages of considering the H4BP as a limit case of the R4BP is that the computation of the positions of the equilibrium points can be performed in terms of closed formulas depending on the mass parameter $\mu$. We recall that $\mu$ represents the relative mass of the second distant body, namely $m_{2}$. Therefore, the evolution of the position and the stability of the equilibrium points can be performed by using closed expressions that provide exact information for all the admissible values of $\mu$. An opposite situation happens for the R4BP, where the lack of closed-form expressions for the positions of the equilibrium points in terms of the masses of the system, except for very few particular cases, represents an obstruction for a general analytic study of the stability. It is worth mentioning that nowadays it is possible to develop numerical studies with high order precision to be used for \textit{a posteriori} analysis in order to provide \textit{computer assisted proofs} of the existence of equilibrium points, periodic orbits and invariant manifolds. The interested reader can consult \cite{Kepley}, \cite{KepleyII}, \cite{BurgosLessardJames} and references therein for more details. Nevertheless, the H4BP allows us to obtain general formulas for the equilibrium points and their stability could be analized by pen and paper techniques, as a consequence, we will privilege this analysis in the following subsection.
\newline

\subsection{Equilibrium points and stability}\label{equilibriumpoints}

In \cite{BurgosGidea} was shown that the H4BP has four planar equilibrium points given by the expressions
\begin{eqnarray*}
L_{1}=\left(\frac{1}{\sqrt[3]{\lambda_{2}}},0,0\right),
L_{2}=\left(-\frac{1}{\sqrt[3]{\lambda_{2}}},0,0\right),
L_{3}=\left(0,\frac{1}{\sqrt[3]{\lambda_{1}}},0\right),
L_{4}=\left(0,-\frac{1}{\sqrt[3]{\lambda_{1}}},0\right).
\end{eqnarray*}

Since the equilibrium points are planar, the characteristic polynomial $P(\eta)$, that provides the information on the linear stability of the equilibrium point in $\mathbb{R}^3$, can be written as

\begin{equation*}
\label{spatialcharpoly}
P(\eta)=P_{spatial}(\eta)P_{planar}(\eta)
\end{equation*}
where
\[
P_{spatial}(\eta):=\eta^2-\Omega_{zz}.
\]
In general, the stability for the planar case is defined by the polynomial
\[
P_{planar}(\eta):=\eta^4-(\Omega_{yy}+\Omega_{xx}-4)\eta^2+\Omega_{xx}\Omega_{yy}-\Omega_{xy}^2.
\]
It is immediate to see that the roots of $P_{spatial}(\eta)$ are given by $\sqrt{\Omega_{zz}}$ and $-\sqrt{\Omega_{zz}}$. Since the equilibrium points are symmetric, it will be enough to consider the evaluation at the points $L_1$ and $L_3$, namely
\[
\Omega_{zz}(L_1)=-(1+\lambda_2), \enskip \Omega_{zz}(L_3)=-(1+\lambda_1).
\]

As we know $\lambda_1>0$,  $\lambda_2>0$, so we obtain two purely imaginary eigenvalues for each equilibrium point, i.e., for the point $L_1$ corresponds $\{\textit{i}\omega_{z1},-\textit{i}\omega_{z1}\}$ and for the point $L_3$ corresponds $\{\textit{i}\omega_{z3},-\textit{i}\omega_{z3}\}$ where $\omega_{z1}=\sqrt{1+\lambda_2}$, $\omega_{z3}=\sqrt{1+\lambda_1}$.

The linear stability on the $xy-$plane, for the equilibrium points $L_{1}$ and $L_{2}$, is characterized by the following proposition \cite{BurgosGidea}.

\begin{proposition} The equilibrium points $L_{1}$ and $L_{2}$ are unstable for $\mu\in[0,1/2]$, with the eigenvalues of the linearized system being of the form $\pm\Lambda$ and $\pm\textit{i}\omega$ with $\Lambda>0$ and  $\omega>0$.
\end{proposition}

It is possible to obtain the value of $\omega$ in terms of $\lambda_1$ and $ \lambda_2$,  therefore in terms of $\mu$. Such value is given by the expression

\begin{equation*}
\label{frequenciyforL1}
\omega=\frac{\sqrt{A+\sqrt{D}}}{\sqrt{2}},
\end{equation*}
with $A=4-2\lambda_2-\lambda_1$, $B=3\lambda_2(\lambda_1-\lambda_2)$ and $D=A^2-4B$. On the other hand, the linear stability on the $xy-$plane for the equilibria $L_{3}$ and $L_{4}$ is described by the following proposition \cite{BurgosGidea}.

\begin{proposition} The equilibrium points $L_{3}$ and $L_{4}$ have the following properties: there exists a value $\mu_{0}$ such that for $\mu\in(0,\mu_{0})$, the eigenvalues of the linearized system are of the form $\pm\textit{i}\omega_{1}$ and $\pm\textit{i}\omega_{2}$, for $\mu=\mu_{0}$ we have a pair of the eigenvalues $\pm\textit{i}\omega$ of multiplicity 2, and  for $\mu\in(\mu_{0},1/2]$ the eigenvalues are of the form
$\pm\alpha\pm\textit{i}\omega$ with $\alpha>0$ and $\omega>0$.
\end{proposition}

The exact value of $\mu$ where the change of linear stability occurs is
\[\mu_{0}=\frac{1}{450}\left(225-\sqrt{3(5227+2368\sqrt{21})}\right),\]
which is approximately $\mu_{0}\approx 0.011942$.  In the Solar System there exists several equilateral configurations being Sun-Jupiter-Asteroids the most famous, in this case $\mu = 0.00095$ associated to Jupiter-Sun mass-ratio. The corresponding equilibrium points $L_{3}$ and $L_{4}$ are linearly stable, as happens in other cases with potential applications.
\newline

\subsection{Resonant periodic orbits}\label{resonances}

Another advantage of Hill's approach is that we can study the resonances between frequencies associated to the eigenvalues of the equilibrium points. If $L_{3}$ and $L_{4}$ are linearly stable we can consider the ratio $
\omega_{2}/\omega_{1}=k$ with $k\in\mathbb{Z}$ to find a sequence of values of $\mu$ where such ratio it is satisfied. In \cite{Burgos} are reported the values

\begin{equation}
\label{resonancesequation}
\begin{split}
\mu_k=&\frac{1}{2}-\frac{1}{6\sqrt{3}}\sqrt{\frac{5227+1184r-5K^2-32Kr-38K}{(K-25)^2}}
\end{split}
\end{equation}
with $r=\sqrt{84-3K}$ and $K=(k^2-1)^2/(k^2+1)^2$.  For instance, the exact values for $k=2,3$ are respectively

\begin{eqnarray*}
\mu_2&=\frac{1}{2}-\frac{1}{462}\sqrt{\frac{5}{3}(10181+458\sqrt{2073})},\\
\mu_3&=\frac{1}{2}-\frac{1}{1218}\sqrt{5(24077+6464\sqrt{57})}.
\end{eqnarray*}

In Table \ref{table:nonlin} we provide some values of $\mu_k$ with 6 decimals, however (\ref{resonancesequation}) provides exact results that can be taken with any order of precision.

\begin{table}[ht]
\caption{Resonances between frequencies.} 
\centering 
\begin{tabular}{c c c c}
\hline\hline 
$k$ & $\mu$ & $k$ & $\mu$ \\ [0.5ex] 
\hline 
1 & $\mu_{0}$ & 6 & 0.001293 \\ 
2 & 0.007733 & 7 & 0.000965 \\
3 & 0.004390 & 8 & 0.000746 \\
4 & 0.002713 & 9 & 0.000594 \\
5 & 0.001817 & 10 & 0.000483 \\ [1ex] 
\hline 
\end{tabular}
\label{table:nonlin} 
\end{table}
In Table \ref{table:nonlin} we can observe that the sequence $\{\mu_k\}$ is decreasing. It is not difficult to see that if $k\rightarrow\infty$ then $K\rightarrow1$, $r\rightarrow9$ and $\mu_k\rightarrow0$.
\newline

For the equilibrium points $L_{1}$ and $L_{2}$ we can also study the resonances among the vertical and planar frequencies in order to determine periodic motions in the linear approximation. Following  the notation of chapter $10$ of \cite{Sz}, the general solution of the linear equations of motion around a saddle$\times$center$\times$center equilibrium point is

\begin{eqnarray}\label{linearmotionL1}
\begin{split}
\xi&=\zeta_0\cos(\omega t)+\frac{\eta_0}{\alpha}\sin(\omega t),\\
\eta&=\eta_0\cos(\omega t)+\alpha\xi_0\sin(\omega t),\\
\zeta&=C_1\cos(\omega_{z1} t)+C_2\sin(\omega_{z1} t),
\end{split}
\end{eqnarray}
where $\zeta_0$, $\eta_0$, $\alpha$, $\xi_0$, $C_1$ and $C_2$ are constant defined by initial conditions and the second derivative $\Omega_{zz}$. The periodic solutions can be classified in three families: the class (a) corresponds to one-dimensional vertical Liapunov orbits for $\eta=\xi=0$, the class (b) corresponds to planar Liapunov orbits for $\zeta=0$, and the class (c) corresponds to spatial Lissajous orbits that satisfy

\begin{eqnarray*}
\begin{split}
\xi&=\zeta_0\cos(\omega t)+\frac{\eta_0}{\alpha}\sin(\omega t),\\
\eta&=\eta_0\cos(\omega t)+\alpha\xi_0\sin(\omega t),\\
\zeta&=C_1\cos\left(\frac{q}{p}\omega t\right)+C_2\sin\left(\frac{q}{p}\omega t\right),
\end{split}
\end{eqnarray*}
where $q/p$ is a rational number so that $q/p=\omega_{z1}/\omega\leq1$. In the case of the H3BP we have that $\omega_{z1}/\omega=2/\sqrt{2\sqrt{7}-1}\approx1$. Nevertheless, with the aim of completing the study of resonances among the equilibrium points in the H4BP, to be used in further analytical studies, we can examine the behaviour of the ratio $q/p$ as function of $\mu$
\newline

We start recalling that both quantities $\lambda_1$ and $\lambda_2$ can be written as $\lambda_2=(3/2)(1-d)$ and $\lambda_1=(3/2)(1+d)$, where $d=\sqrt{1-3\mu+3\mu^2} \in [1/2,1]$ is a decreasing function for $\mu \in [0,1/2]$. So we can write the frequencies $\omega$ and $\omega_{z1}$ in terms of $d$. A straightforward computation shows that the equation $q/p = \omega_{z1}/\omega$ can be written as

\begin{equation*}
\label{r in terms of d}
r(d):=\frac{q}{p}=\frac{\sqrt{6d+10}}{\sqrt{-1-3d+\sqrt{1+222d+225d^{2}}}}\, .
\end{equation*}

The next result shows that this function is indeed a bijection.

\begin{proposition}\label{proprd}
   For all $\mu\in[0,1/2]$, the function $r:[1/2,1]\longrightarrow[r(1),r(1/2)]$ is a bijection.
\end{proposition}

\begin{proof}
   Let us write the derivative
   \[ r'(d)=\frac{6\sqrt{2}(-2(23+33d)+\sqrt{1+222d+225d^{2}})}{\sqrt{5+3d}\sqrt{1+222d+225d^{2}}(-1-3d+\sqrt{1+222d+225d^{2}})^{3/2}}
   \]
   and realize that the numerator satisfies $-2(23+33d)+\sqrt{1+222d+225d^{2}}<0$ for all $d\in[1/2,1]$, that is, the function $r(d)$ is monotically decreasing, with \[ r(1)=\frac{2}{\sqrt{2\sqrt{7}-1}}\approx0.96544\,.\]
\end{proof}

Proposition \ref{proprd} ensures the existence of a unique value $\bar{d}$ such that $r(\bar{d})=1$ which corresponds to $\bar{d}=5/6$. So, we can state that  $q/p\in[r(1),1] \approx [0.96544,1]$. In other words, the ratio stays \textit{close} to $1:1$ resonance.
\newline


Since $r(d)$ is a bijection, we can even compute the range of values of $\mu$ for which $q/p\in[0.96544,1]$. The  proof of the following Proposition is an easy algebraic manipulation of two-second degree algebraic polynomials so will be omitted.

\begin{proposition}\label{proprd2} The range of values of $\mu$ for the existence of local spatial periodic motion of class (c) around the equilibrium point $L_1$ is $\mu\in[0,\mu_p]$ with
\[
\mu_p=\frac{1}{18}\left(9-4\sqrt{3}\right),
\]
which is approximately $\mu_p\approx0.1151$.
\end{proposition}

The above study provides a first insight of the center manifold associated to the equilibrium points $L_1$ and $L_3$. However, it is well known that higher order terms will produce non-linear coupling effects between planar and spatial motions. For instance, when the frequencies in the center manifold become equal to each other we obtain the so called Halo orbits. Therefore, Propositions \ref{proprd} and \ref{proprd2} suggest that there exist Halo orbits in the interval $[0,\mu_p]$ since the ratio $\omega_{z1}/\omega\leq1$. Of course, a further analytical study, including higher-order terms, is necessary in order to prove this conjecture. Nevertheless, the numerical explorations given in subsection \ref{haloorbits} support such conjecture.\\

{\bf Stability of periodic orbits}. For completeness, we wish to recall some basics aspects regarding the stability of periodic orbits. The set of eigenvalues associated to a  periodic orbit in a Hamiltonian system with 3 degrees of freedom looks like $\Lambda=\{1, 1, \lambda_1, \lambda^{-1}_1,\lambda_2, \lambda^{-1}_2\}$, where the eigenvalues $\lambda_1$ and $\lambda_2$ can be computed as the roots of the characteristic polynomial

\begin{equation}\label{periodicspatial charpoly}
Q(\lambda)=\lambda^4-\alpha\lambda^3+\beta\lambda^2-\alpha\lambda+1,
\end{equation}
obtained from the $4\times4$ matrix resulting by restricting the study of periodic orbits to a suitable one dimensional Poincar\'e at a fixed energy level, i.e., we restrict the study to a co-dimension 2 surface $\Sigma$. If we denote such  matrix as $\Phi_\Sigma$, the coefficients of (\ref{periodicspatial charpoly}) are given by
\[
\alpha=Tr(\Phi_\Sigma)\quad \text{and}\quad \beta=\frac{1}{2}\left(Tr(\Phi_\Sigma)^2-Tr(\Phi^2_\Sigma)\right).
\]

The study of (\ref{periodicspatial charpoly}) is simplified with the use of the auxiliary variable $s=\lambda+\lambda^{-1}$. In terms of $s$, the characteristic polynomial becomes
\[
Q(s)=s^2-\alpha s+\beta-2,
\]
whose roots are the so called \textit{stability parameters}. It is easy to see that the \textit{stability parameters} are given by the expressions
\[
s_1=\frac{1}{2}\left(\alpha+\sqrt{D}\right) \quad s_1=\frac{1}{2}\left(\alpha-\sqrt{D}\right),
\]
with $D=\alpha^2-4(\beta-2)$. Therefore, the criterion for linear stability can be stated by requiring that $D\geq0$ and $\vert s_i\vert<2$ for $i=1,2$. When one of the stability parameters satisfies $\vert s_i\vert=2$ the corresponding orbit is called \textit{critical orbit}; these orbits are relevant because they point out possible bifurcations. The interested reader in the details of the above discussion and/or in a deeper study of stability and bifurcations can consult \cite{Scheereslibro} and \cite{Zagouras} for more information.

\subsection{Invariant $z$--axis}\label{invariantz}

It is not difficult to see that the $z-$axis $\{(x,y,z)\,|\,x=y=0\,,p_{x}=p_{y}=0\}$ is an invariant set, and that the corresponding dynamics is defined by equations
\begin{eqnarray*}
     \dot{z}&=&p_{z}\,,\\
     \dot{p}_{z}&=&-z-\frac{z}{|z|^{3}}\, ,
\end{eqnarray*}
which are obtained from the Hamiltonian (\ref{rotatedspatialhamiltonian}) restricted to the vertical axis, we mean

\begin{equation}\label{hamiltonianzaxis}
   H(z,p_{z})=\frac{1}{2}p_{z}^{2}+\frac{1}{2}z^{2}-\frac{1}{|z|}\, .
\end{equation}

The corresponding phase space is shown in Figure \ref{verticalmotion}. All motions are bounded in position unlike momentum due to binary collisions and there is no dependence on the parameter $\mu$.
\newline

\begin{figure}
  \centering
  \includegraphics[width=3in,height=2.5in]{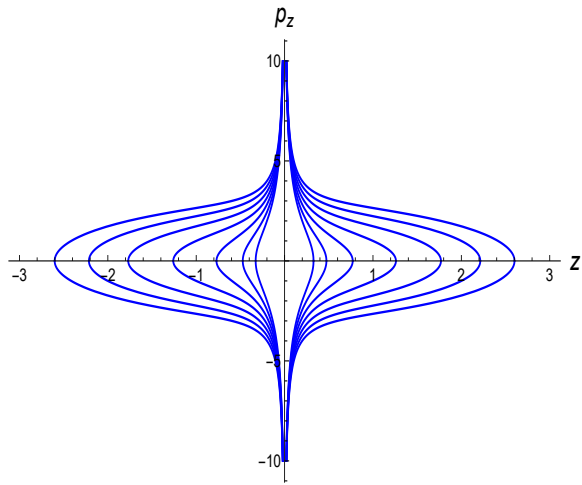}
  \caption{Phase space of the invariant $z$-axis.}\label{verticalmotion}
\end{figure}

The analysis performed in \cite{BelbrunoII} shows that, when $z>0$ in the Hamiltonian (\ref{hamiltonianzaxis}),  there exists a \textit{polar orbit} $z(t)$ for which the particle moves between $z=0$ and $z=d(h)$, where $d(h)$ is a solution of the equation $\frac{z^{2}}{2}-\frac{1}{z}=h$, with $h$ a constant energy. Since this equation has solution for any value of $h$ (equivalently for any $C$) the polar orbit exists for all values of $h$ and its amplitude depends on the energy level as it is shown in Figure \ref{verticalmotion}. It is not difficult to see that something similar happens when $z<0$ so we obtain another \textit{polar orbit} for this case. We will call \textit{northern} and \textit{southern} \textit{polar orbits} to \textit{polar orbits} with $z>0$ and $z<0$, respectively. For both kinds of polar orbits, the period is uniformly bounded by $\pi$ for any value of $h$.
\newline

The numerical explorations of the H4BP, given in Section \ref{sec:numericalresults}, suggest that the \textit{polar orbits} bifurcate with several families of periodic orbits.

\section{Automatic differentiation for the Hill four-body problem: the equivalent polynomial system}\label{sec:autoDiff}
\noindent
An efficient way to implement high-order Taylor methods for the case where the vector field is not a polynomial is by means of \textit{automatic differentiation}, a process in which we add auxiliary variables and differential equations in order to obtain a polynomial vector field  whose solutions correspond to the original non-polynomial vector field. We must mention that  the idea of replacing a given nonlinear system of differential equation with a polynomial vector field is not new, many examples  appear in the literature as early as the works \cite{Knuth} and \cite{Jorba}. The main advantage of using this method is that we can obtain high precision numerical computations which provides very good approximations for real periodic orbits and as a consequence we can use this data to perform \textit{a posteriori} analysis in order to develop rigorous proofs of the existence of some orbits of particular interest. The interested reader in this technique for proving existence of periodic orbits, and other invariant objects, may consult the works \cite{BurgosLessardJames}, \cite{Lessard}, \cite{Mireles} and references therein. The reader who is not familiar with the implementation of automatic differentiation in the Taylor method may consult \cite{BurgosRodriguez} to find an elementary introduction to this subject.
\newline

The computation of periodic orbits and their characteristic multipliers requires the numerical integration of the equations of motion and variational equations. For this purpose, we have implemented a variable step and order Taylor method controlled by tolerances $2.22045\times10^{-14}$ and machine epsilon for relative and absolute errors, respectively. It is worth mentioning that the developed Taylor-based numerical integrator is just controlled by the feasible tolerances given by the float point arithmetic at hand. If extended precision arithmetic is used then more stringent tolerances can be set into the numerical integrator, therefore the order of the method will be adapted automatically in order to satisfy such tolerances. However, in the classical double-precision arithmetic other numerical methods can be implemented to perform the computations with the tolerances used in this work, for instance, a Runge-Kutta method with step size control of orders $7$ and $8$, or a numerical integrator \textit{ode113} (based on Adams-Bashfort-Moulton methods) of \textit{Matlab} can work as well, nevertheless, sometimes these methods are time-consuming and are restricted by a maximum order so they would be insufficient in order to extend the computations to multiple precision. In the following lines we offer some details of the implementation of the  variable step and order Taylor method with automatic differentiation.
\newline

Consider the equations \eqref{finalequations}, and define the automatic
differentiation variable
\[
r(x,y,z) := \sqrt{x^2 + y^2+ z^2},
\]
so that
\[
\Omega_x = \lambda_2 x - \frac{x}{r(x,y,z)^3},\quad \Omega_y= \lambda_1 y - \frac{y}{r(x,y,z)^3},\quad \Omega_z= - z  - \frac{z}{r(x,y,z)^3}.
\]

We also introduce the variables
\[
x_1 := x, \quad
x_2 := y , \quad
x_3 := z , \quad
x_4 := \dot{x}, \quad
x_5 := \dot{y}, \quad
x_6 := \dot{z}, \quad
\]
thus the equations of motion, associated to the co-rotating frame (\ref{finalequations}), or in a equivalent way (\ref{standard}), becomes
\begin{equation}\label{eq:SH4BP}
\left(
\begin{array}{c}
\dot{x}_1 \\
\dot{x}_2 \\
\dot{x}_3 \\
\dot{x}_4 \\
\dot{x}_5 \\
\dot{x}_6
\end{array}
\right) =
\left(
\begin{array}{c}
x_4 \\
x_5 \\
x_6 \\
2 x_5 + \Omega_{x_1} \\
-2 x_4 + \Omega_{x_2} \\
 \Omega_{x_3}
\end{array}
\right).
\end{equation}
Now we append the variable
\[
x_7 = \frac{1}{r}.
\]
The meaning of the previous equation is trying to rewrite the term $1/r(x,y,z)$ in the equations. Note that
\begin{equation*}
\dot{x}_7= \frac{-1}{r^2} \dot{r} = - x_7^3(x_1x_4+x_2x_5+x_3x_6).
\end{equation*}

Therefore, in terms of the new variables, we can write
\begin{align*}
\Omega_{x_1} &=
\lambda_2x_1-x_1x_7^3,
\\
\Omega_{x_2} &=
 \lambda_1x_2-x_2x_7^3,
 \\
\Omega_{x_3} &= -x_3-x_3x_7^3.
\end{align*}

Finally, in terms of state vector ${\bf x} = (x_1,x_2,x_3,x_4,x_5,x_6,x_7)^T$, the equations of motion (\ref{eq:SH4BP}) define a polynomial vector field, we mean
\begin{equation*}
\label{augmentedvectorfield}
\dot{\textbf{x}} = g(\textbf{x})
\end{equation*}
where
\begin{equation*}
g(\textbf{x}) =
\begin{pmatrix}
x_4 \\
x_5 \\
x_6 \\
2x_5 + \lambda_2x_1-x_1x_7^3 \\
 -2 x_4 + \lambda_1x_2-x_2x_7^3\\
- x_3 - x_3x_7^3\\
 - x_7^3(x_1x_4+x_2x_5+x_3x_6)
\end{pmatrix}.
\end{equation*}

A natural question that arises from this transformation is if the periodic orbits of both systems are equivalent in the sense that a periodic orbit of the augmented system corresponds to a periodic orbit of the original system. In \cite{Kepley} we can find a study which gives a positive answer to the  former question in the phase space with restriction $x_7=1/r$, so we can work with the augmented system in order to compute periodic orbits for the original one.
\newline

Our scheme of numerical continuation is based on a classic predictor-corrector method that exploits the pseudo-arclength continuation algorithm (e.g. see \cite{Keller}). The use of this scheme requires to append the variational equations to the seven-dimensional augmented vector field. In other words, we need to consider $6$ systems of differential equations of the form

\begin{equation*}
\label{variationalequations}
\dot{\textbf{c}}_i=DF\textbf{c}_i,
\end{equation*}
where $DF$ is the Jacobian matrix of the original vector field (\ref{standard}) expressed with new variables and $\textbf{c}_i=(c_i^1,c_i^2,c_i^3,c_i^4,c_i^5,c_i^6)^T\in\mathbb{R}^6$ is a six-dimensional vector for $i=1,2,...,6$. Of course, we can form a fundamental matrix of solutions, denoted as $\Psi(t)$, by considering $\Psi(t)=col(c_1, c_2,...,c_6)$ to obtain information about the linear stability of the orbits (it will be discussed in Section \ref{sec:numericalresults}). In summary, we have developed a variable step and order integrator for a set of $43$ differential equations. The codes for the numerical integrator of the vector field of the H4BP can be found in \cite{Burgossoftware}.

\section{Numerical explorations}\label{sec:numericalresults}

In Section \ref{classification} we classified the different types of symmetric periodic orbits that appeared in our study, we mean symmetric periodic orbits of types \textbf{I-IV}. In order to compute these orbits we solve numerically specific boundary value problems (BVPs) with the use of certain {\it seeds}. The BVP depends on the type of symmetry of the periodic orbit. According to the description given in Section \ref{classification}, if we use
$$
{\bf u}(t) = (x(t),y(t),z(t),\dot{x}(t),\dot{y}(t),\dot{z}(t))
$$
as vector solution, and denote the characteristic time by $T_0$, the symmetric periodic orbits of types \textbf{I-IV} are solutions of (\ref{standard}) restricted to the following boundary conditions:

\begin{itemize}

\item type \textbf{I}
$$
\begin{array}{l}
y(0) = 0, \enskip z(0) = 0, \enskip\dot{x}(0) = 0, \\
y(T_0) = 0, \enskip z(T_0) = 0, \enskip \dot{x}(T_0) = 0, \\
\end{array}
 $$

\item type \textbf{II}
$$
\begin{array}{l}
y(0) = 0, \enskip \dot{x}(0) = 0, \enskip \dot{z}(0) = 0, \\
y(T_0) = 0, \enskip \dot{x}(T_0) = 0, \enskip \dot{z}(T_0) = 0,
\end{array}
$$

\item type \textbf{III}
$$
\begin{array}{l}
y(0) = 0, \enskip z(0) = 0, \enskip \dot{x}(0) = 0, \\
y(T_0) = 0, \enskip \dot{x}(T_0) = 0, \enskip \dot{z}(T_0) = 0.
\end{array}
$$

\item type \textbf{IV}
$$
\begin{array}{l}
y(0) = 0, \enskip z(0) = 0, \enskip \dot{x}(0) = 0, \\
x(T_0) = 0, \enskip \dot{y}(T_0) = 0, \enskip \dot{z}(T_0) = 0.
\end{array}
$$

\end{itemize}

The solution of each BPV is given by the initial condition ${\bf u}(0) = {\bf u}_0$ and characteristic time $T_0$ (for instance, for the family of type \textbf{I} we have that ${\bf u}_0 = (x_0,0,0,0,\dot{y}_0,\dot{z}_0)$). In a generic way, for a fixed value of $\mu$, the solution of each BVP depends on one parameter.  For a fixed value of $\mu$, by the cylinder theorem \cite{Meyer}, each symmetric periodic family conforms a curve in the space of initial conditions (in this case $x\dot{y}\dot{z}-$space).  On the other hand, if we also vary the value of the mass parameter $\mu$, the initial conditions of corresponding families give rise to a surface in the space of initial conditions.
\newline

The required {\it seeds} for solving the BVPs are obtained considering bifurcations among the families. For instance, in the case of three dimensional vertical orbits, Halo and Lane orbits, the {\it seeds} are obtained as initial conditions of symmetric \textit{planar critical orbits} to generate initial conditions. Nevertheless, we will see that these three families are related in an interesting way as the linear approximation suggests.  This approach was considered in the seminal work of M. H\'enon \cite{Henon1974} where the author introduced the so-called vertical stability index (denoted as $a_{v}$) to detect planar periodic orbits that can be continued to the spatial case (this index provides information about whether a planar periodic orbit belongs at the same time to a family of spatial periodic orbits).  We remark that planar periodic orbits and their stability indexes are studied in \cite{Henon} for the classical Hill's problem; it was shown that five families, named $a1v$, $a2v$, $a3v$, $g1v$ and $g'2v$, possess critical orbits. In a later work, M. Michalodimitrakis continued these orbits to the spatial case \cite{Michalodimitrakis}.
\newline

In the following we show the numerical study of symmetric periodic orbits. The corresponding initial conditions were obtained by means of solving numerically different boundary value problems (see beginning of this Section) with the techniques of Section \ref{sec:autoDiff}. We remark that with the application of an appropriate symmetry, or reversing symmetry to the computed periodic orbits (to be shown), then other periodic orbits can be obtained (see Remark \ref{newperiodic}). For instance, suppose a solution ${\bf u}(t)$ of type ${\bf II}$, we mean, it passes through fixed points of reversing symmetry $R_3$ at times $t=0$ and $t=T_0$, so is $2T_0$ periodic, then $R_1 {\bf u}(0)$ defines the initial condition of a $2T_0$ periodic orbit different from the original one.
\newline

The numerical results concerning symmetric periodic orbits around the tertiary are grouped in Section \ref{group2} where we distinguish two kind of orbits: ones coming from the classical Hill's problem and other new which do not. In a similar way, the results of the study of symmetric periodic orbits, related with the dynamics around $L_1$, and therefore around $L_2$, are shown in Section \ref{group1}. In order to get those orbits, we have developed a numerical continuation by varying $C$ for several values of $\mu$.  We should mention that we have focused on symmetric periodic orbits since these objects have shown to be a good starting point for further explorations of other families of periodic orbits with some astronomical interest.
\newline

\subsection{Periodic orbits around the tertiary}\label{group2}
In this Section we explore some families of periodic orbits around the tertiary body for a couple of values of $\mu$ with some astronomical interest. We have considered $\mu=0.00095$, related to the mass of Jupiter, and $\mu=0.00547$, that corresponds to the system $HD28185$; both values of $\mu$ satisfy Routh's criterion for equilateral configuration stability.

\subsubsection{Periodic Orbits emanating from planar direct orbits}
The families $g1v$ and $g'2v$ emanate from planar direct orbits $g$ as it was reported in \cite{Burgos} and possess symmetry of type \textbf{III}. We have found that family $g1v$ for $\mu=0.00095$ starts at a bifurcation with the planar family $g$ and all of its orbits are quite unstable until a turning point which is reached at $C=2.483432$. After this point, the orbits remain unstable until a double critical orbit appears, that is $\vert s_1\vert=\vert s_2\vert=2$, which happens at $C=4.12327$. As we keep increasing the value of $C$ a stable zone appears at $C=4.133271$ until the family ends with a bifurcation with a planar orbit at $C=4.284141$. The initial conditions for the critical orbits are shown in Table \ref{Criticalg1vcase1}. The behaviour of this family is qualitatively similar for the second case $\mu=0.00547$, nevertheless, there are more critical orbits in the family and the stable zone is reduced to the interval $C\in[4.132271,4.279181]$. The initial conditions for the critical orbits are shown in Table \ref{Criticalg1vcase2}. The shape of the orbits for this case can be seen in Figure \ref{g1vfigures}. We emphasize that the family $g1v$ is reported as unstable in the classical H4BP, therefore the main effect of a fourth perturbing body is again the appearance of stable motion.
\newline

\begin{figure}
  \centering
\begin{tabular}{cc}
  \includegraphics[width=2.0in,height=2.0in]{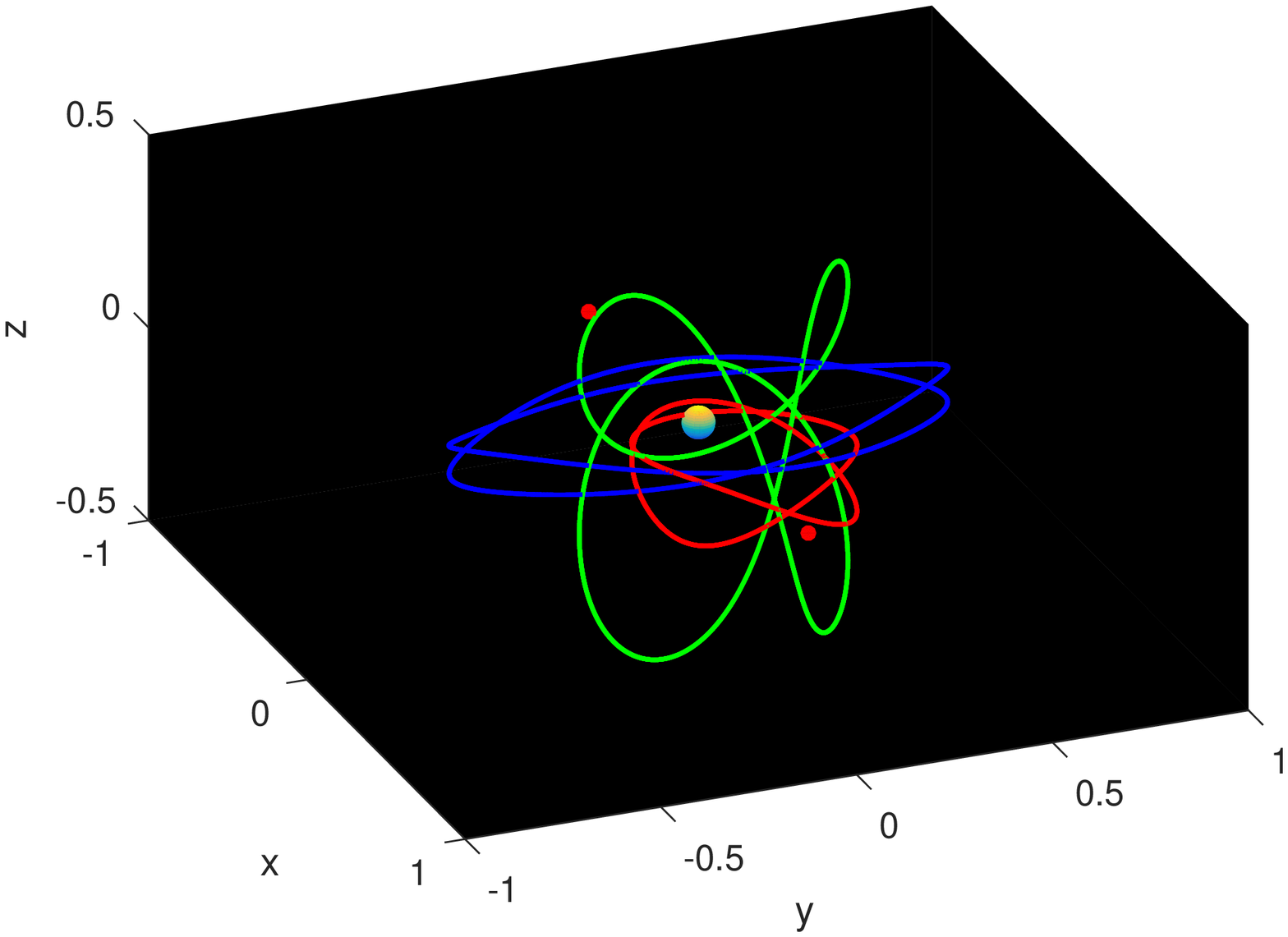}& \includegraphics[width=2.0in,height=2.0in]{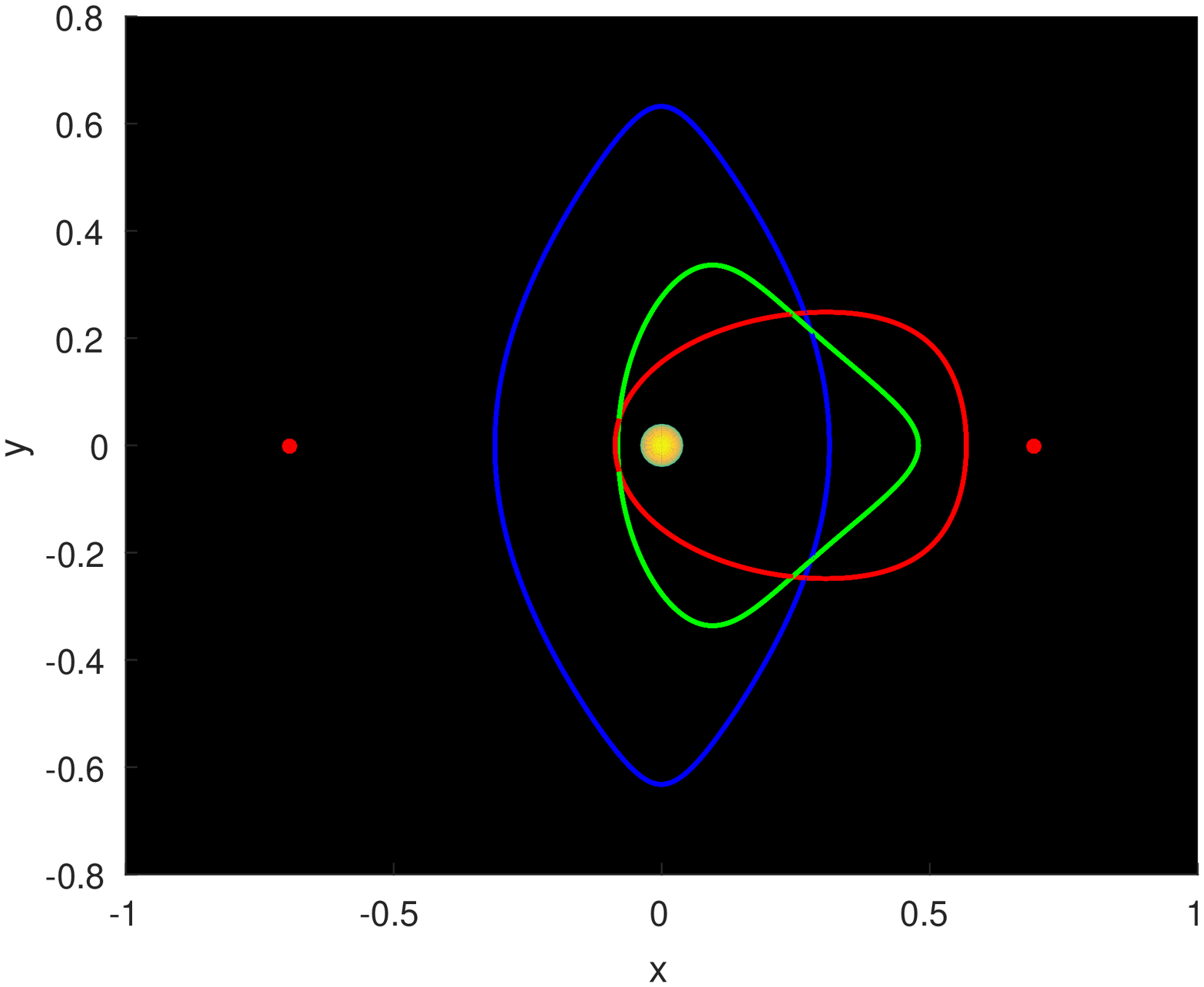}
\end{tabular}
 \caption{Periodic orbits belonging to family $g1v$ for $\mu=0.00547$ with different views. The blue orbit emanates from the planar critical orbit.} Red dots correspond to equilibrium points.\label{g1vfigures}
\end{figure}

\begin{table}[ht]
\caption{Initial conditions for critical orbits, case $\mu=0.00095$, family $g1v$.} 
\centering 
\resizebox{14cm}{!} {
\begin{tabular}{c c c c c c c}
\hline\hline 
$s_1$ & $s_2$ & $x_{0}$ & $\dot{z}_{0}$ & $\dot{y}_{0}$ & $T/4$ & $C$\\ [0.5ex] 
\hline 
2 & 23317  & 0.311044457819049  &  0.0 &   1.913203128661090  &    1.472442101072728 & 3.059642  \\
2 & 183  & 0.401582652445777 &  1.645050043578752  &   0.523578470955765 &    1.459054927378793  & 2.483432   \\
2 & 2 &  0.560062992514573 &  0.447754036062618 &   0.433143806642077 &    1.070414684724931  & 4.12327 \\
1.673 & 2 &   0.570845794094986 &  0.0 &   0.443069460233186&      1.051596721798848 & 4.284141   \\
\hline 
\end{tabular}
}
\label{Criticalg1vcase1} 
\end{table}

\begin{table}[ht]
\caption{Initial conditions for critical orbits, case $\mu=0.00547$, family $g1v$.} 
\centering 
\resizebox{14cm}{!} {
\begin{tabular}{c c c c c c c }
\hline\hline 
$s_1$ & $s_2$ & $x_{0}$ & $\dot{z}_{0}$ & $\dot{y}_{0}$ & $T/4$ & $C$\\ [0.5ex] 
\hline 
3.07 & 2  & 0.312348437026665 &  0.0 &   1.902484026495065  &    1.468067850531849 &3.075145 \\
2 &  21722  & 0.312368430078792 &  0.041340285729567  &   1.901990415167071 &    1.468078799429547   & 3.074942 \\
2 & 21629 &  0.312485459713801 &  0.108285970081361 &   1.899097496123795 &    1.468146125495763 & 3.073742 \\
2 &  21603 &   0.312514707890854 &  0.119270884090533 &   1.898374889382144&     1.468163017617618  &  3.073442  \\
 2 &   161 &   0.402791728939239 &  1.637749036679892 &   0.522580949085163 &      1.455590819188837  &2.494762  \\
  2 &   2 &   0.560171301047021 &  0.446609759343781 &   0.434884051267122 &      1.071939946966573  &4.119271 \\
     1.695 &   2 &   0.570959729608385  &  0.0&   0.444602520352016 &     1.053297958011874  &4.279181\\
\hline 
\end{tabular}
}
\label{Criticalg1vcase2} 
\end{table}

The family $g'2v$ for $\mu=0.00095$ starts from a bifurcation with a planar orbit at $C=4.279530$. As the value of $C$ is decreased monotonically, the orbits are composed of increasing vertical oscillations combined with a circulation around the $z-$axis as it is shown in Figure \ref{g2vfigures}. In this case there are not stable periodic orbits, but a pair the critical orbits appear at $C=0.035753$ and $C=2.159953$.  The behaviour is qualitatively similar when $\mu=0.00547$ in the sense that the family rises from a planar critical orbit and then it tends to periodic orbits with increasing vertical oscillations. Nevertheless, there is a zone of stability for both cases. The evolution of the stability coefficients are shown in Figure \ref{g1vandg2vstabilityfigures} and the initial conditions of the critical orbits are given in Tables \ref{Criticalg2vcase1} and \ref{Criticalg2vcase2}.

\begin{figure}
  \centering
\begin{tabular}{cc}
  \includegraphics[width=2.0in,height=2.0in]{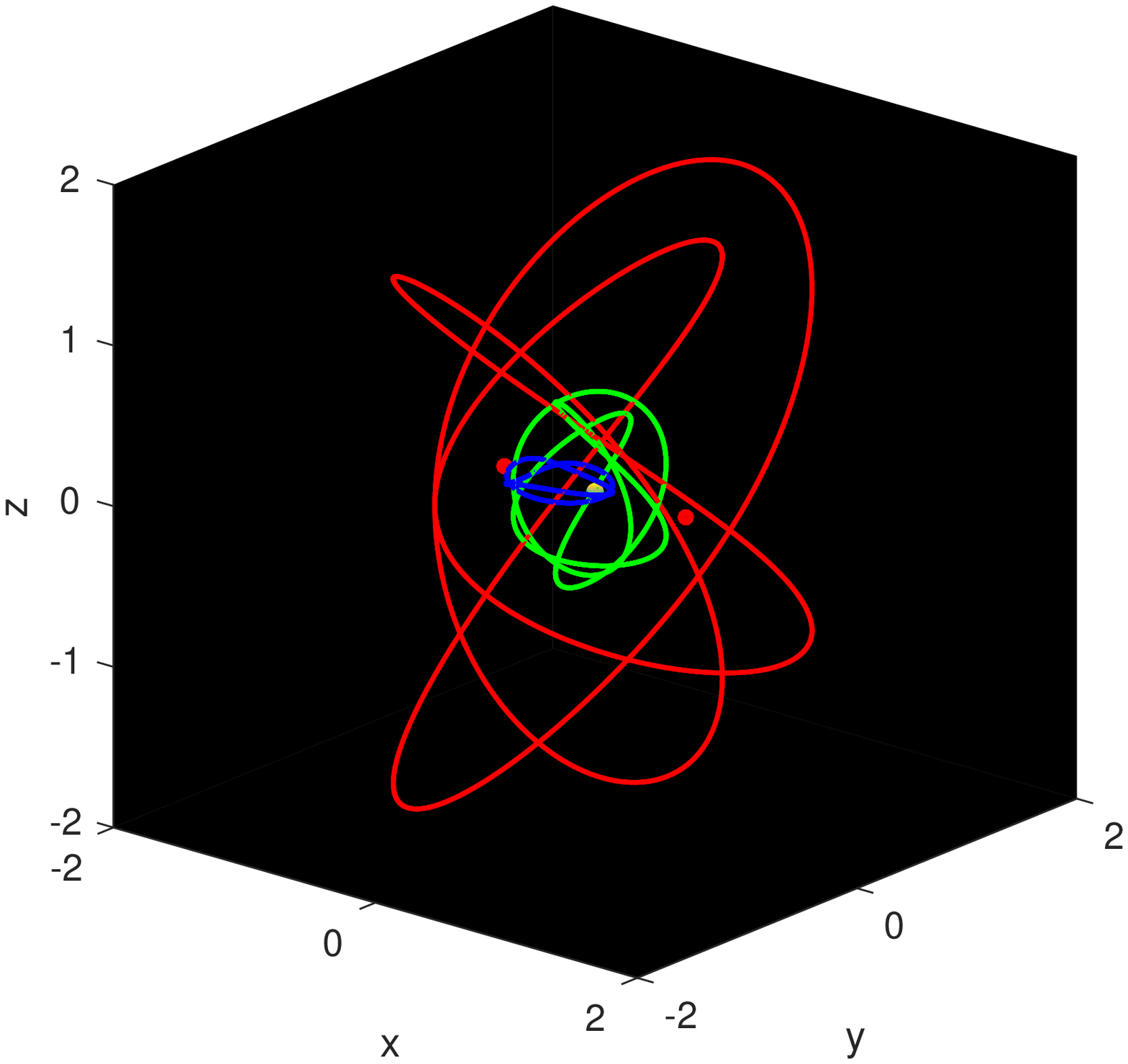}& \includegraphics[width=2.0in,height=2.0in]{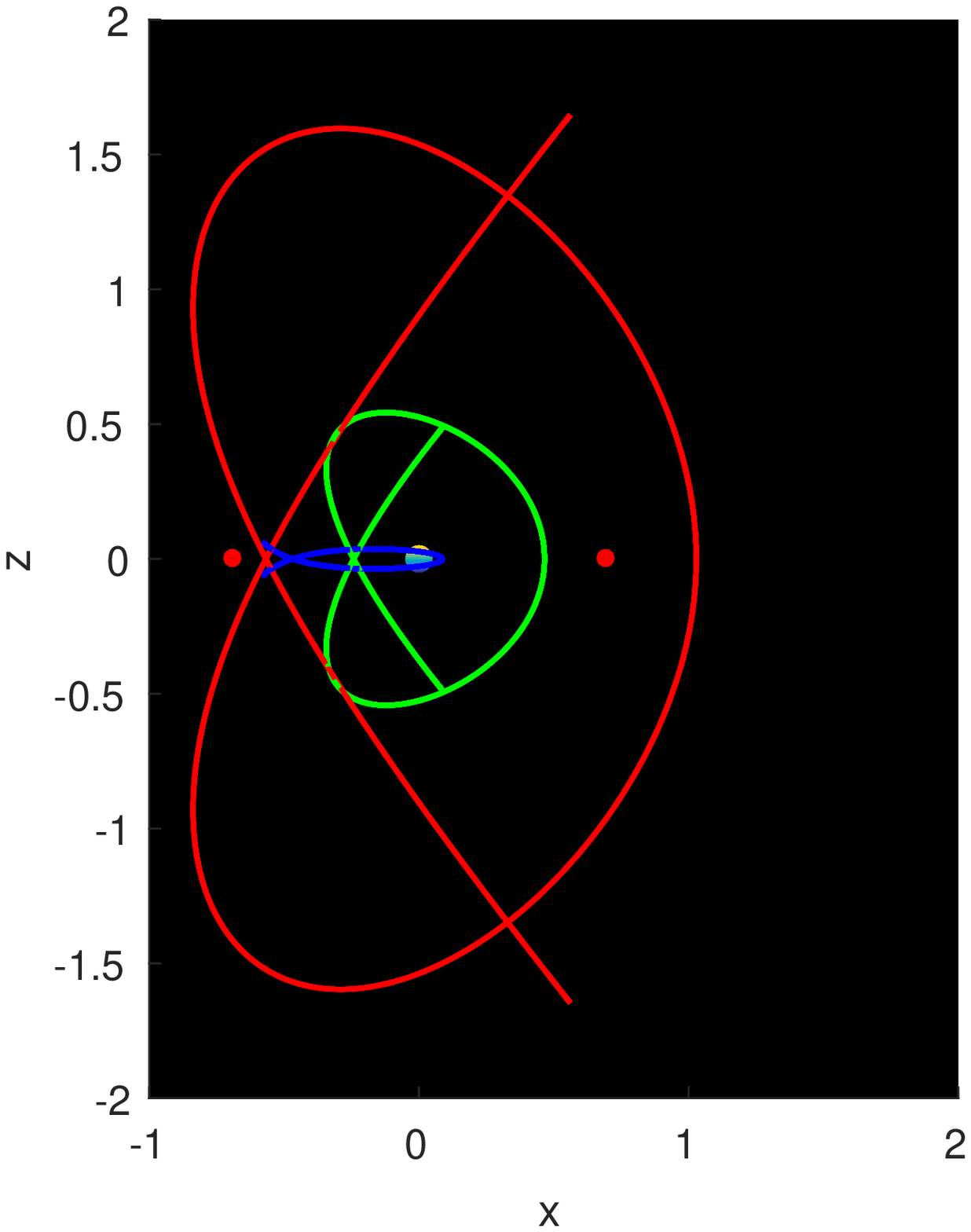}\\
  \includegraphics[width=2.0in,height=2.0in]{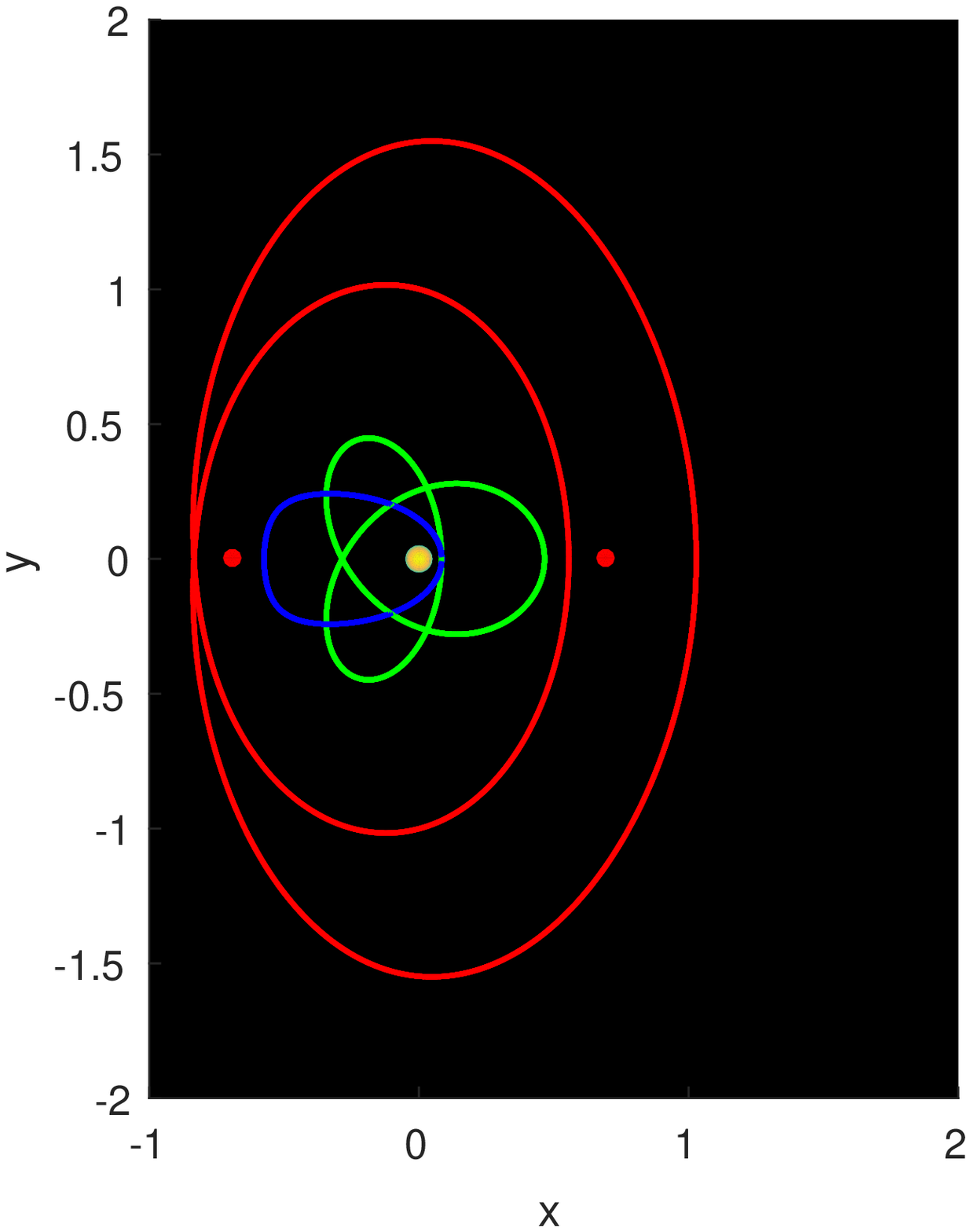}& \includegraphics[width=2.0in,height=2.0in]{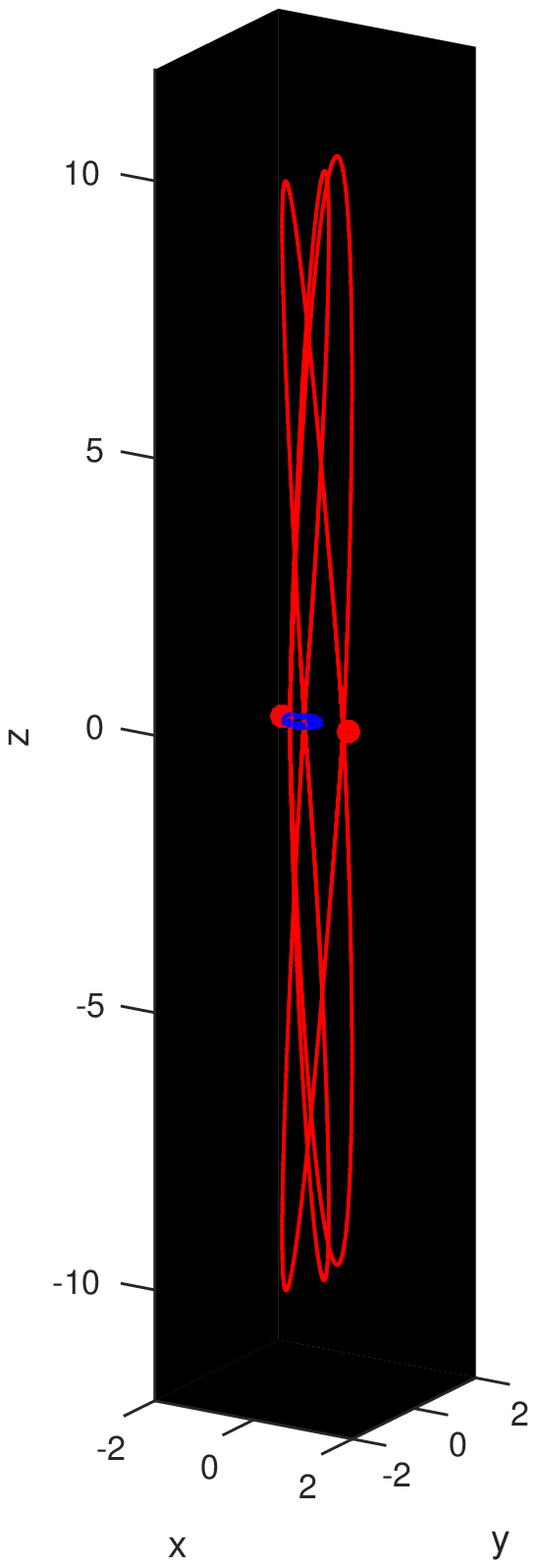}
\end{tabular}
 \caption{ Periodic orbits belonging to family $g'2v$ (top and bottom-left frames) for $\mu=0.00095$ with different views and values of $C$ starting in an orbit (blue) close to a planar critical orbit of $g'2v$. Planar critical orbit (blue) and last computed orbit for $C=-10^4$ are shown in bottom-right frame. Red dots correspond to equilibrium points.}\label{g2vfigures}
\end{figure}

\begin{table}[ht]
\caption{Initial conditions for critical orbits belonging to family $g'2v$ for $\mu=0.00095$.} 
\centering 
\resizebox{14cm}{!} {
\begin{tabular}{c c c c c c c}
\hline\hline 
$s_1$ & $s_2$ & $x_{0}$ & $\dot{z}_{0}$ & $\dot{y}_{0}$ & $T/4$ & $C$\\ [0.5ex] 
\hline 
1.189 & 2  & 0.082382785420325 &  0.0 &    4.474112396195602  &    1.083381688739804 & 4.279581\\
2 & 43.04  & 0.363286428454016  &  1.904537068852055 &    -0.337242456915256  &    1.518547036633827 & 2.159953 \\
2 & 2.023  & 0.672579519372136 &  0.547559528845478 &     -1.998543113394389 &     2.510515447099607 & 0.035753\\
\hline 
\end{tabular}
}
\label{Criticalg2vcase1} 
\end{table}

\begin{table}[ht]
\caption{Initial conditions for critical orbits belonging to family $g'2v$ for $\mu=0.00547$} 
\centering 
\resizebox{14cm}{!} {
\begin{tabular}{c c c c c c c}
\hline\hline 
$s_1$ & $s_2$ & $x_{0}$ & $\dot{z}_{0}$ & $\dot{y}_{0}$ & $T/4$ & $C$\\ [0.5ex] 
\hline 
1.156 & 2  & 0.082169168742817 &  0.0 &  4.481802309476099 &    1.092034641406013 & 4.273650\\
2 & 41.72  & 0.363167874041121  &  1.904782813283298  &    -0.329994220045412  &    1.517987521371478 & 2.164053\\
2 & 2.025  & 0.669543395827101 &  0.560226165381085 &     -1.989110192397709 &     2.492778188972285 & 0.056053\\
\hline 
\end{tabular}
}
\label{Criticalg2vcase2} 
\end{table}

\begin{figure}
  \centering
\begin{tabular}{cc}
  \includegraphics[width=2.0in,height=2.0in]{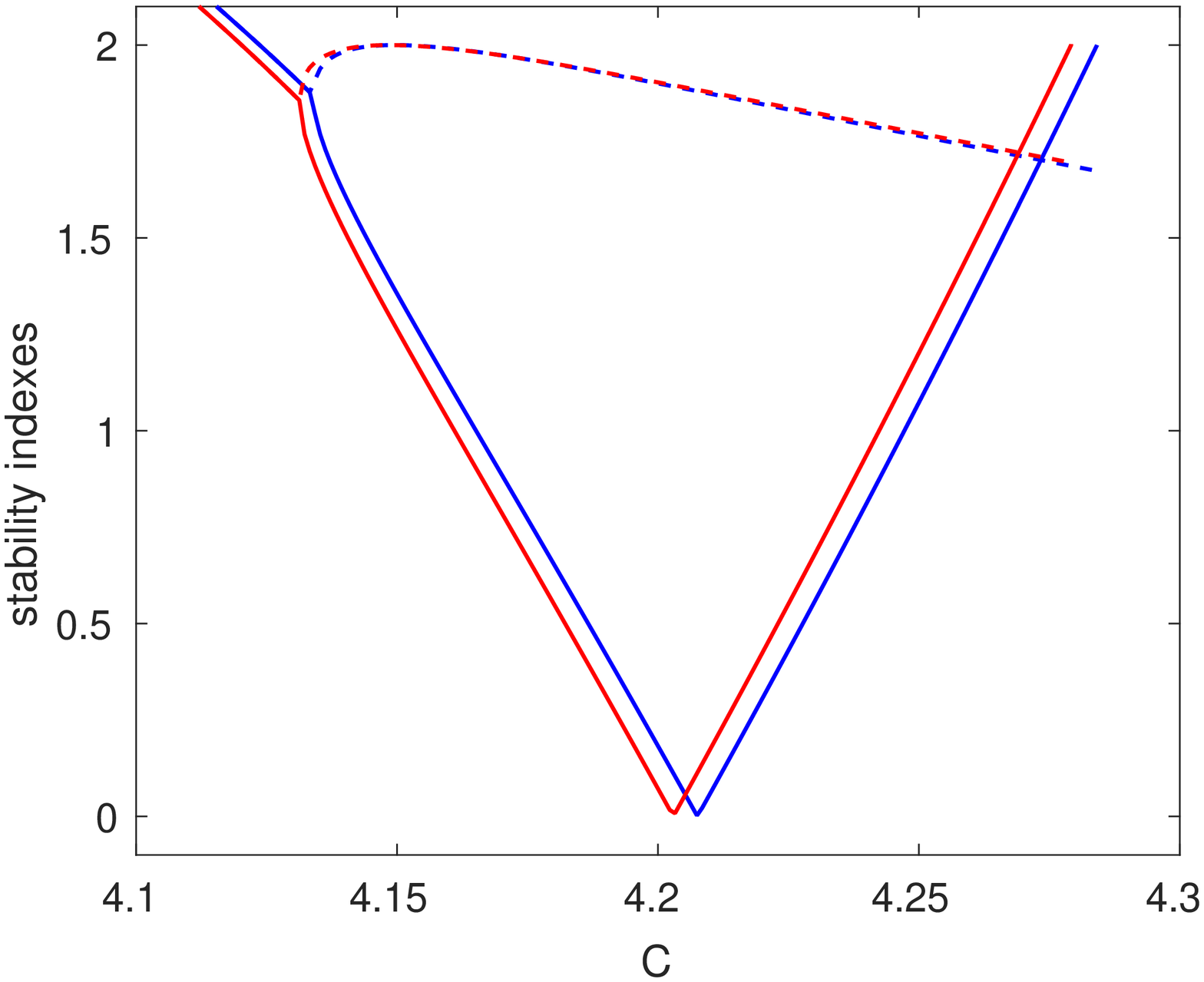}& \includegraphics[width=2.0in,height=2.0in]{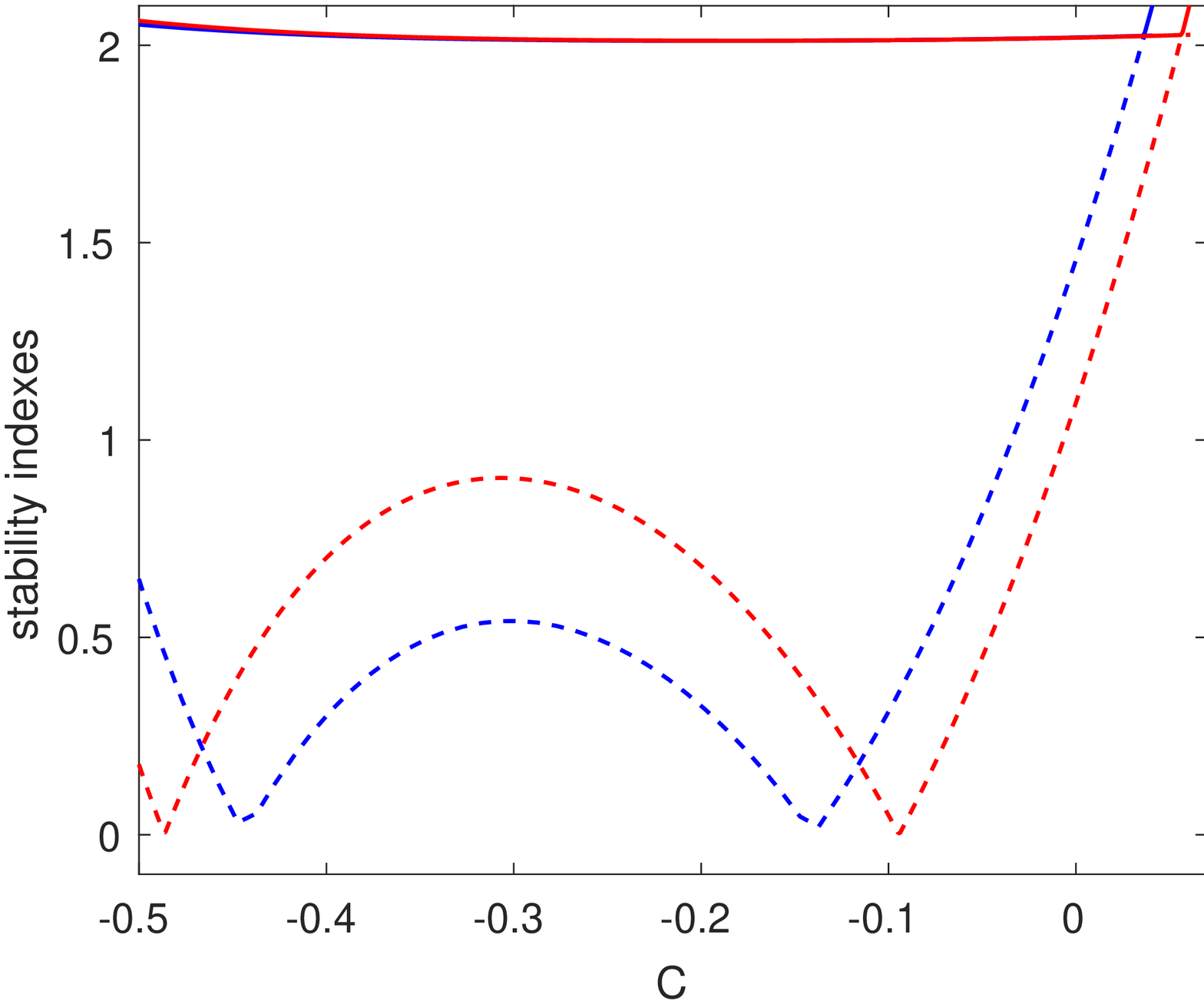}\\
  \includegraphics[width=2.0in,height=2.0in]{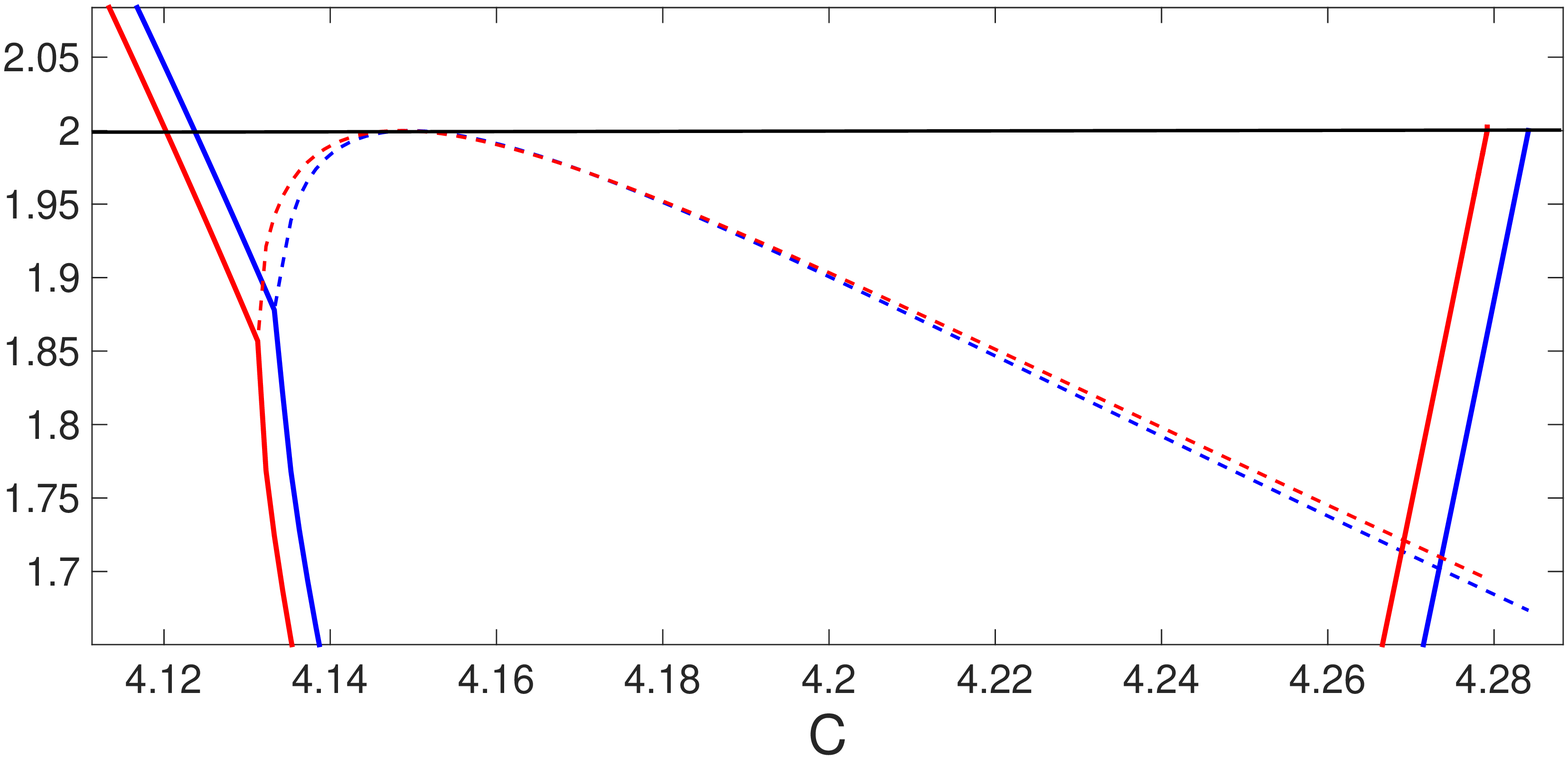}& \includegraphics[width=2.0in,height=2.0in]{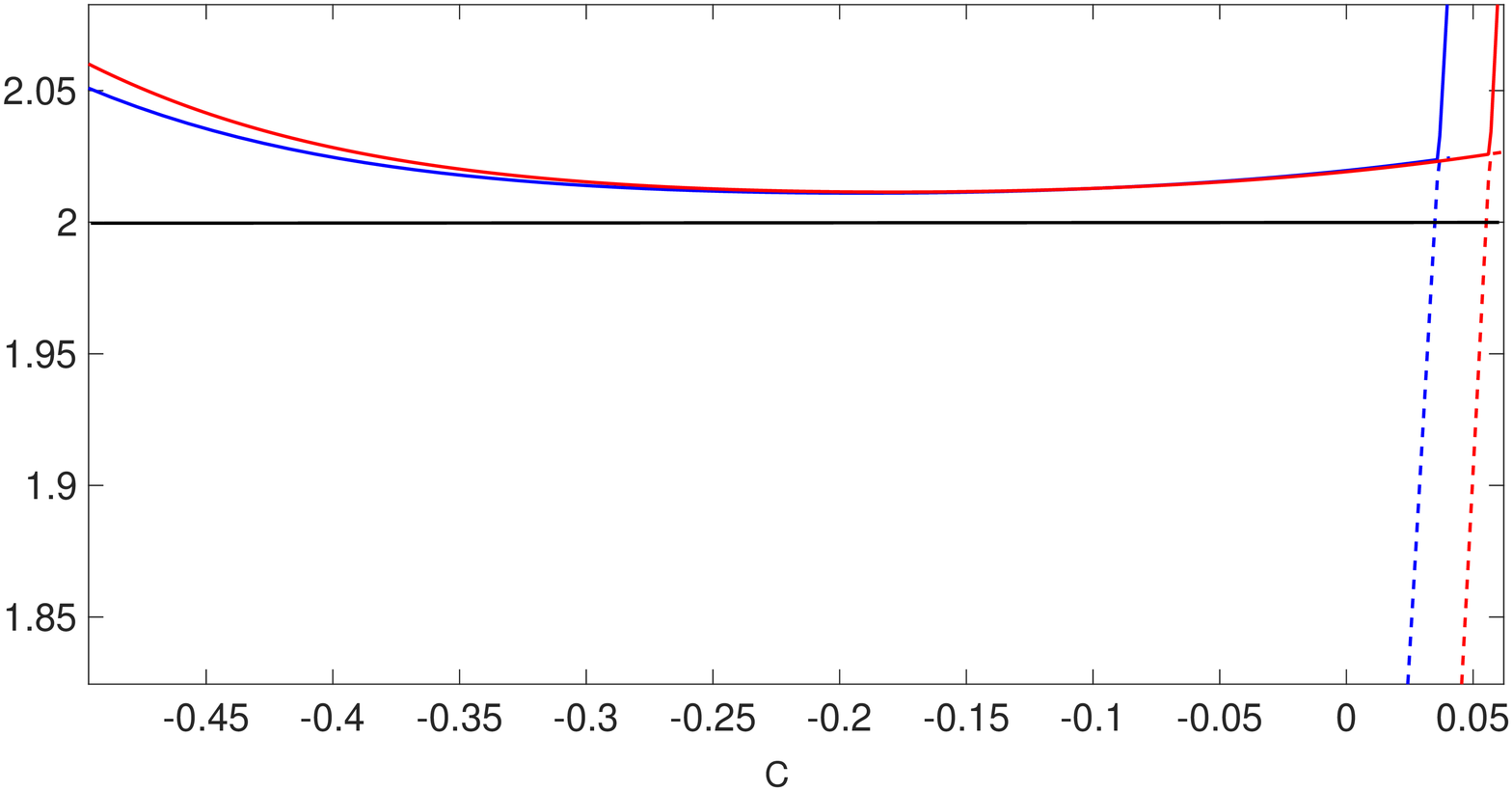}
\end{tabular}
 \caption{Top frames. Stability indexes for families $g1v$ (left) and $g'2v$ (right), for $\mu=0.0095$ (blue) and $\mu=0.00547$ (red). Bottom frames. Magnification of the previous figures. Slashed and solid curves represent indexes $s_1$ and $s_2$, respectively.}
\label{g1vandg2vstabilityfigures}
\end{figure}

\subsubsection{Periodic orbits emanating from planar retrograde orbits}

Let us start with a summary of results related to retrograde periodic orbits of the H3BP that will be necessary for exposing our results. According to pioneer numerical explorations \cite{Henon1974}, \cite{Michalodimitrakis}, there are not vertical critical orbits in the family of retrograde planar periodic orbits in the classical Hill's problem ($\mu=0$ in our case), so called $f$ family. Further explorations of three dimensional periodic orbits confirm that there are not periodic orbits emanating from the retrograde planar family $f$. For instance, in \cite{Russell} we can find a full exploration of spatial periodic orbits where the author uses the restricted three-body problem with mass ratios so small that his results are similar to those who appear in classical Hill's problem; such explorations support that H3BP does not possess periodic orbits as the ones shown in this Section.  The reader can consult the book \cite{Scheereslibro} to find a good compilation of work related to Hill's problem and its applications for orbital motion.
\newline

We underline that $f$ is not the unique family of planar retrograde motion in the classical Hill's problem, there exists other family which is named $l$, an analogous to the family of the restricted three--body problem. The family $l$ contains retrograde periodic orbits more distant than the members of the $f$ family. If we explore the members of this family for the H3BP we can find that there are not vertical critical families neither in this case.  Therefore, we can conclude that it is not possible to generate spatial orbits through planar families of retrograde periodic orbits in the H3BP.
\newline

Nevertheless, the situation turns out interesting when we go beyond the classical Hill's problem by increasing the mass parameter for $\mu>0$. Although $f$ remains stable, see \cite{Burgos}, the family $l$ now contains vertical critical orbits for $\mu=0.0095$ that corresponds to the Sun-Jupiter case.  The results in \cite{Burgos} show that as we increase the value of $\mu$ more critical orbits appear, even now in the family $f$ for $\mu=0.03$ and the approximated values $C\approx-1.84$, $C\approx-2.07$. As a consequence of this new feature, we can follow the bifurcating planar orbits from the family $f$ to the three dimensional case. Our numerical explorations produced two families of spatial orbits that we have named $f1v$ and $f2v$, whose members have symmetries \textbf{II} and \textbf{IV}, respectively. The evolution of these families and their relation with the so-called family $l$ is given in the following.
\newline

The geometry of the orbits can be appreciated in Figures \ref{f1vfigures} and \ref{f2vfigures}. As the reader can observe, at the beginning of each family the orbits are near Keplerian ones and as the value of $C$ decreases monotonically the orbits become more elongated in such a way that they tend asymptotically to an orbit on the $z-$axis which is composed by the union of {\it southern} and {\it northern polar orbits}. Actually, this behaviour holds for each value of $\mu \in (0,0.5]$. There exist other remarkable properties of the families related to the stability. For instance, for the family $f1v$ we have wide ranges, in the Jacobi constant, for which the orbits are stable. In Figure \ref{f1vandf2vstabilityfigures} we show the stability indexes for the cases $\mu=0.0095$, $\mu=0.00547$ and $\mu=0.03$. By means of increasing $\mu$, we obtained that the orbits tend to become more unstable, in such a way that for $\mu=0.5$ all of them are unstable. Something similar happens to $f2v$. However, we do not have stability since the stability index $s_2$ is slightly bigger than $2$ for $\mu=0.0095$, $\mu=0.00547$ and $\mu=0.03$.
\newline

The evolution of the stability indexes are shown in Figure \ref{f1vandf2vstabilityfigures}. We should mention that we are plotting the absolute value of the indexes versus the value of $C$, this is why some peaks appear, for instance, in the curves for family $f1v$ when $\mu=0.03$ the index $s_2$ changes its sign since for $C\approx-43.8999$, $s_2\approx-0.000130$ and for $C\approx-43.9099$, $s_2\approx 0.000038$. Something similar happens for the family $f2v$ when $\mu=0.03$, for $C\approx-15.099$, $s_1\approx-0.004$ and for $C\approx-15.109$, $s_1\approx-0.007$.  The pitchfork structure observed in the left side in Figure \ref{f1vandf2vstabilityfigures} corresponds to a change of stability index, from real to complex  for $C\approx-58.0999$.
\newline

\begin{figure}
  \centering
\begin{tabular}{cc}
  \includegraphics[width=2.0in,height=2.0in]{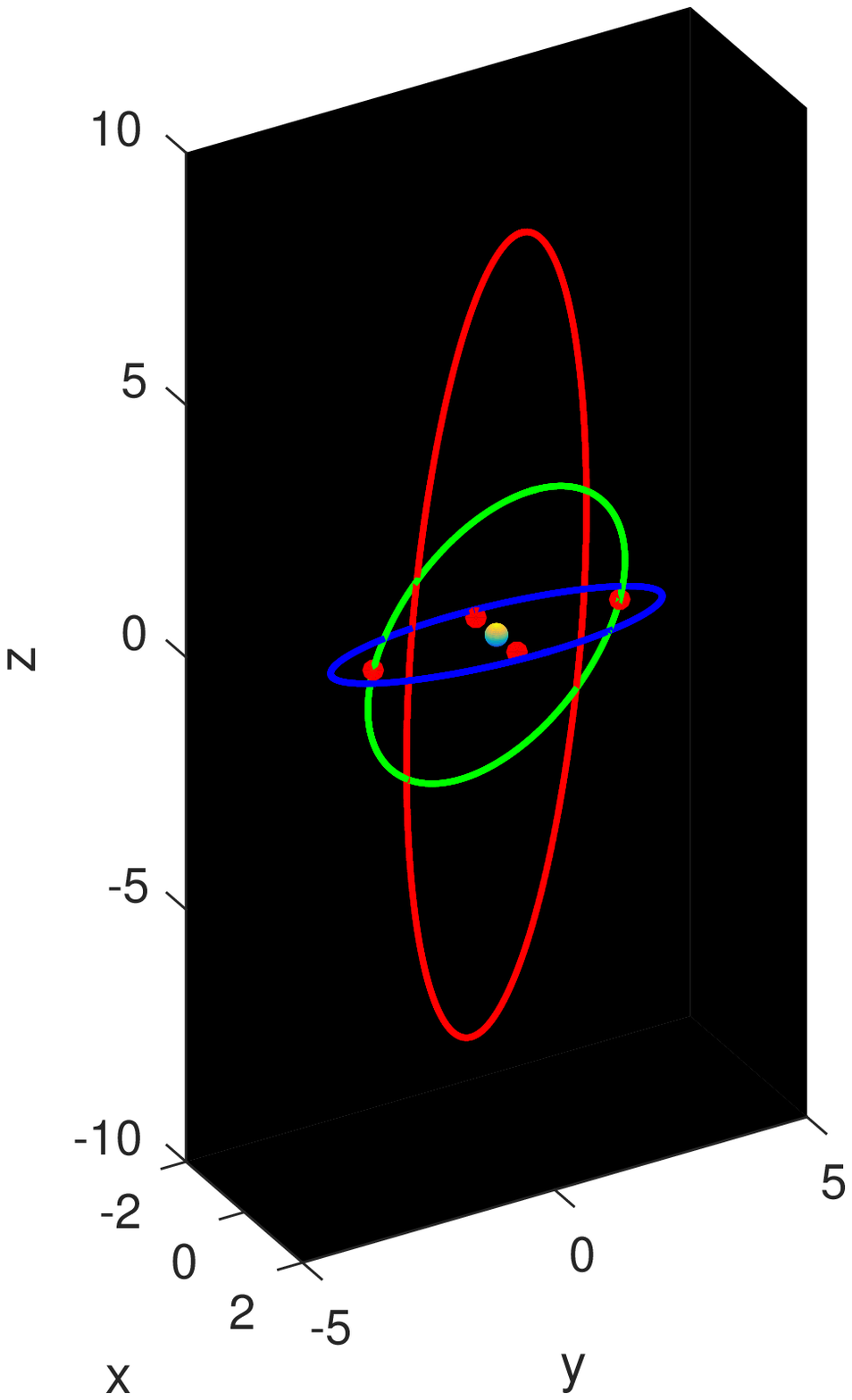}& \includegraphics[width=2.0in,height=2.0in]{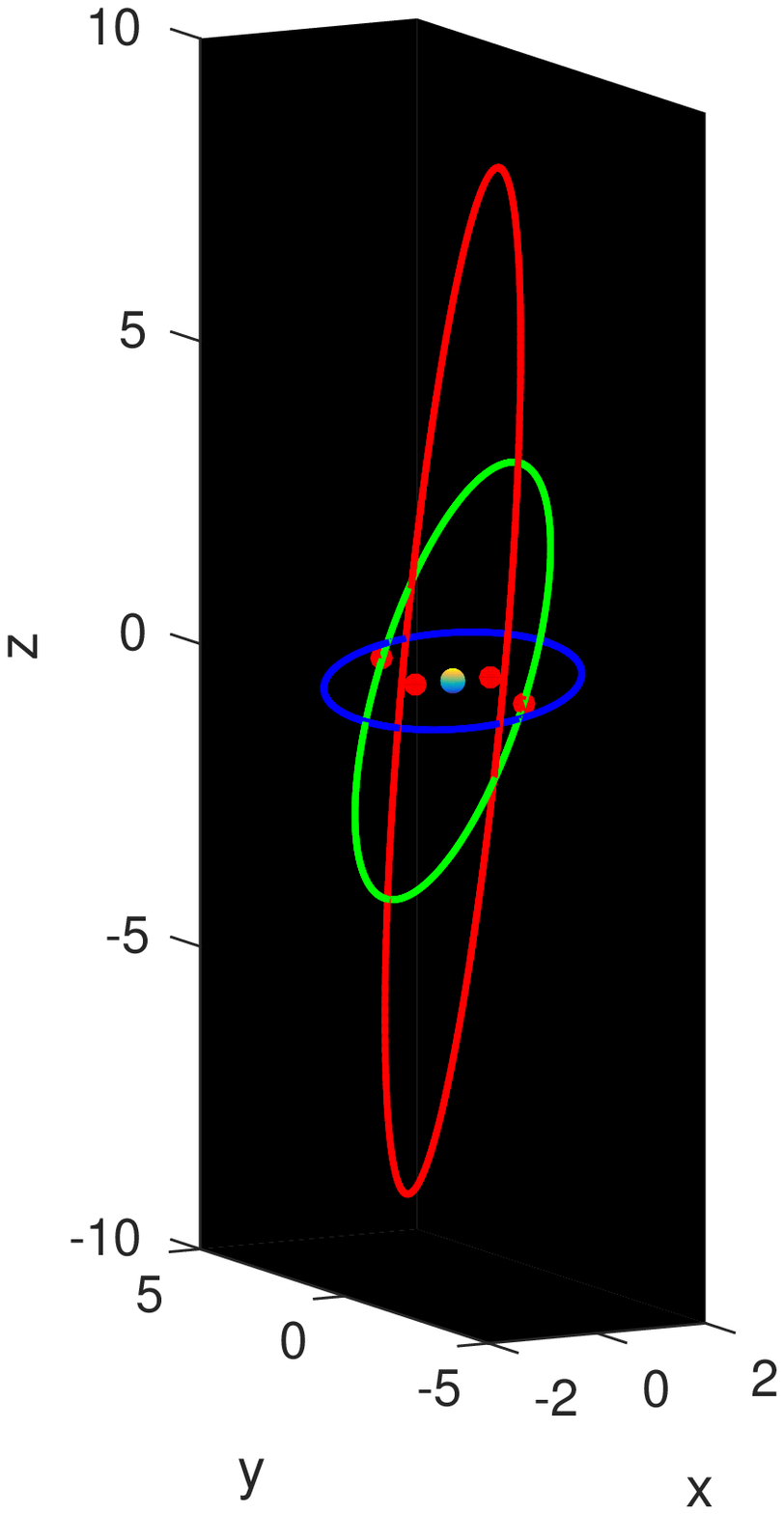}\\
  \includegraphics[width=2.0in,height=2.0in]{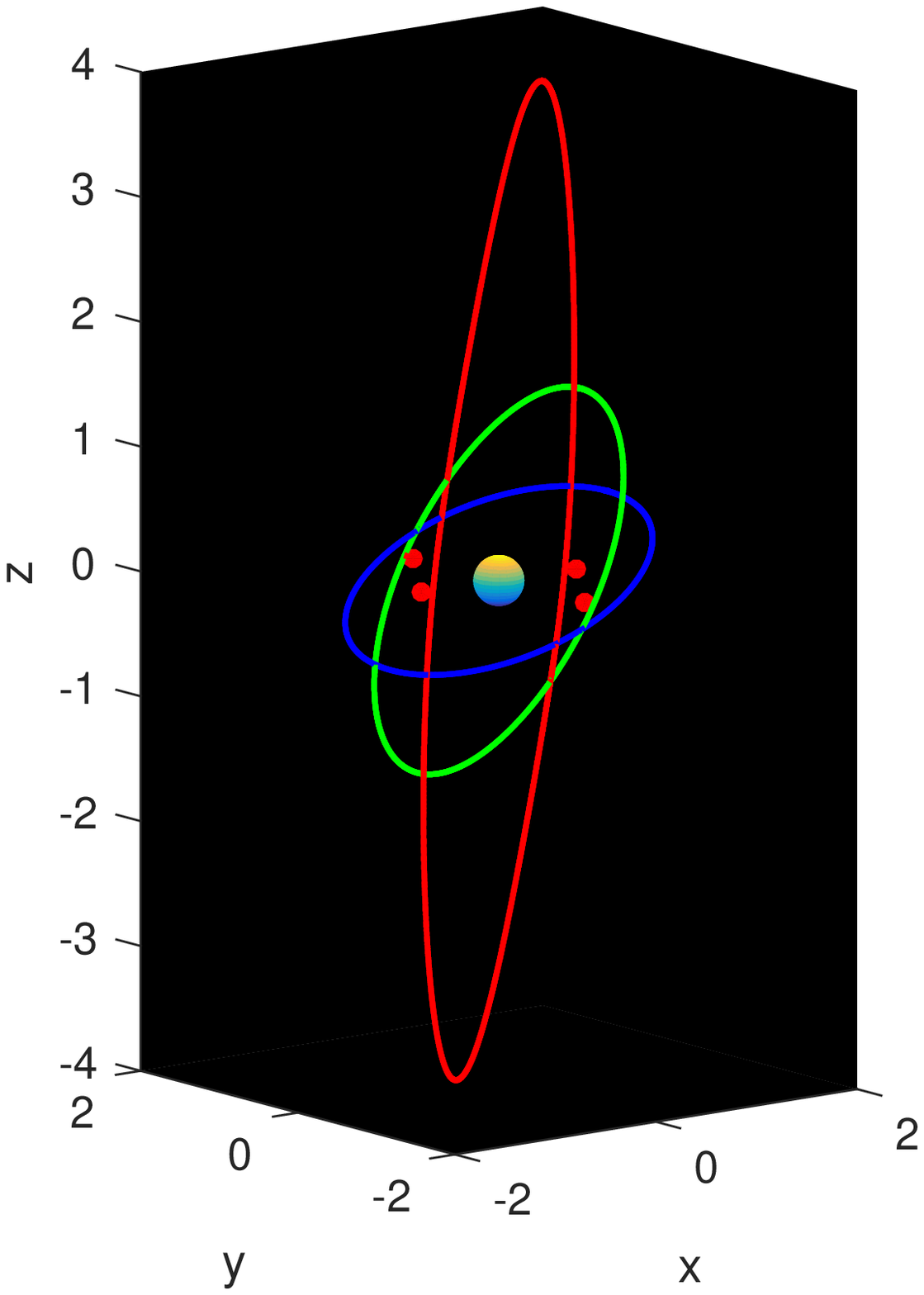}& \includegraphics[width=2.0in,height=2.0in]{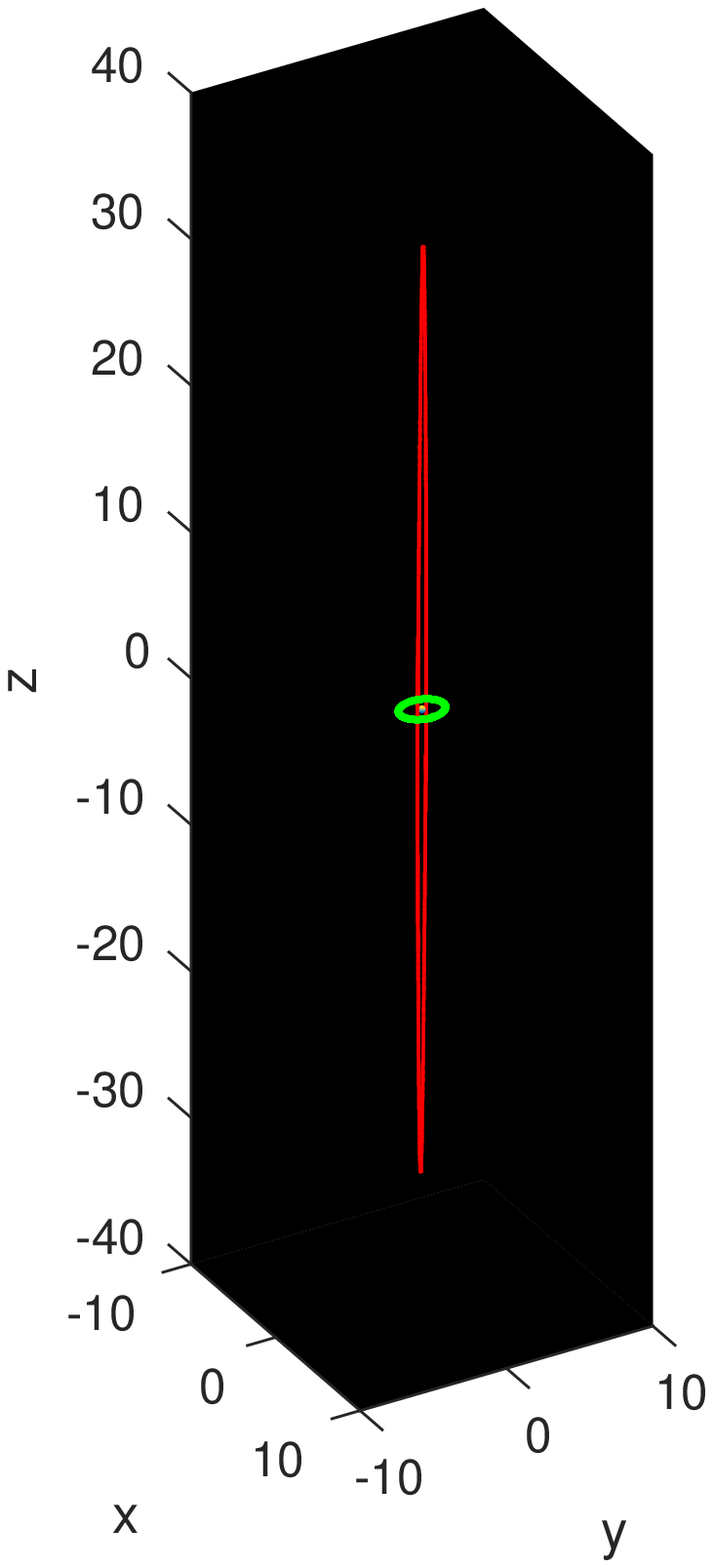}
\end{tabular}
 \caption{Different views for periodic orbits within family $f1v$, for $\mu=0.03$ and $\mu=0.5$, are given in first and second row, respectively. Red dots correspond to equilibrium points.}
 \label{f1vfigures}
\end{figure}

There are critical orbits within $f1v$ and $f2v$ for all values of $\mu$ mentioned above. In Table \ref{Criticalf2v} we have shown initial conditions for each critical orbit in the family $f2v$. For the case $\mu=0.00095$ there are two critical orbits at $C=-31.8301$ and $C=-252.429$, for $\mu=0.00547$ two critical orbits at $C=-9.35730$ and $C=-78.7190$, for $\mu=0.03$ two critical orbits at $C=-2.07890$ and $C=-25.7990$, and finally when $\mu=0.5$ the same situation happens at $C=1.39270$ and $C=-178.5990$.
\newline

\begin{figure}
  \centering
\begin{tabular}{cc}
  \includegraphics[width=2.0in,height=2.0in]{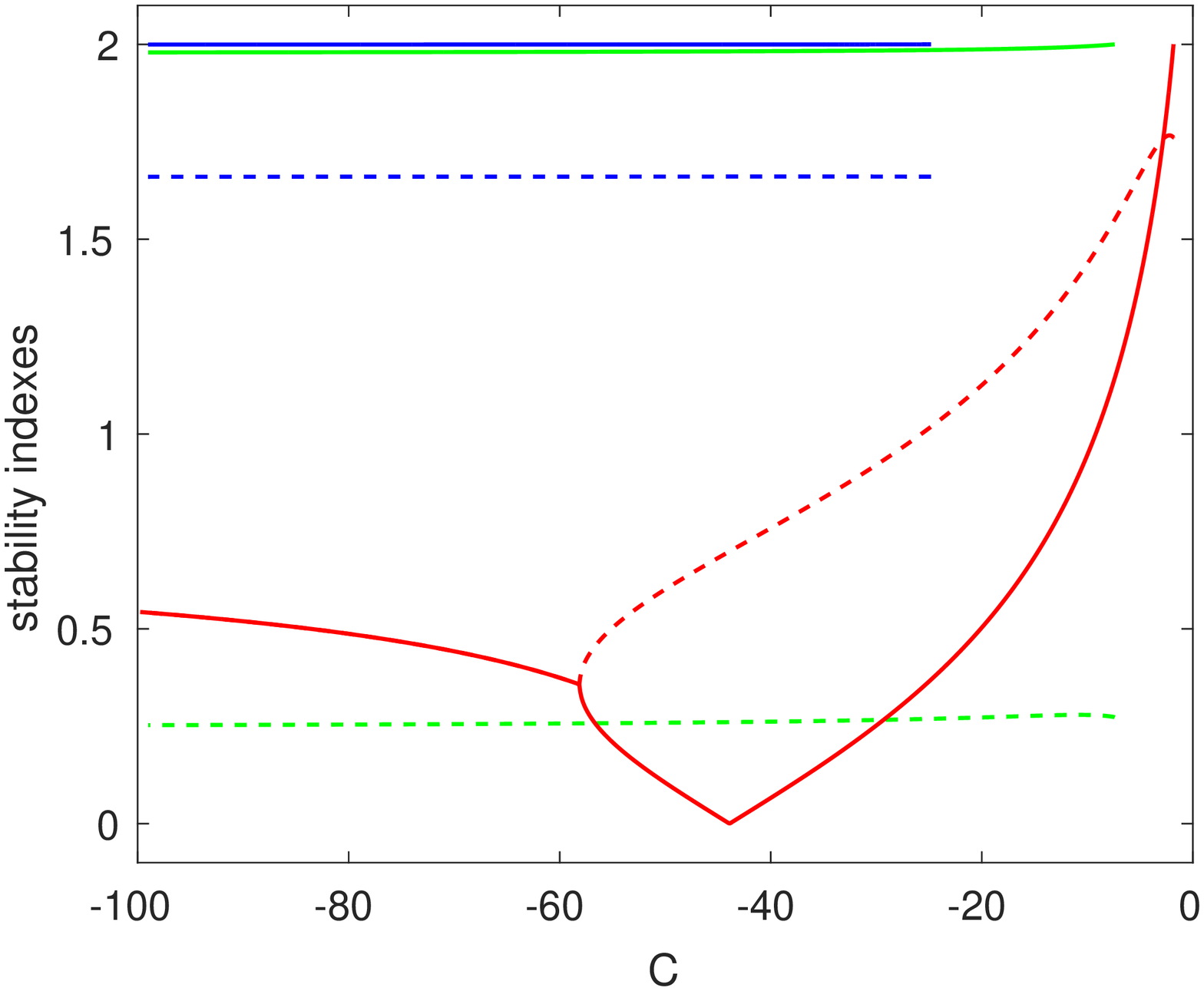}& \includegraphics[width=2.0in,height=2.0in]{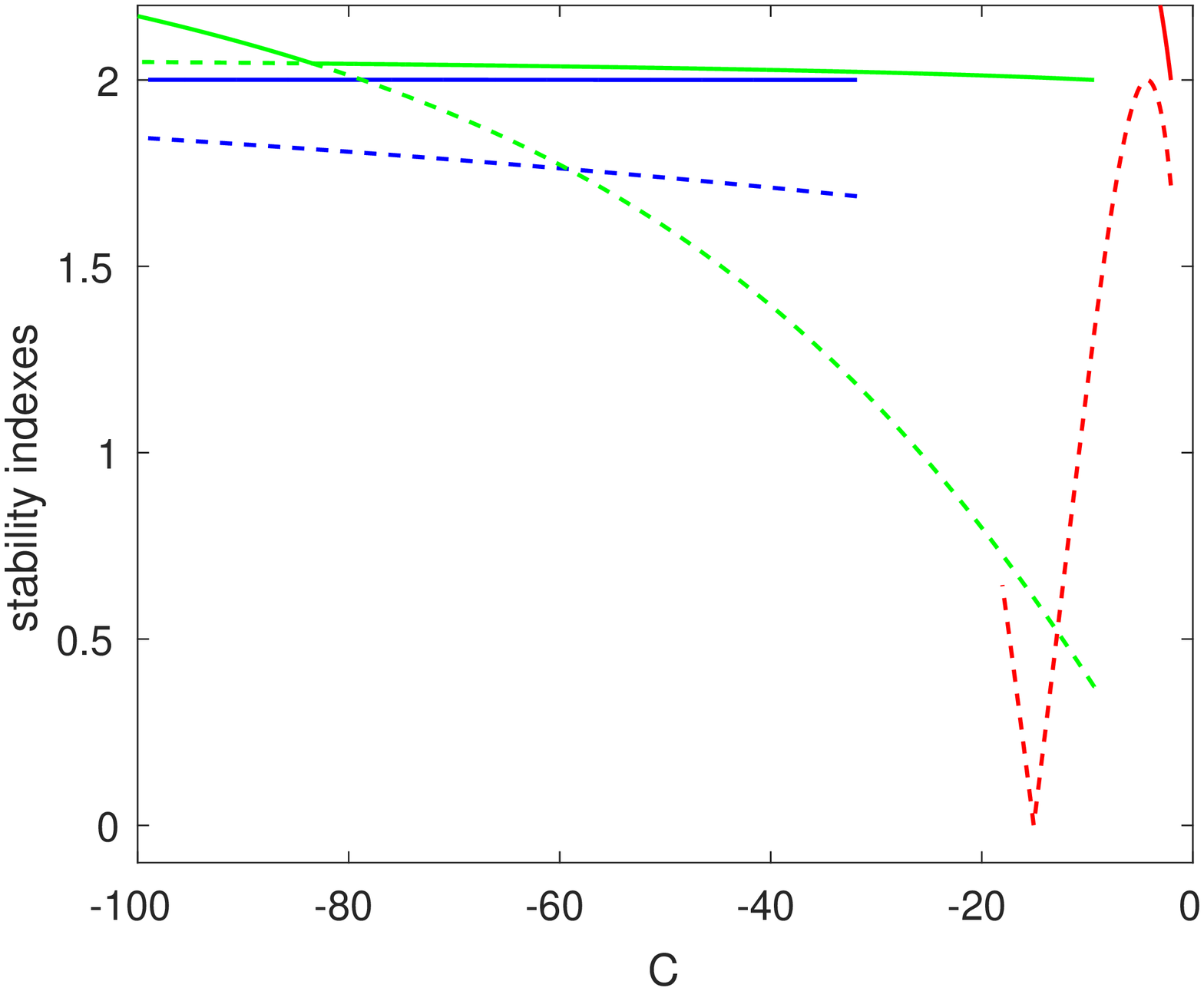}\\
 \includegraphics[width=2.0in,height=2.0in]{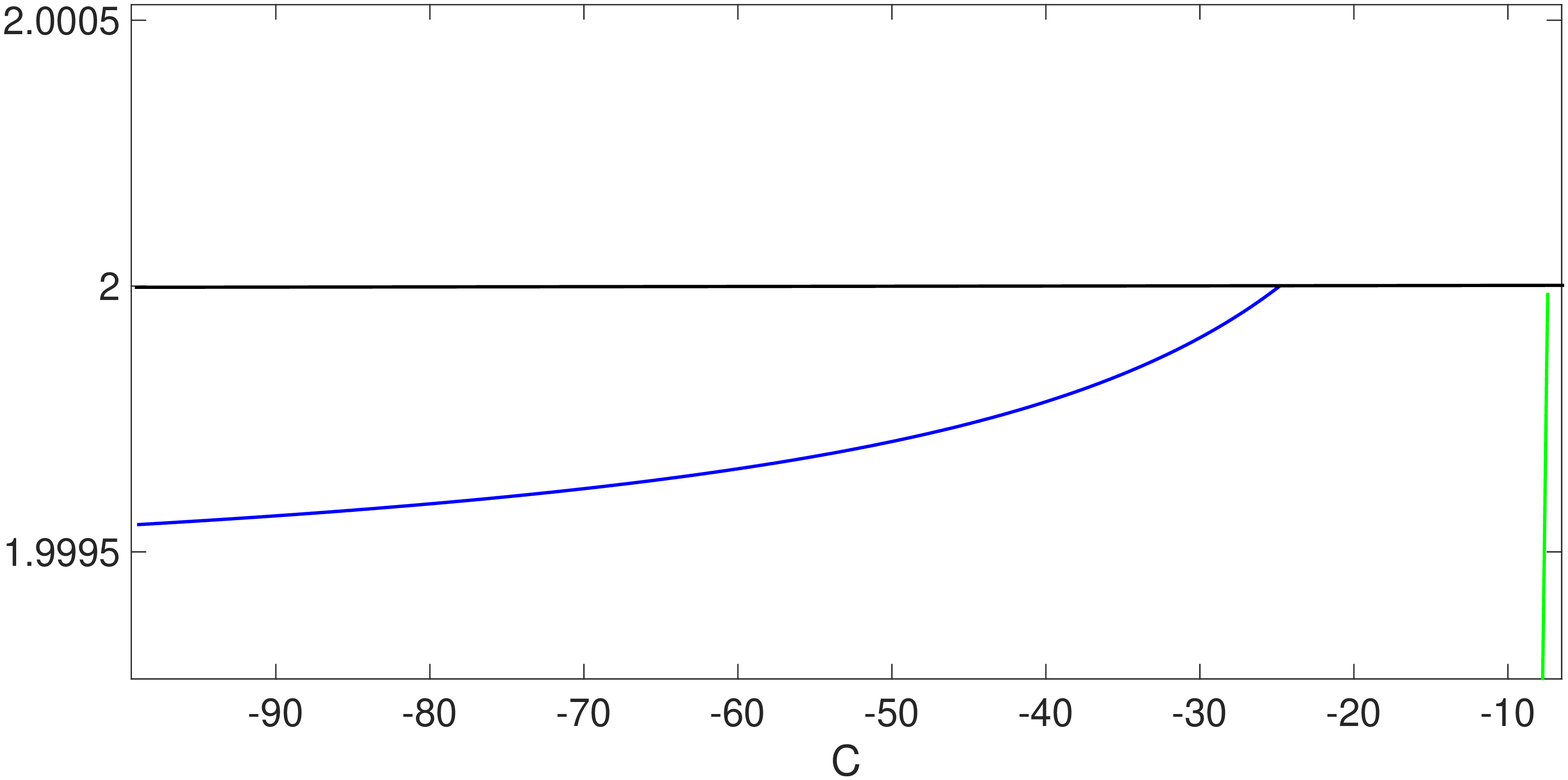}& \includegraphics[width=2.0in,height=2.0in]{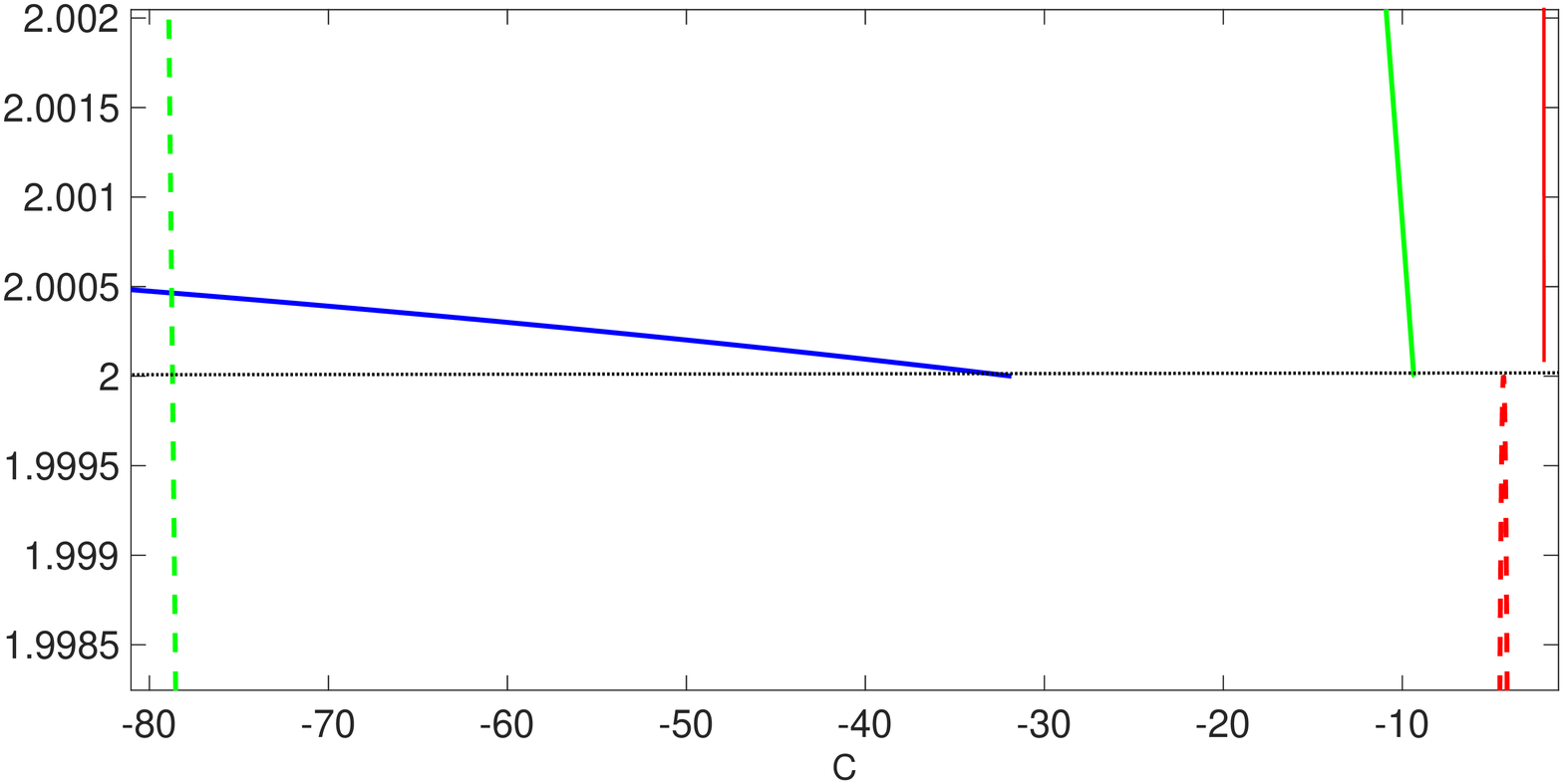} 
\end{tabular}
\caption{Top frames. Stability indexes for families $f1v$ (left) and $f2v$ (right), for $\mu=0.0095$ (blue), $\mu=0.00547$ (green) and $\mu=0.03$ (red). Bottom frames. Magnification of the previous figures. Slashed and solid curves represent indexes $s_1$ and $s_2$, respectively. }
\label{f1vandf2vstabilityfigures}
\end{figure}

There are not critical orbits in the family $f1v$, except the ones corresponding to bifurcations with planar orbits. For the values $\mu=0.00095$, $\mu=0.00547$, $\mu=0.03$, $\mu=0.5$ the bifurcation occurs at $C=-24.81470$, $C=-7.40680$, $C=-1.83040$ and $C=1.11540$, respectively. The initial conditions of each case can be seen in Table \ref{Criticalf1v}.

\begin{figure}
  \centering
\begin{tabular}{cc}
  \includegraphics[width=2.5in,height=1.8in]{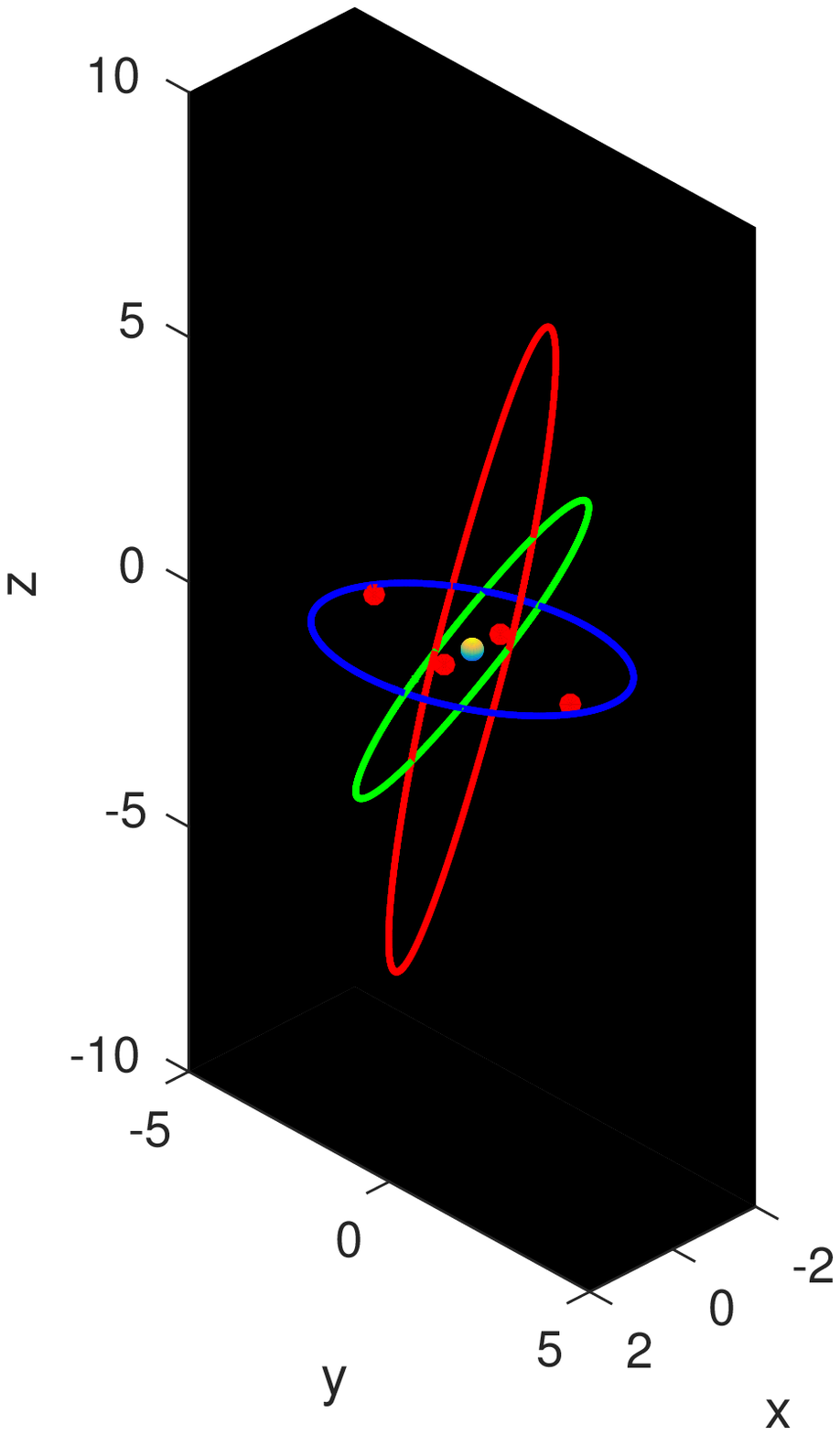}& \includegraphics[width=2.5in,height=1.8in]{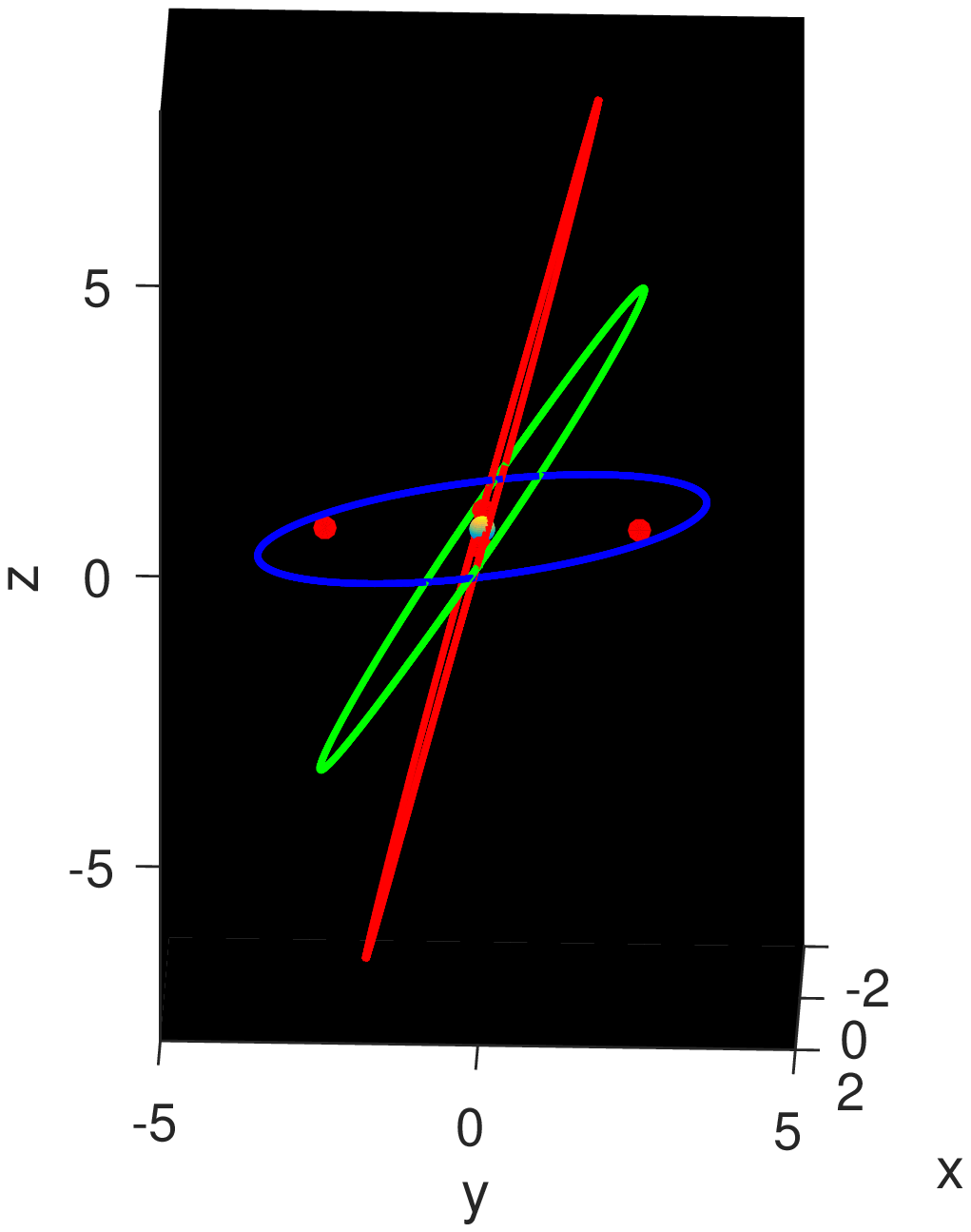}\\
  \includegraphics[width=2.5in,height=1.8in]{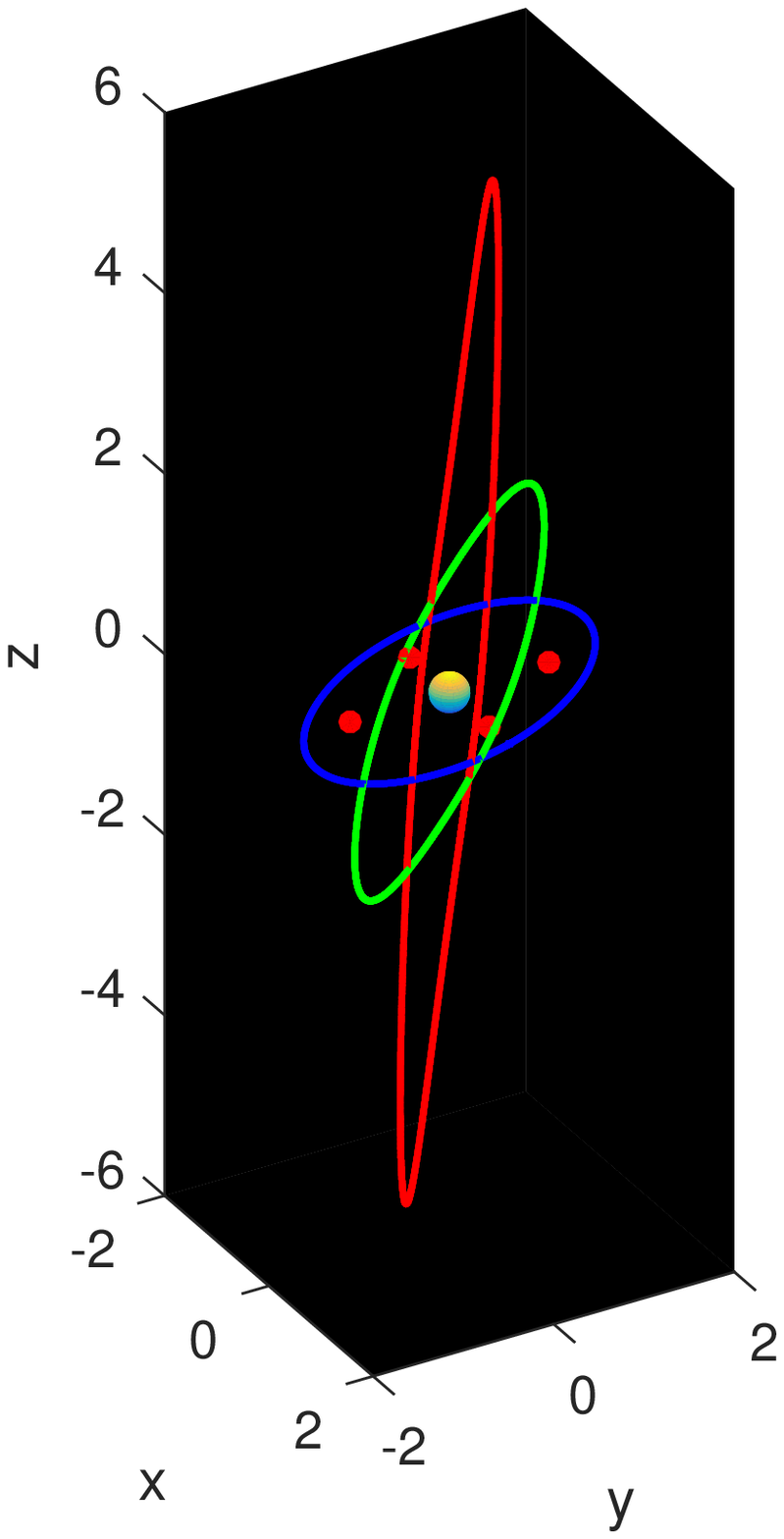}& \includegraphics[width=2.5in,height=1.8in]{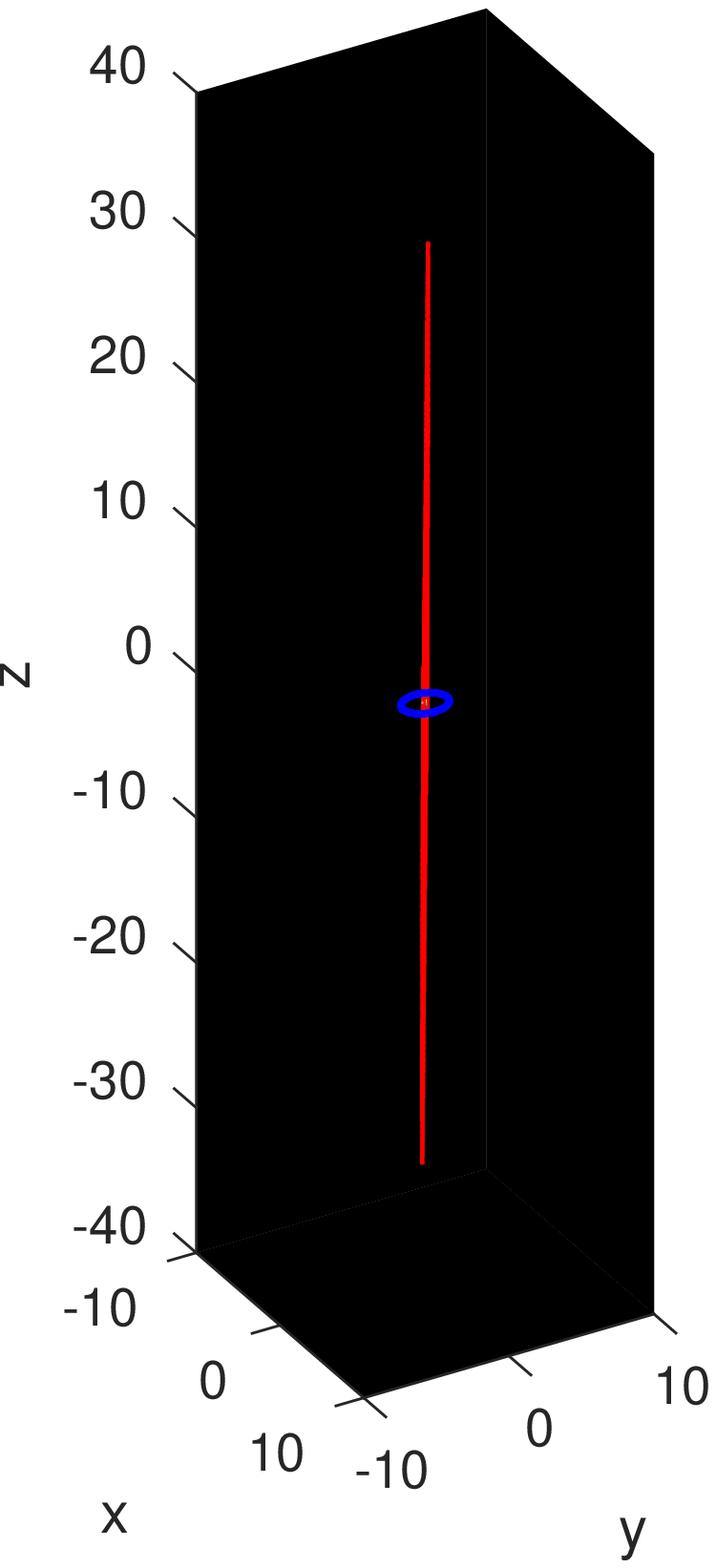}
\end{tabular}
\caption{Different views for orbits within family $f2v$, for $\mu=0.03$ and $\mu=0.5$, are given in first and second row, respectively. Red dots correspond to equilibrium points.}
\label{f2vfigures}
\end{figure}

\begin{table}[ht]
\caption{Initial conditions for critical orbits belonging to family $f2v$. First and second rows correspond to $\mu=0.00095$, third and fourth rows to $\mu=0.00547$, fifth and sixth to $\mu=0.03$ and seventh and eighth rows to $\mu=0.5$.} 
\centering 
\resizebox{14cm}{!} {
\begin{tabular}{c c c c c c c}
\hline\hline 
$s_1$ & $s_2$ & $x_{0}$ & $\dot{z}_{0}$ & $\dot{y}_{0}$ & $T/4$ & $C$\\ [0.5ex] 
\hline 
1.687 & 2  & 5.684318357140213 &  0.0 &   -11.359899708252261  &     1.569933896365237 &-31.830100 \\
2 &  2.001  & 3.805153135319000 &  -15.446312770335599  &   -7.600832458845457 &    1.570412322635000  & -252.42890   \\
0.372 & 2 &  3.196249565181473 &  0.0 &   -6.364393688192154 &    1.565861003785000 &- 9.357300 \\
2 &  2.043 &   2.137000533224141 &  -8.677008319211858 &   -4.243654513195445&     1.568605219743404  &-78.71900  \\
1.717 &   2 &   1.885684739136071 &  0.0 &   -3.683559571730871 &     1.544420430266403  & -2.078700  \\
  2 &   3.151 &   1.250803179799384 &  -5.117165015123839&   -2.408598233735526 &     1.559429117736457 &  -25.79900 \\
  2 &   322.7 &    1.167202998543122  &  0.0&    -1.840139773623922 &      1.352460950907277  & 1.392700\\
    2 &   148.7 &    0.346672927777300  &  -13.580640736995900&    -0.452474876303797 &      1.569114389268850 & -178.5990\\
\hline 
\end{tabular}
}
\label{Criticalf2v} 
\end{table}

\begin{table}[h!]
\caption{Initial conditions for critical orbits belonging to family $f1v$. First, second, third and fourth rows correspond to $\mu=0.00095$, $\mu=0.00547$, $\mu=0.03$ and  $\mu=0.5$, respectively.} 
\centering 
\resizebox{14cm}{!} {
\begin{tabular}{c c c c c c c}
\hline\hline 
$s_1$ & $s_2$ & $x_{0}$ & $z_{0}$ & $\dot{y}_{0}$ & $T/2$ & $C$ \\ [0.5ex] 
\hline 
1.660 & 2  & 5.018825684425040 &  0.0 &   -10.036194771688093  &    3.134509395884316 & -24.81470\\
0.272 & 2 &  2.841184752918077 &  0.0 &   -5.677016018058075 &    3.102515880798499 & -7.406800 \\
1.761 &   2 &   1.726239106614033 &  0.0 &   -3.424781515503815 &     2.969523669896602   &-1.830400  \\
  2 &   195.2 &    1.139963051312378  &  0.0&    -1.887578683046186 &      2.515936791605256 & 1.115400 \\
\hline 
\end{tabular}
}
\label{Criticalf1v} 
\end{table}

Of course, a natural question would be: what does happen with periodic orbits when $\mu\to0$? The explorations revealed, on one hand, an asymptotic behaviour  to $z-$axis, and on the other hand the orbits tend to planar critical orbits of the family $l$ above described!  In Figure \ref{lfigures} we observe the vertical stability curve for the family $l$ with $\mu=0.00095$ where two vertical critical orbits appear and correspond to the ending of families $f1v$ and $f2v$.  We realize the size of the orbits of the family $l$ increases as $\mu\to0$. We conjecture that the three dimensional orbits tend to infinity while they surround the equilibrium points $L_3$ and $L_4$ and these points tend to infinity as $\mu\to0$.

\begin{figure}[h!]
  \centering
\begin{tabular}{cc}
  \includegraphics[width=2.5in,height=1.8in]{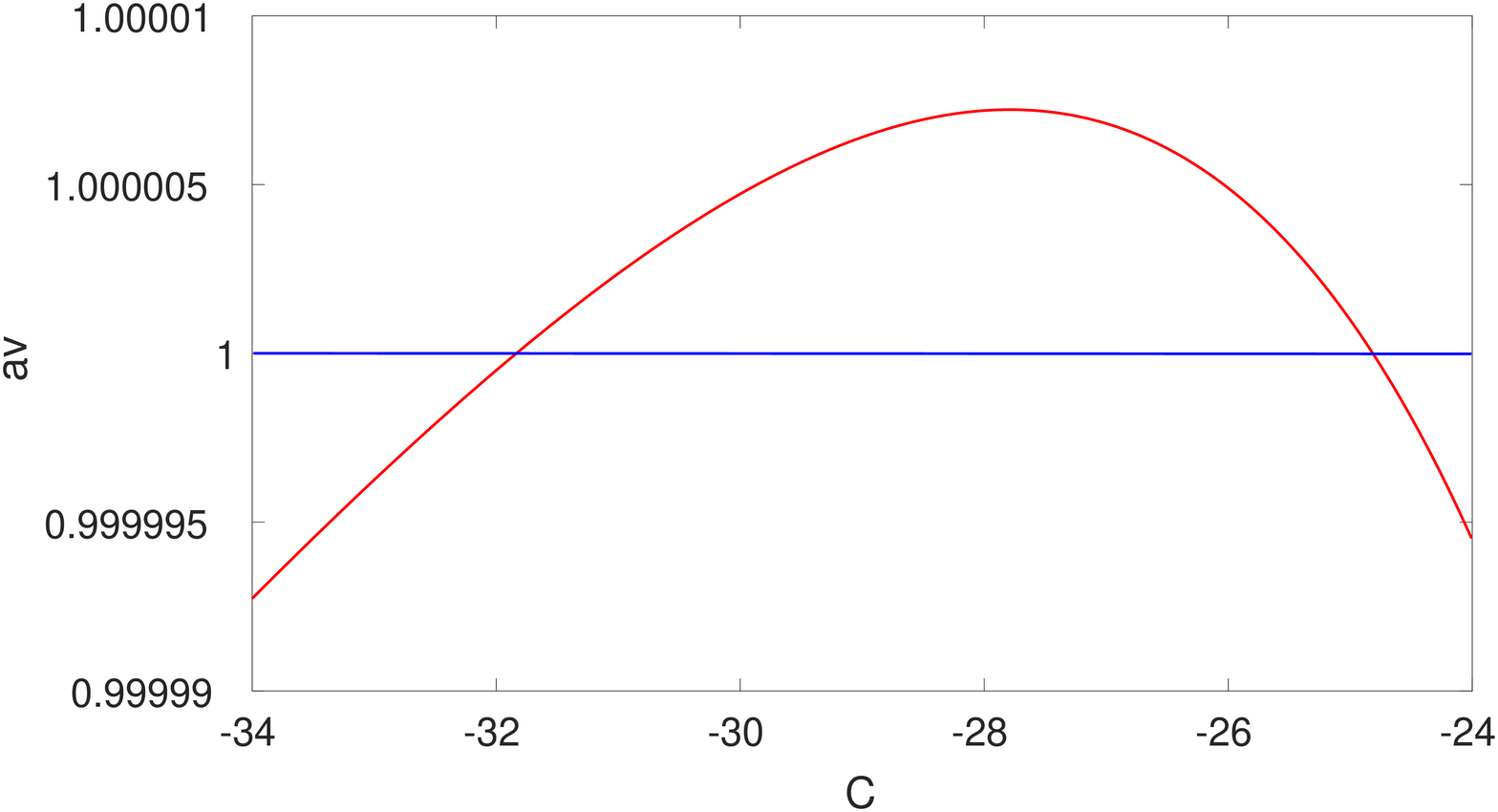}& \includegraphics[width=2.5in,height=1.8in]{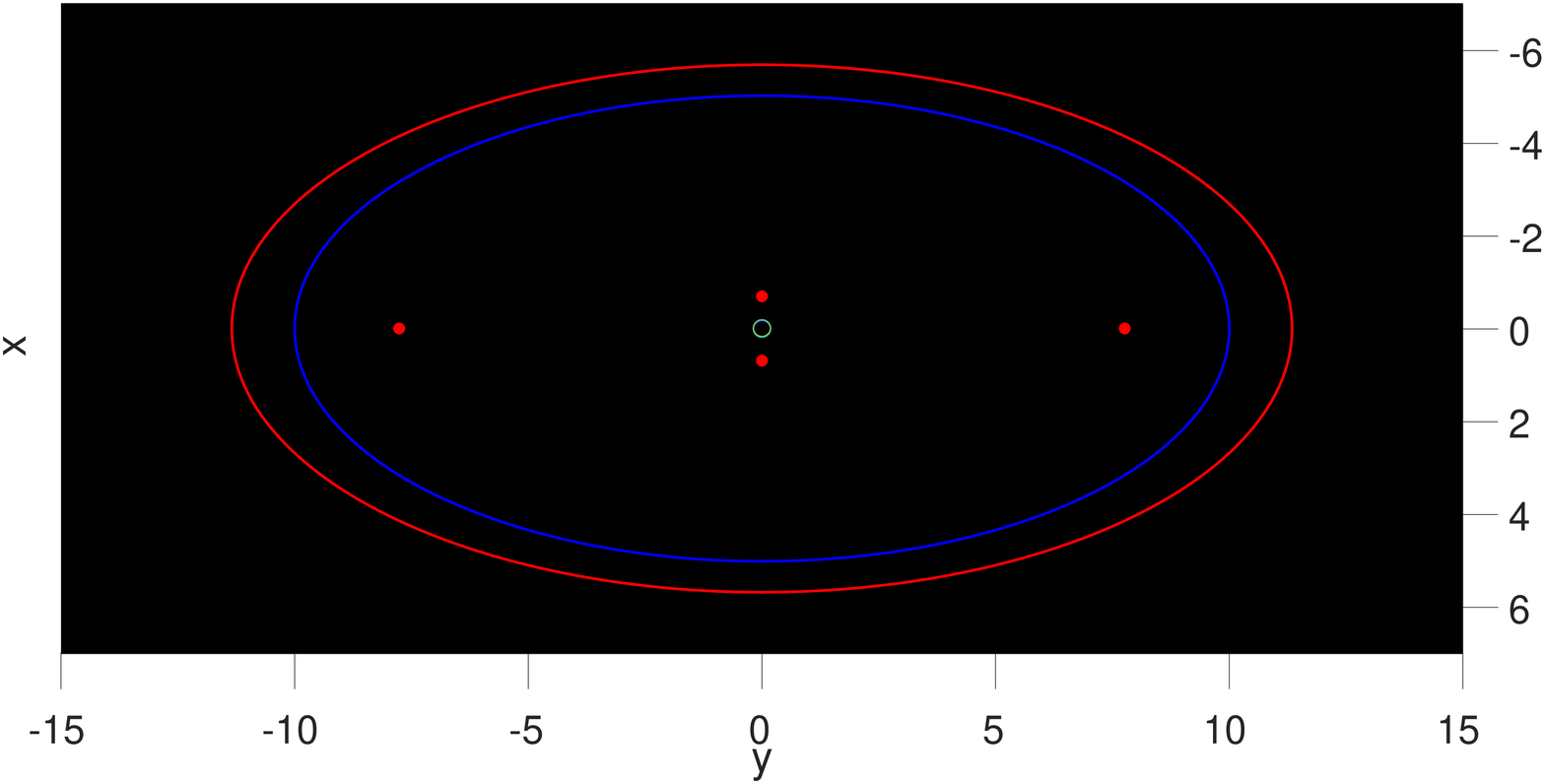}
\end{tabular}
\caption{Left. Characteristic curve of the vertical stability index $a_v$ for the planar family $l$ for $\mu=0.00095$. Right. Bifurcating planar orbits for the family $f2v$ (red) and $f1v$ (blue).}
\label{lfigures}
\end{figure}

\subsection{ Three dimensional periodic orbits around $L_1$}\label{group1}

\subsubsection{Vertical Orbits}
This family of eight-shaped periodic orbits corresponds to a continuation of linear vertical oscillations for motion class $(a)$ of the equations (\ref{linearmotionL1}). The members of this family possess symmetry of type \textbf{III}. In the H3BP, it was found \cite{Pelaez} that the size of these eight shaped orbits increases as the value of $C$ decreases monotonically; this behaviour stands even when $C\to-\infty$.  All orbits are unstable and there is just one critical spatial orbit for $C = 0.512431$. In particular, the behaviour of the corresponding vertical families for the three collinear equilibrium points for the R3BP is quite different from the reported  in \cite{ZagourasII} and \cite{Gomez}. In the R3BP case, the vertical families end at bifurcation with a planar orbit for several values of $\mu$.
\newline

For the H4BP we have found that the vertical orbits are unstable and show the same qualitative behaviour as in the classical H3BP. However, now there are several critical orbits for some $\mu\in(0,1/2]$ whose initial conditions can be found in Table \ref{Criticalvertical}. As a matter of fact, these critical orbits correspond to bifurcations with other families of periodic orbits of this problem.  In Section \ref{haloorbits}, we will see that the vertical orbits are connected in an interesting way with the family of Halo orbits. Moreover, in Section \ref{Lanefamily} we will show that the family of vertical orbits is connected to the one of planar Liapunov orbits through another family of periodic orbits which have been named as Lane family. The critical orbits shown in Table \ref{Criticalvertical} are true bifurcations that will be followed in Section \ref{Lanefamily}.

\begin{table}[ht]
\caption{Initial conditions for critical orbits belonging to the vertical family. First row corresponds to $\mu=0.00095$, second and third rows to $\mu=0.03$, fourth row to $\mu=0.1$ and fifth row to $\mu=0.5$.} 
\centering 
\resizebox{14cm}{!} {
\begin{tabular}{c c c c c c c}
\hline\hline 
$s_1$ & $s_2$ & $x_{0}$ & $\dot{z}_{0}$ & $\dot{y}_{0}$ & $T/4$ & $C$\\ [0.5ex] 
\hline 
2 & 435.8  & 0.540712768638617  &  1.923960163342287  &    -0.581718029353020  &    1.058338628606103 & 4.293059  \\
2 &  526.4  & 0.561988213414619   &  1.749237807880294  &    -0.446479184775191  &     0.982011756142002  & 1.225958  \\
2 &  525.7  & 0.561922346788838   &  1.749818438199988  &    -0.446891515959117  &     0.982243141474447  & 1.223758  \\
2 &  1424.2  & 0.670700818421059    &   0.892265688101642   &    -0.086904548188239  &    0.837337764530224 & 3.429539  \\
2 &  625.3  & 0.721064718369786    &   0.871684257590844  &    -0.082685710206543  &    0.900258042567795  & 3.176858 \\
 [1ex] 
\hline 
\end{tabular}
}
\label{Criticalvertical} 
\end{table}

\subsubsection{Halo Orbits}
\label{haloorbits}
This family of periodic orbit was named $a1v$ for the H3BP in \cite{Michalodimitrakis} where it was reported the evolution of such family as follows: it starts at a bifurcating orbit of the family of planar Lyapunov orbits at $C = 4.00531$,  then as the value $C$ decreases the orbits become more inclined with respect to the plane $xy$. After passing a reflection (turning point), the value of $C$ increases and the width of the orbits starts to decrease in such a way the orbits tend to collision. Furthermore, the author made the conjecture: the orbits tend asymptotically towards a rectilinear collision motion on the half-axis $x=y=0$, $z\leq0$. Further explorations, as the contained in \cite{Pelaez}, provide data about the stability of this family; almost all orbits within the family are unstable except in a short region of $C$ with length $\Delta C=0.02612$.
\newline

Considering the boundary problem for the symmetry \textbf{II}, we perform a numerical continuation for $C$ taking several values of $\mu$. We found that there exists some range for $\mu$ in the interval $[0,1/2]$ where the behaviour of the families is analogous to the one of the H3BP; the families start at bifurcation with the planar family, pass by a reflection and tend to rectilinear collision motion by means of reaching \textit{southern polar orbits}.  However, we noticed that the main effect of considering a mass parameter $\mu$ different from zero in the continuation is on the stability of the orbits. For instance, for $\mu=0.00095$ there are several critical orbits and a stable zone of length $\Delta C=0.026307$ between the values $C=1.069223$ and $C=1.09553$ which is slightly bigger than the range in the H3BP. Something similar happens for other values of $\mu$; when $\mu=0.03$ the stable zone has length $\Delta C=0.032528$ and appears between $C=1.074566$ and $C=1.07094$, and for $\mu=0.1$ the stable zone has width $\Delta C=0.04856$ and it is delimited by $C=1.133094$ and $C=1.084534$. The case $\mu=0.5$ is not an exception; the stable zone has length $\Delta C=0.1317$ and appears between $C=1.089994$ and $C=1.221694$. 
\newline

 Nevertheless, there is a remarkable property of this family, there exists a value of $\mu_c$ for which the family does not start at a bifurcation with the planar family of Lyapunov orbits, it starts at a bifurcation with the vertical family as it is shown in Figure \ref{halofigures2}. We can observe that the characteristic curves reach  the value $z_0=0$ for $\mu=0.1$ and $\mu=\mu_p$ (considering its value in double precision with round up). However, for $\mu$ slightly greater than $\mu=\mu_p$ the characteristic curve does not reach the value $z_0=0$ any more. Moreover,  it was observed that this property remains true for $\mu\in(\mu_p,0.5]$ in such a way the family starts in a bifurcation with the vertical family and it ends with rectilinear collision motion. We should stress that this is not just a qualitative argumentation since in the upper right corner of Figure \ref{halofigures2} we can observe that the stability index $s_1$ reach the value $2$ for $\mu=\mu_p$ and $\mu=0.11755$ but for the first case the family ends at a planar Lyapunov family and in the second case does not.
\newline

\begin{figure}
  \centering
\begin{tabular}{cc}
  \includegraphics[width=2.5in,height=1.8in]{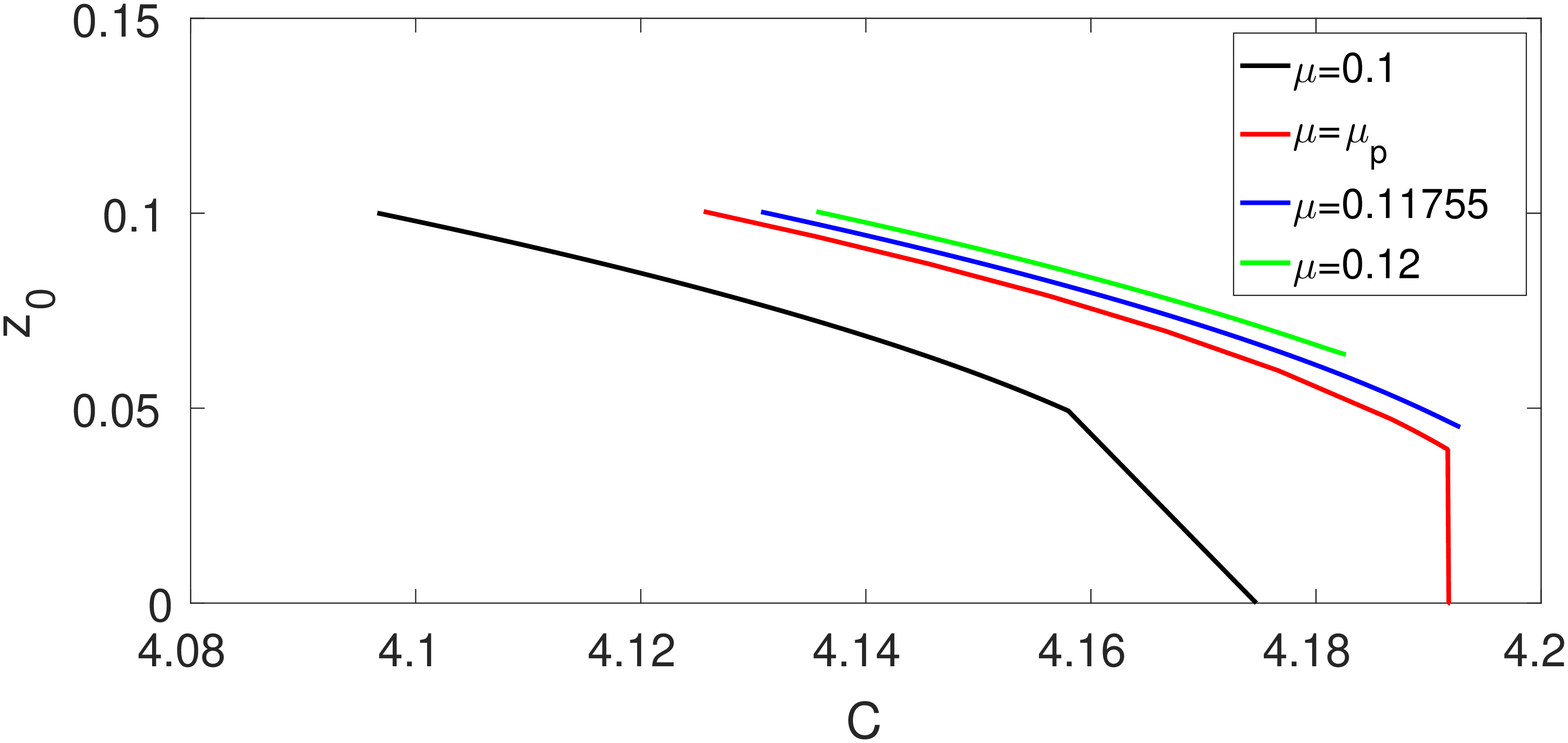}& \includegraphics[width=2.5in,height=1.8in]{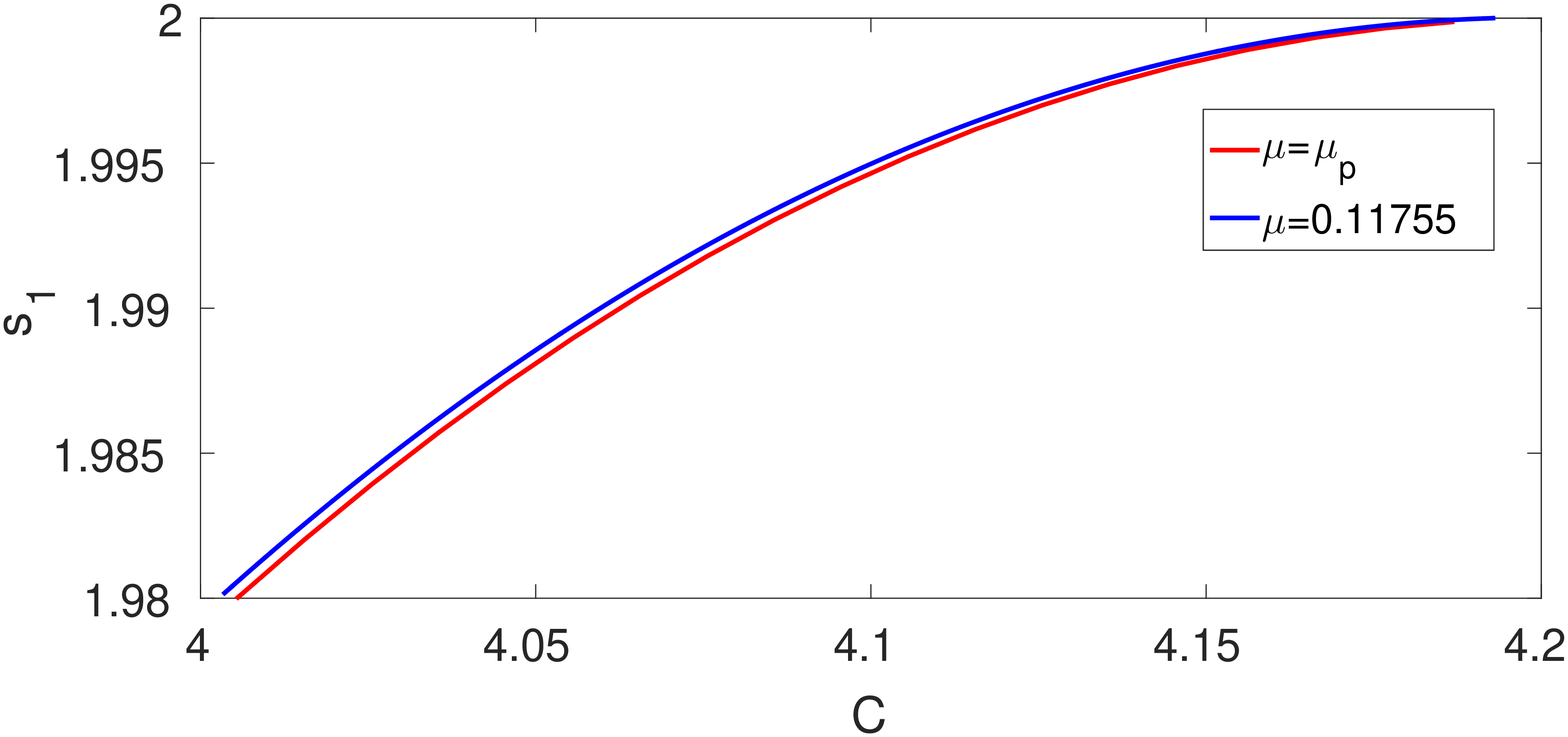}\\
  \includegraphics[width=2.5in,height=1.8in]{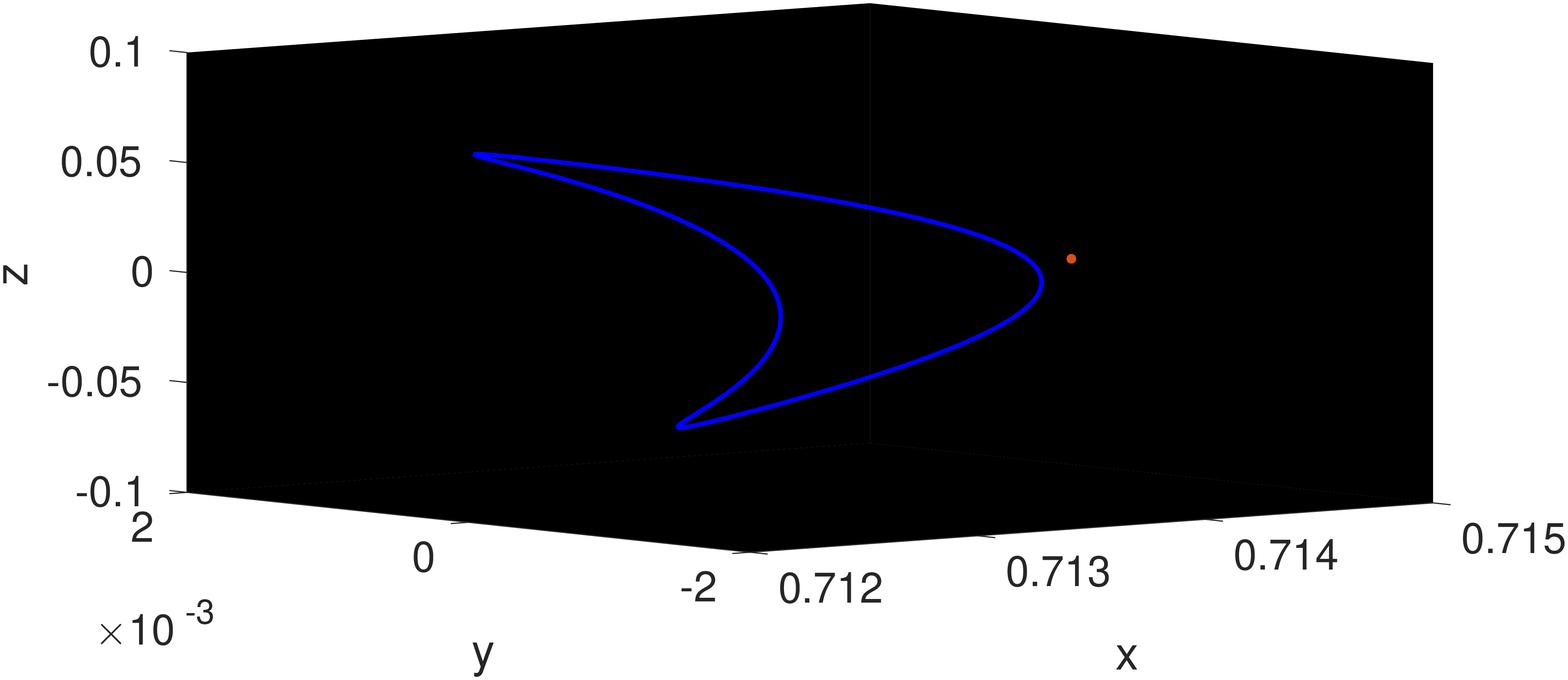}& \includegraphics[width=2.5in,height=1.8in]{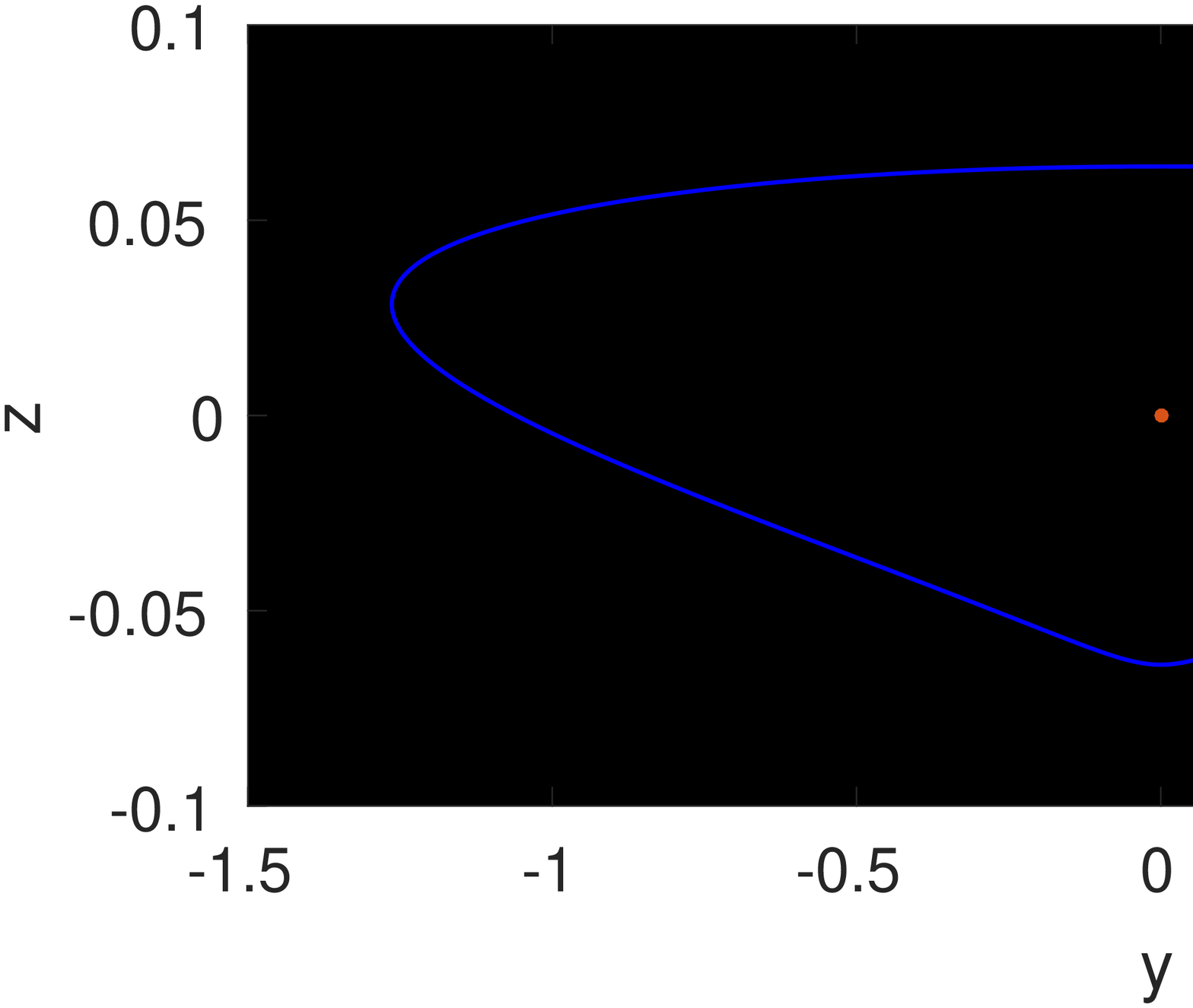}
\end{tabular}
\caption{First row. Projections of the final parts of characteristic curves for the family of Halo orbits on the planes $(C,z_0)$ and $(C,s_1)$ respectively. Second row. Two perspectives of the orbit of bifurcation for the Halo orbits for $\mu=0.12$ and $C\approx4.1826495$. After this bifurcation, another loop appears in the orbit to evolve in a eight-shaped orbit.}
\label{halofigures2}
\end{figure}

Therefore, the numerical explorations support the conjecture made in Section \ref{resonances}. Although the mechanism for generating the Halo orbits is not starting at planar Lyapunov critical orbits for the cases when $\mu>\mu_p$, we have decided to keep the name Halo  since they are diffeomorphic to circles just before the bifurcations where the family starts and ends. In Figure \ref{halofigures} we can observe the geometry of Halo periodic orbits for a couple of values of $\mu$. In Tables \ref{Criticalhalo1case1} - \ref{Criticalhalo1case4} we have shown data of critical orbits for the Halo family.

\begin{figure}
  \centering
\begin{tabular}{cc}
  \includegraphics[width=2.0in,height=2.0in]{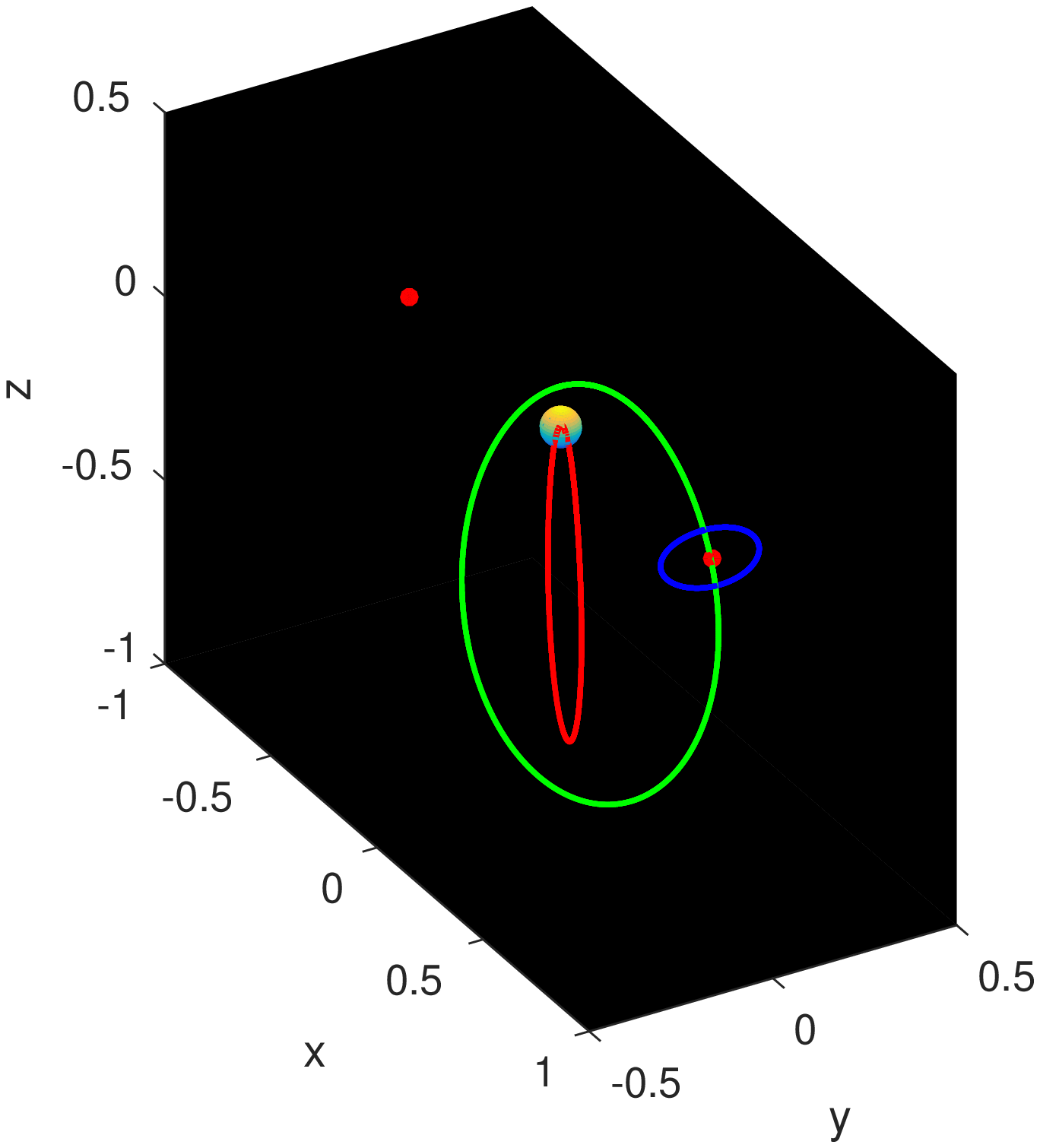}& \includegraphics[width=2.0in,height=2.0in]{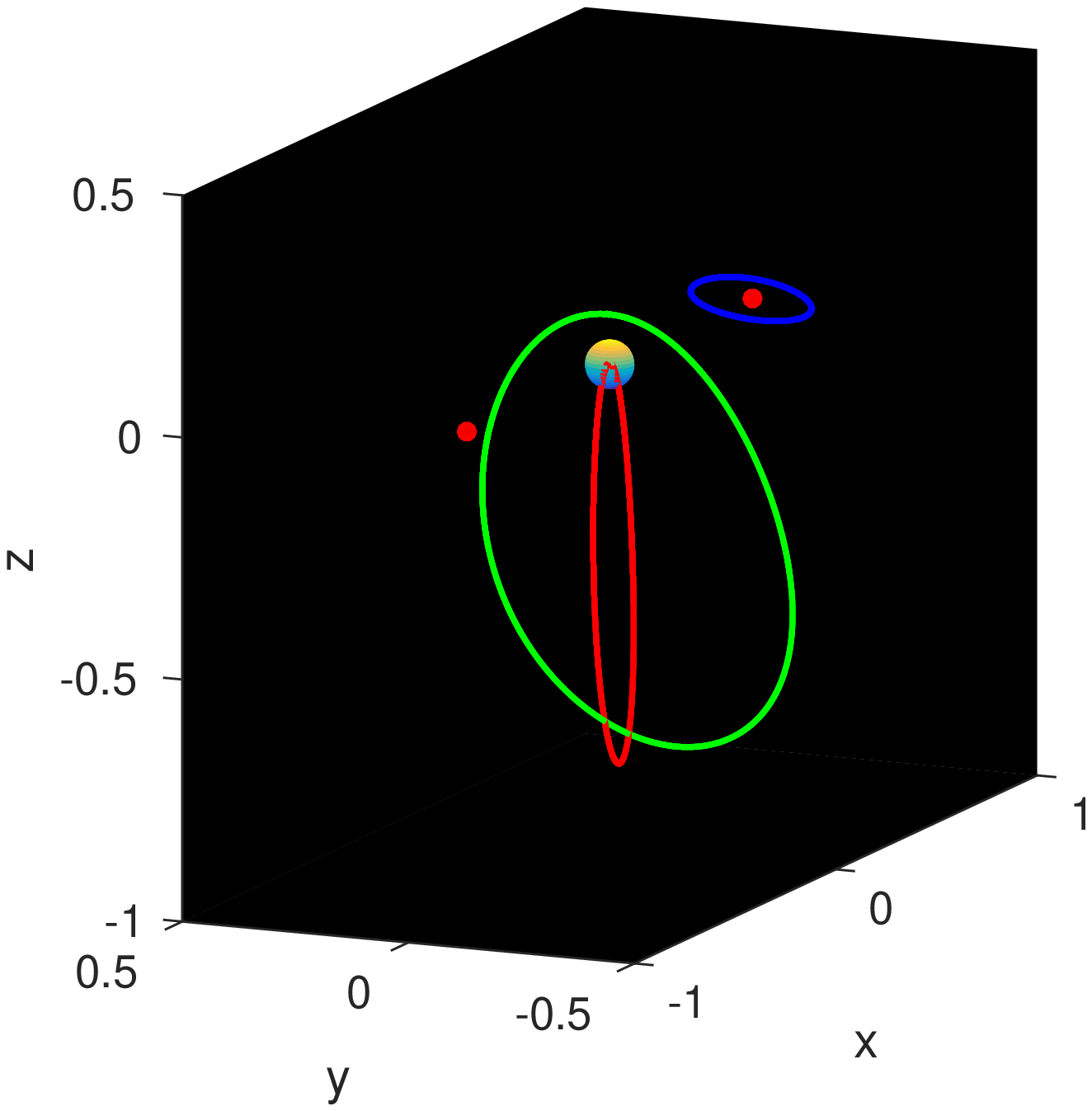}\\
  \includegraphics[width=2.0in,height=2.0in]{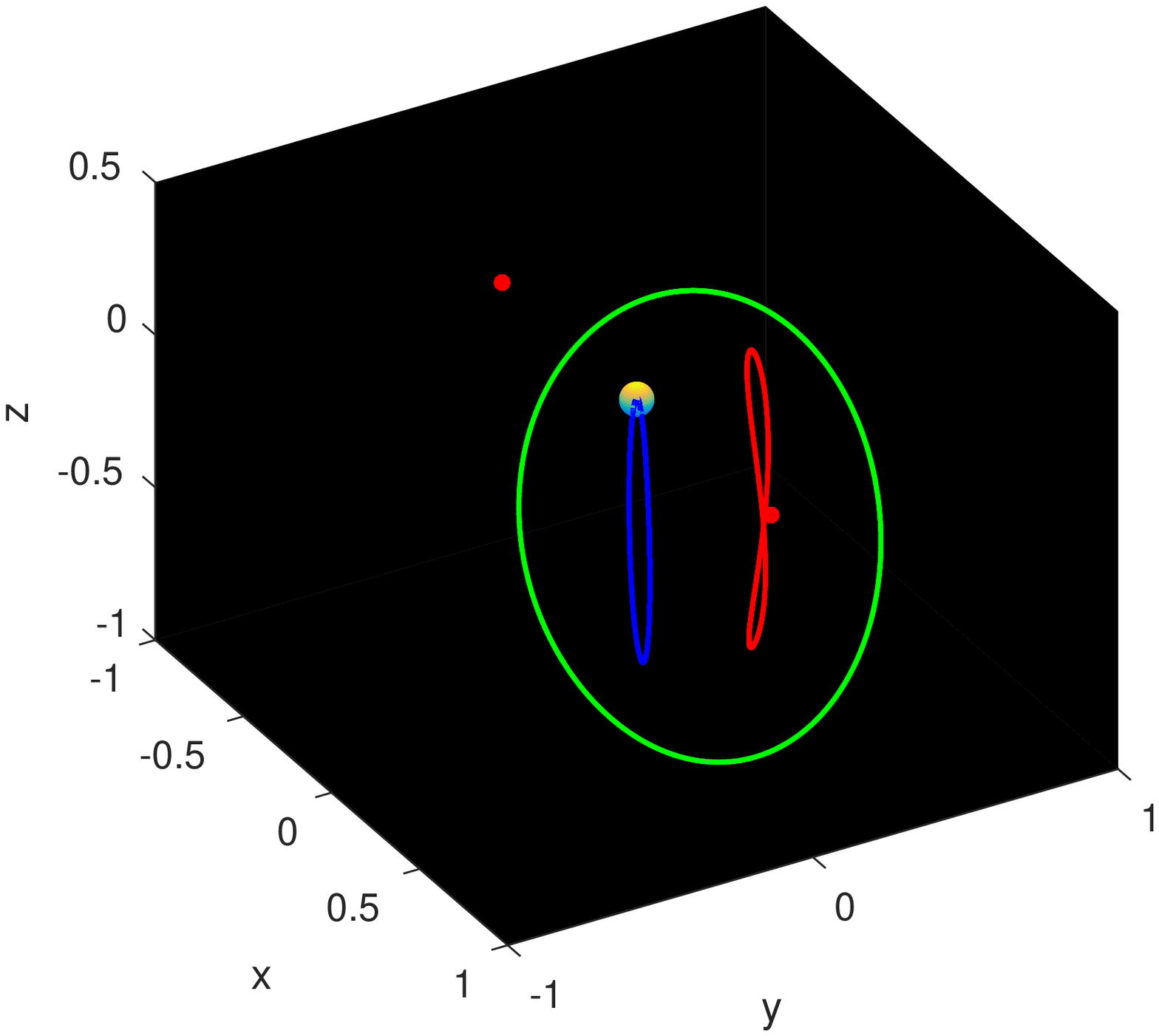}& \includegraphics[width=2.0in,height=2.0in]{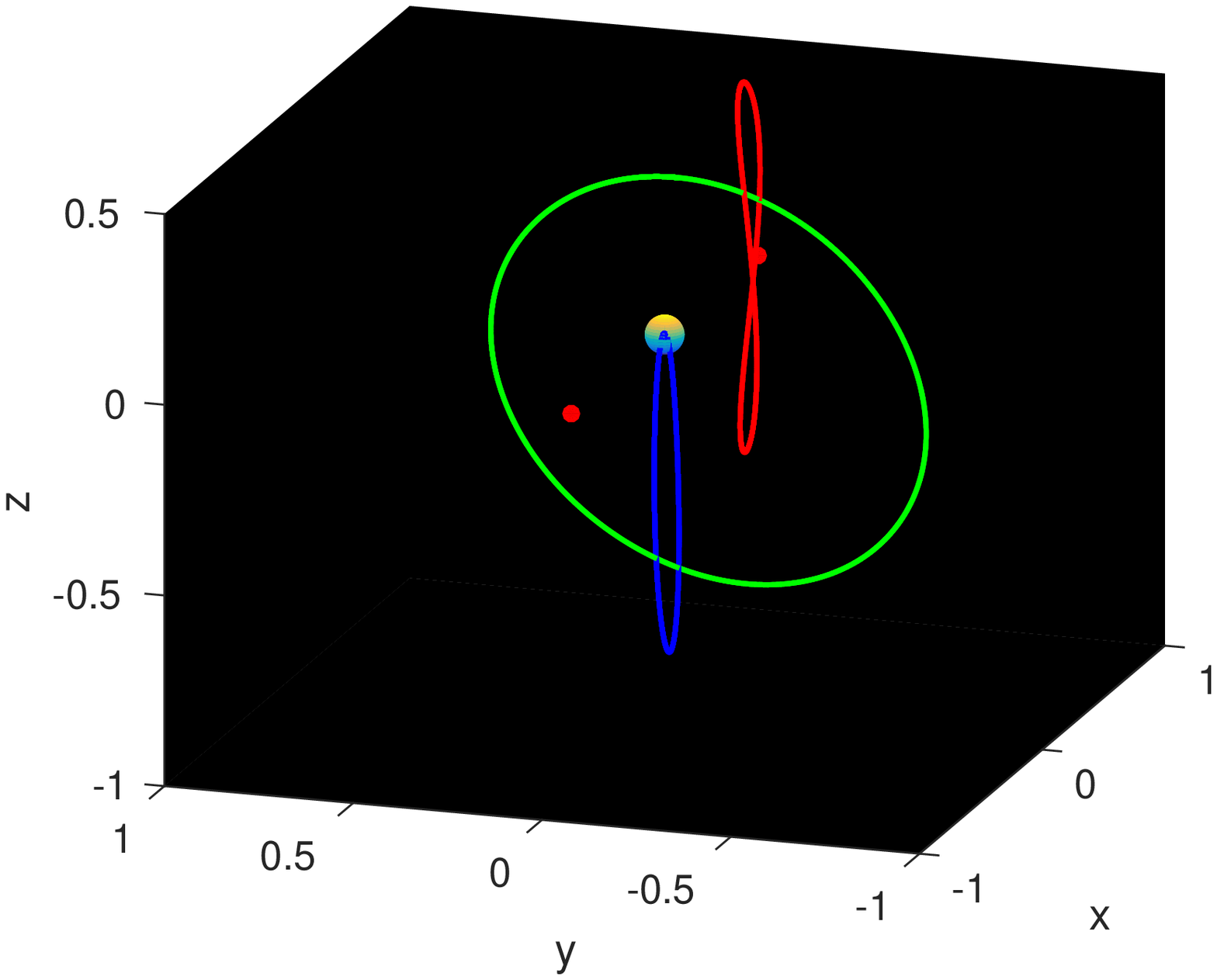}
\end{tabular}
\caption{Different views for orbits within Halo family. First row, evolution from a planar critical orbit (blue) to an ejection-collision orbit (red) for $\mu=0.1$. Second row, evolution from a vertical critical orbit (red) to an ejection-collision orbit (blue) for $\mu=0.5$. Red dots correspond to equilibrium points.}
\label{halofigures}
\end{figure}

\begin{table}[ht]
\caption{Initial conditions for critical orbits belonging to Halo family for $\mu=0.00095$. First row corresponds to a bifurcation with the planar family of Liapunov orbits.} 
\centering 
\resizebox{14cm}{!} {
\begin{tabular}{c c c c c c c}
\hline\hline 
$s_1$ & $s_2$ & $x_{0}$ & $z_{0}$ & $\dot{y}_{0}$ & $T/2$ & $C$\\ [0.5ex] 
\hline 
2 & 1730.4 & 0.581672037294313 & 0.0 & 0.668336524962183 & 1.541431516853591 & 4.005994\\
2 & 34.98 & 0.104820295485459  & 0.315162159626700 & 2.148218466217613  & 1.355379140358816 & 1.340394\\
2 & 1.328 & 0.101840495133734  & 0.313414776304381 & 2.161812236323527   & 1.351672367357241 &  1.328394 \\
1.027 & 2 &  0.011703595185131 &  0.207968429015019 & 2.913689930254455   & 1.146724290330344 & 1.069223\\
2 & 0.422 &  -0.000648771717150  &  0.166131017074055 & 3.303856999073573   & 1 .065009799310460 & 1.095530 \\
3.259 & 2 &  -0.000077584280011  &   0.002946820658802 & 26.014809437875140    & 0.723938391949591 &1.692140 \\
 [1ex] 
\hline 
\end{tabular}
}
\label{Criticalhalo1case1} 
\end{table}

\begin{table}[ht]
\caption{Initial conditions for critical orbits belonging to Halo family for $\mu=0.03$. First row corresponds to a bifurcation with the planar family of Liapunov orbits.} 
\centering 
\resizebox{14cm}{!} {
\begin{tabular}{c c c c c c c}
\hline\hline 
$s_1$ & $s_2$ & $x_{0}$ & $z_{0}$ & $\dot{y}_{0}$ & $T/2$ & $C$\\ [0.5ex] 
\hline 
2 & 1794.5 & 0.601287913380110 &   0.0 &  0.575263625144359 &      1.561152709206505  & 4.055694 \\
2 & 35.9  & 0.107315014557271 &   0.321820367624817 &  2.112727941575960 &      1.370114315788723  &1.362094 \\
2 & 34.11 & 0.103591787584037 &  0.319646959222008 &  2.129398767787492 &      1.365497793441479  & 1.347094\\
0.90 & 2 &  0.010003086858506 &  0.208613778483544 &  2.908315619968743  &     1.149607605949286 & 1.074566\\
2 & 0.453 &  -0.002730254287789 &  0.161458136703433 &  3.354429957717300 &      1.057050430881994  &1.107094 \\
3.154 & 2 &   -0.000080090720024 & 0.002937297832600 &  26.056691025195214 &     0.723622718183423 & 1.693794 \\
 [1ex] 
\hline 
\end{tabular}
}
\label{Criticalhalo1case2} 
\end{table}

\begin{table}[ht]
\caption{Initial conditions for critical orbits belonging to Halo family for $\mu=0.1$. First row corresponds to a  bifurcation with the planar family of Liapunov orbits.} 
\centering 
\resizebox{14cm}{!} {
\begin{tabular}{c c c c c c c}
\hline\hline 
$s_1$ & $s_2$ & $x_{0}$ & $z_{0}$ & $\dot{y}_{0}$ & $T/2$ & $C$\\ [0.5ex] 
\hline 
2 & 1971.2 & 0.669327209253065 &  0.0 &   0.243998312372139& 1.611289914305902   & 4.174694\\
2 & 37.89  & 0.113892766509675 &  0.338770610358851 &   2.025868226362990& 1.407431377783558  &1.413094  \\
2 & 34.23 &  0.105884144100618 &  0.334078565282522 &   2.060417783160357 &     1.397508079271822 &  1.381094 \\
0.639 & 2 &   0.006211217319152 &  0.210789666424723 &   2.890528624397537 &     1.158270994928659 & 1.084534  \\
2 & 0.598  & -0.006822974119050 &  0.150495590867728  &  3.481404287348189 &     1.038606170016169  & 1.133094 \\
2.915 & 2 & -0.000087210031082 &  0.002945451413095  &  26.019644506947962 &    0.723735450159685 & 1.693724  \\
 [1ex] 
\hline 
\end{tabular}
}
\label{Criticalhalo1case3} 
\end{table}

\begin{table}[ht]
\caption{Initial conditions for critical orbits belonging to Halo family for $\mu=0.5$. First row corresponds to a bifurcation with the vertical family.} 
\centering 
\resizebox{14cm}{!} {
\begin{tabular}{c c c c c c c}
\hline\hline 
$s_1$ & $s_2$ & $x_{0}$ & $z_{0}$ & $\dot{y}_{0}$ & $T/2$ & $C$\\ [0.5ex] 
\hline 
2 & 625.4  & 0.648277423604457  &  0.483515677156544  &   0.087533359865646  &    1.800474774189599  & 3.177144  \\
2 & 42.36  & 0.158469155949397 &  0.426750603593817 &   1.627279340761100  &    1.591412388621520  &1.619794   \\
2 & 22.10 &  0.090191445216378 &  0.381064472116173 &   1.897063403336006  &    1.502888070606410  & 1.381594\\
0.190 & 2 &   -0.008805949009954 &  0.216328361973418 &   2.846214515756042 &     1.192438243600371  &1.089994   \\
2 & 1.367  & -0.013548716819445 & 0.105436491714333  &  4.193038116150593  &    0.962899180087805  &1.221794  \\
2.474 & 2 & -0.000110820946092 &  0.002927606844806 &   26.096189666644001 &    0.732843158560505  &1.651794 \\
 [1ex] 
\hline 
\end{tabular}
}
\label{Criticalhalo1case4} 
\end{table}

\subsubsection{The Lane family}
\label{Lanefamily}
The Lane family, named as $a2v$ in \cite{Michalodimitrakis}, is  comprised of periodic orbits with symmetry of type \textbf{I}. There exist two analogous families for the collinear equilibrium points of the R3BP also named \textit{Lanes} \cite{Gomez}. We use the same name in this work in order to compare our results with the ones of the R3BP. The numerical explorations in our problem show that these families are indeed lanes, or bridges, between the planar Lyapunov orbits and the vertical family that starts and ends in bifurcating periodic orbits of both families. Because of the symmetries of the equations of this problem, there exist two symmetric families which are analogous to the \textit{Lanes} families of the R3BP. Both are related to each other by the symmetry and consequently the behaviour of both families is identical.  Therefore, the whole family is a two-lane family that starts in a critical planar Lyapunov orbit and ends in a critical orbit of the vertical family. In view of the observations already mentioned, it will be enough to explore one part, namely \textit{Lane 1}.
\newline

Considering several values of the mass parameter $\mu$, we found that Lane family is comprised of unstable orbits and that some of them are critical. In Table \ref{Criticallane1} we show the critical orbits detected for some values of $\mu$. Following the evolution of the family as $\mu$ increases, we observe that the characteristic curves tend to reduce its size in the space $(C,x_0,v_{y0})$, in the sense that bifurcating orbits in the vertical and planar families approach to each other. For instance, for $\mu=0.00095$ the family changes monotonically from $C=1.245898$ (bifurcation with the planar family) to $C=0.535298$ (bifurcation with the vertical family), for $\mu=0.03$ it goes from $C=1.793458$ to $C= 1.225758$, for $\mu=0.1$ it goes from $C=3.589648$ to $C= 3.429161$, and for $\mu=0.114087$ it goes from $C=4.157287$ to $C=4.144146$.
\newline

According to the observation that the value of $C$ for the planar bifurcating orbit increases as $\mu$ does, and that the value of $C_1$, defined as the value of $C$ at $L_1$, decreases monotonically with $\mu$, we conjecture that there exists a value of $\mu$ for which the Lane family does not exist any more. Actually, such a value of $\mu$ should be close to the value $\mu_p$ mentioned in Section \ref{haloorbits} since it was observed numerically that, for instance, for $\mu=0.114087$ the Lane family starts in a planar critical orbit at $C=4.157287$ and ends in a vertical orbit at $C=4.144146$.  An analytic estimation of $\mu_p$ and its relation with the disappearance of the Lane family requires further analytical study that it will be shown in a forthcoming work. The evolution of the orbits can be seen in Figure \ref{lane1figures}.

\begin{table}[ht]
\caption{Initial conditions for critical orbits belonging to family Lane $1$. First and second rows correspond to $\mu=0.00095$, third to sixth rows to $\mu=0.03$ and seventh to thirteenth rows to $\mu=0.1$.} 
\centering 
\resizebox{14cm}{!} {
\begin{tabular}{c c c c c c c}
\hline\hline 
$s_1$ & $s_2$ & $x_{0}$ & $\dot{z}_{0}$ & $\dot{y}_{0}$ & $T/2$ & $C$\\ [0.5ex] 
\hline 
2 & 336.8  & 0.947776279601374  &  0.0  &    -1.886060615200088  &    2.058578156228173 & 1.245898 \\
2 & 435.8 & 0.542204288997138  &  -1.920518788500280  &    -0.588459343018466  &    2.116676844841206 &0.535298 \\
2 & 458.6 & 0.913421990785547 &  0.0  &   -1.686195111313342   &     1.925046023338045 & 1.793458\\
2 & 526.1 & 0.580982542362875 &  -1.707368794292224  &   -0.538047463781664   &     1.963885549526655 & 1.227858\\
2 & 526.4 & 0.566697194548073  &  -1.739117162480509   &    -0.469643626495060   &     1.964018694781028 & 1.226058\\
2 & 526.4 & 0.561982518993931  &  -1.749289926100880  &     -0.446518323542460   &     1.964065581758477 & 1.225758\\
2 & 1440.4 & 0.831058677841325  &  0.0  &     -0.859102486479120   &     1.668621755939243 & 3.589648 \\
2 &  1424.2 & 0.675394306861832  &  -0.887333573679594  &     -0.114149626052293   &     1.674669702607736 & 3.429690 \\
2 &   1424.2 & 0.675347465236596  &   -0.887386533168669   &     -0.113879512746934   &     1.674669818837174  & 3.429687\\
2 &   1424.1 & 0.675336297671779  &   -0.887399811839130   &     -0.113815599195322  &      1.674669958473229 & 3.429685 \\
2 &   1424.2 & 0.675203060436454  &   -0.887549300177709   &     -0.113046526080681  &     1.674670164552536 &3.429678 \\
2 &   1423.9 & 0.670676991502061 &   -0.892480645821746   &     -0.086918881588729  &     1.674707591490506 & 3.429170 \\
2 &   1423.9 & 0.670682592737372 &   -0.892478843883959    &     -0.086954600007957  &     1.674708199795873 & 3.429163 \\
[1ex] 
\hline 
\end{tabular}
}
\label{Criticallane1} 
\end{table}

\begin{figure}
  \centering
\begin{tabular}{cc}
  \includegraphics[width=2.0in,height=2.0in]{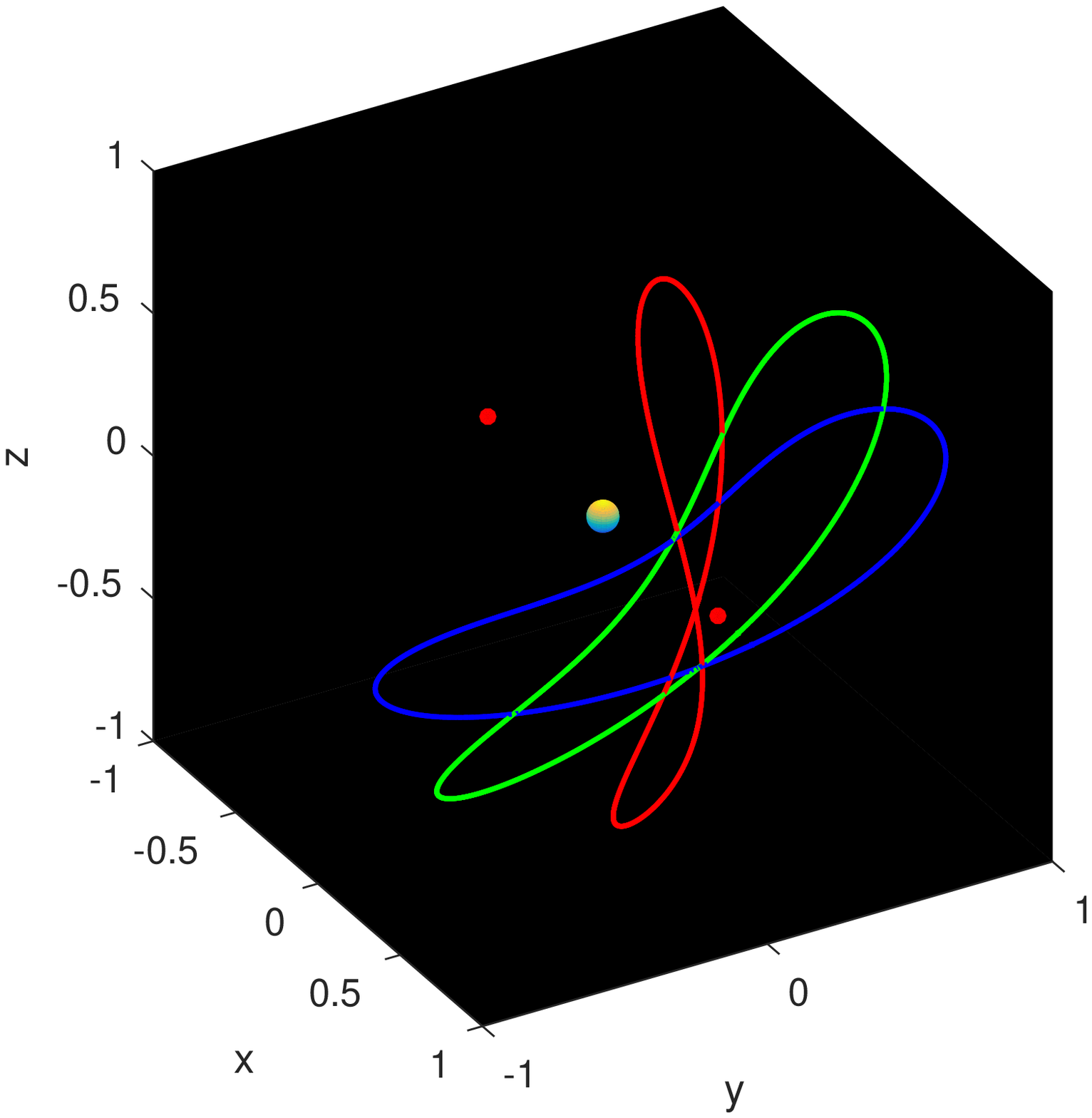}& \includegraphics[width=2.0in,height=2.0in]{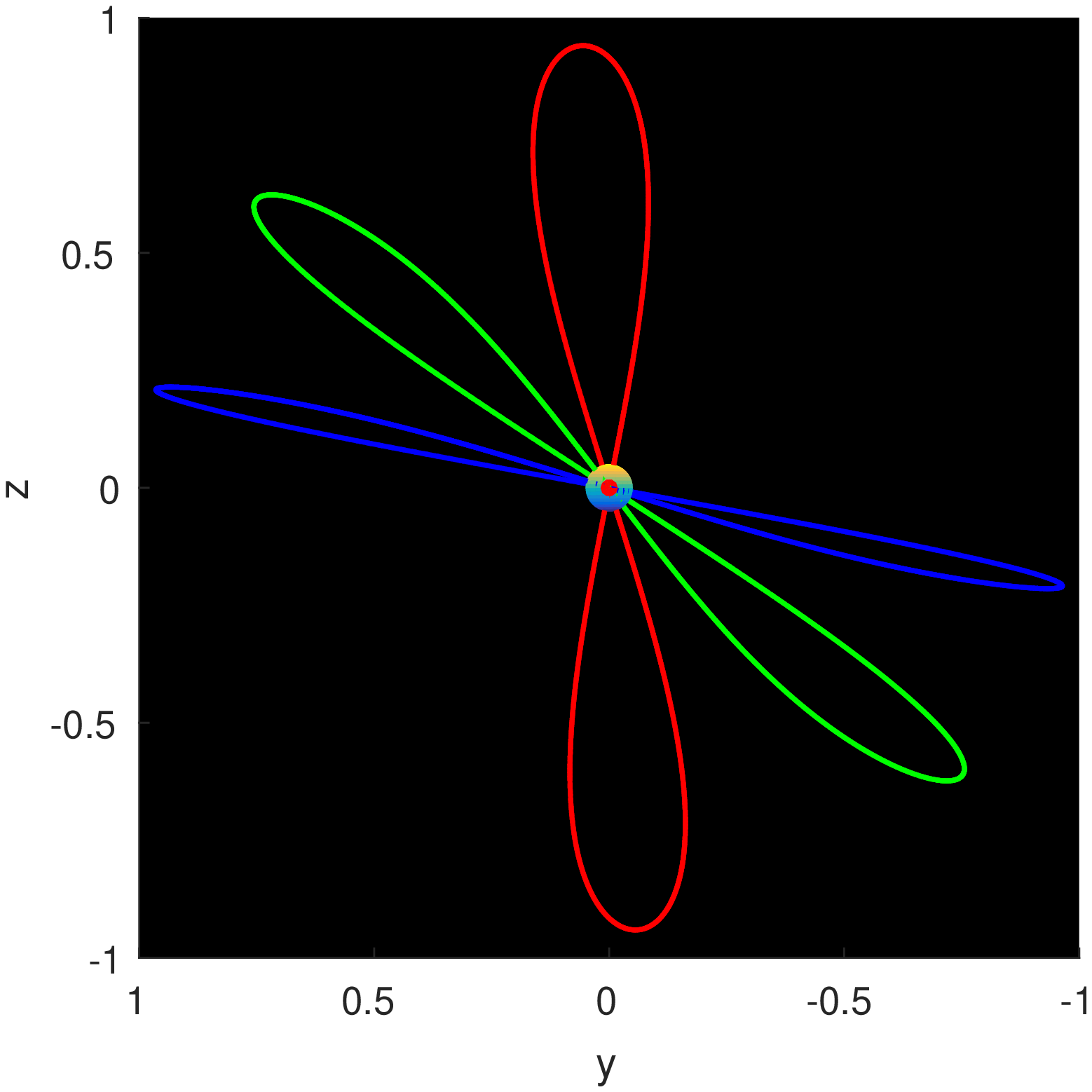}\\
 \includegraphics[width=2.0in,height=2.0in]{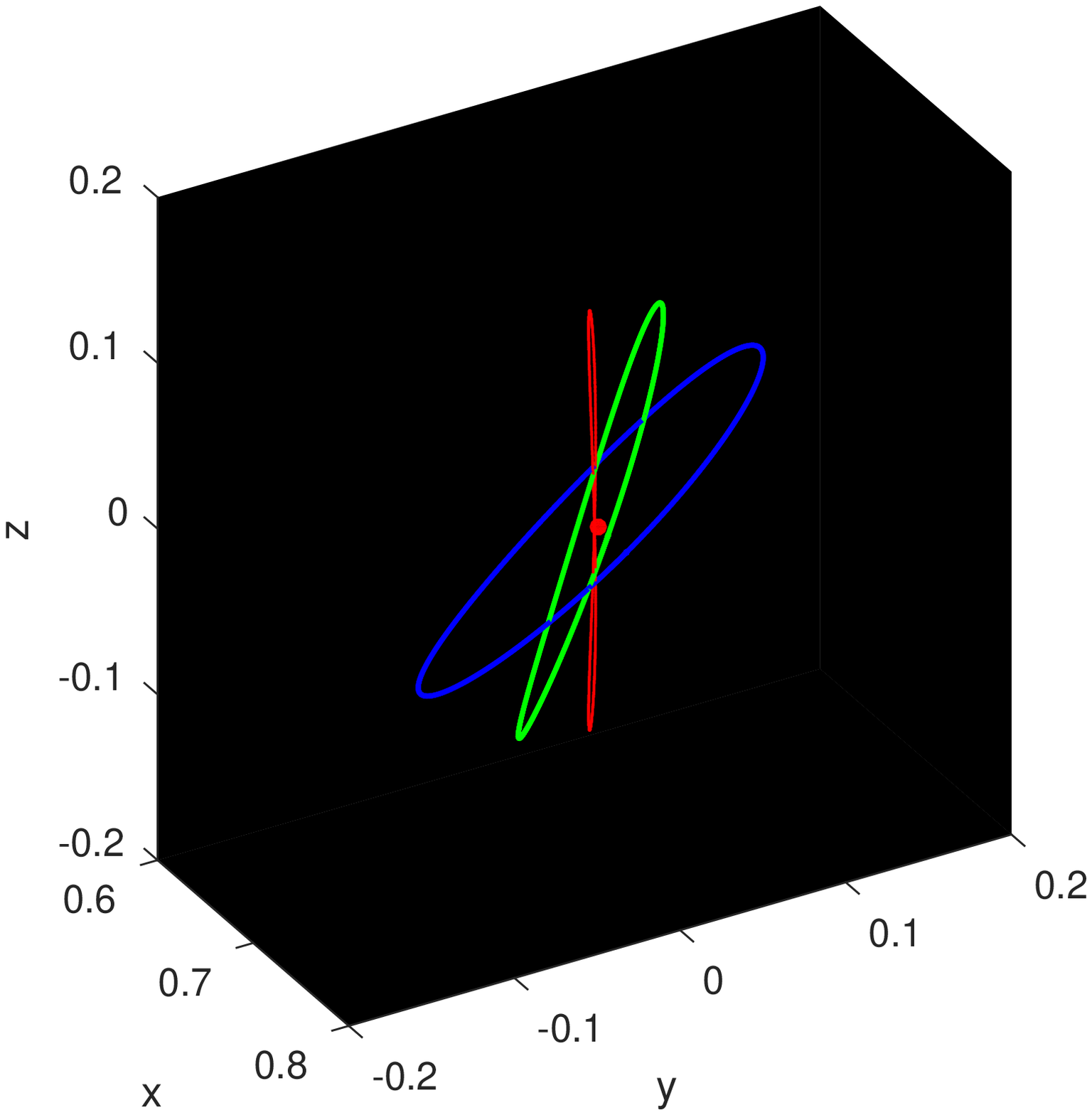}& \includegraphics[width=2.0in,height=2.0in]{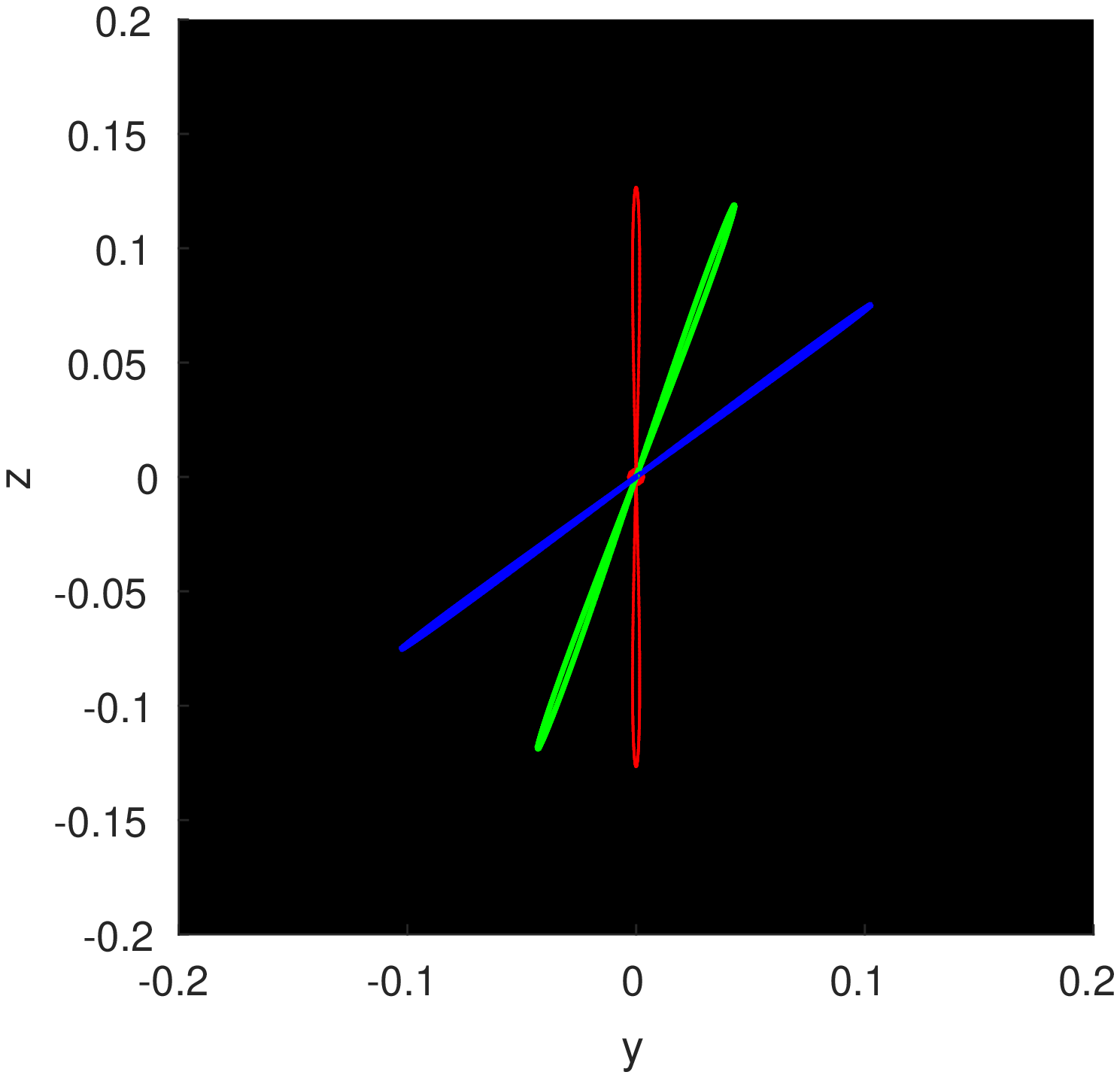}
\end{tabular}
\caption{Different views for periodic orbits within family Lane $1$, for $\mu=0.03$ and $\mu=0.114087$, are given in first and second row, respectively. Red dots correspond to equilibrium points.}
\label{lane1figures}
\end{figure}

\subsubsection{The family a3v}
In the numerical explorations we realized that the behaviour of the family $a3v$ is analogous to the one of the classical H3BP. The family departs from a period doubling bifurcation with a planar Liapunov orbit and as the constant $C$ decreases monotonically the inclination of the orbits increases and tend to collision with the tertiary body. Such behaviour, and the instability of the periodic orbits, hold for all $\mu \in [0,0.5]$. There are critical orbits whose initial conditions are given in the Table \ref{Criticala3v} for some values of $\mu$.  In Figure \ref{figures_a3v_case1} we can observe three periodic orbits for $\mu=0.5$. The members of this family have symmetry type \textbf{III}.

\begin{table}[ht]
\caption{Initial conditions for critical orbits belonging to family $a3v$. First, second and third rows correspond to $\mu=0.00095$, fourth and fifth rows to $\mu=0.2$ and sixth and seventh rows to $\mu=0.5$.} 
\centering 
\resizebox{14cm}{!} {
\begin{tabular}{c c c c c c c}
\hline\hline 
$s_1$ & $s_2$ & $x_{0}$ & $\dot{z}_{0}$ & $\dot{y}_{0}$ & $T/4$ & $C$\\ [0.5ex] 
\hline 
2 & 22112  & 1.215964201848766 &  0.0 &  -2.467775191236130  &     2.823692692996263 & -0.012582\\
2 & 27826  & 1.197949332150484 &  0.495799627813836 &   -2.431535619465253  &     2.816015129308572 & -0.186482\\
2 & 35184  & 1.178152075897276 &  0.715439512365338 &   -2.391579755999035  &    2.808143136235180 & -0.372773\\
2 &  565.6  & 1.060636666432214 &  0.0 &   -1.770290287017953  &     2.613402863642579 & 1.655977\\
2 &  573  & 1.046071830822358 &  0.433920959256934 &   -1.747246278450896  &     2.604579716310594 & 1.495787\\
2 &  15582  & 1.049806859467960  &  0.0 &   -1.415400237548225  &     2.567385833417275 & 2.381467\\
2 &  5128.2 & 1.022159719609811   &  0.587678513478811  &   -1.392074692111009  &     2.545407711915551 & 2.024227\\
\hline 
\end{tabular}
}
\label{Criticala3v} 
\end{table}

\begin{figure}
  \centering
\begin{tabular}{cc}
  \includegraphics[width=2.0in,height=2.0in]{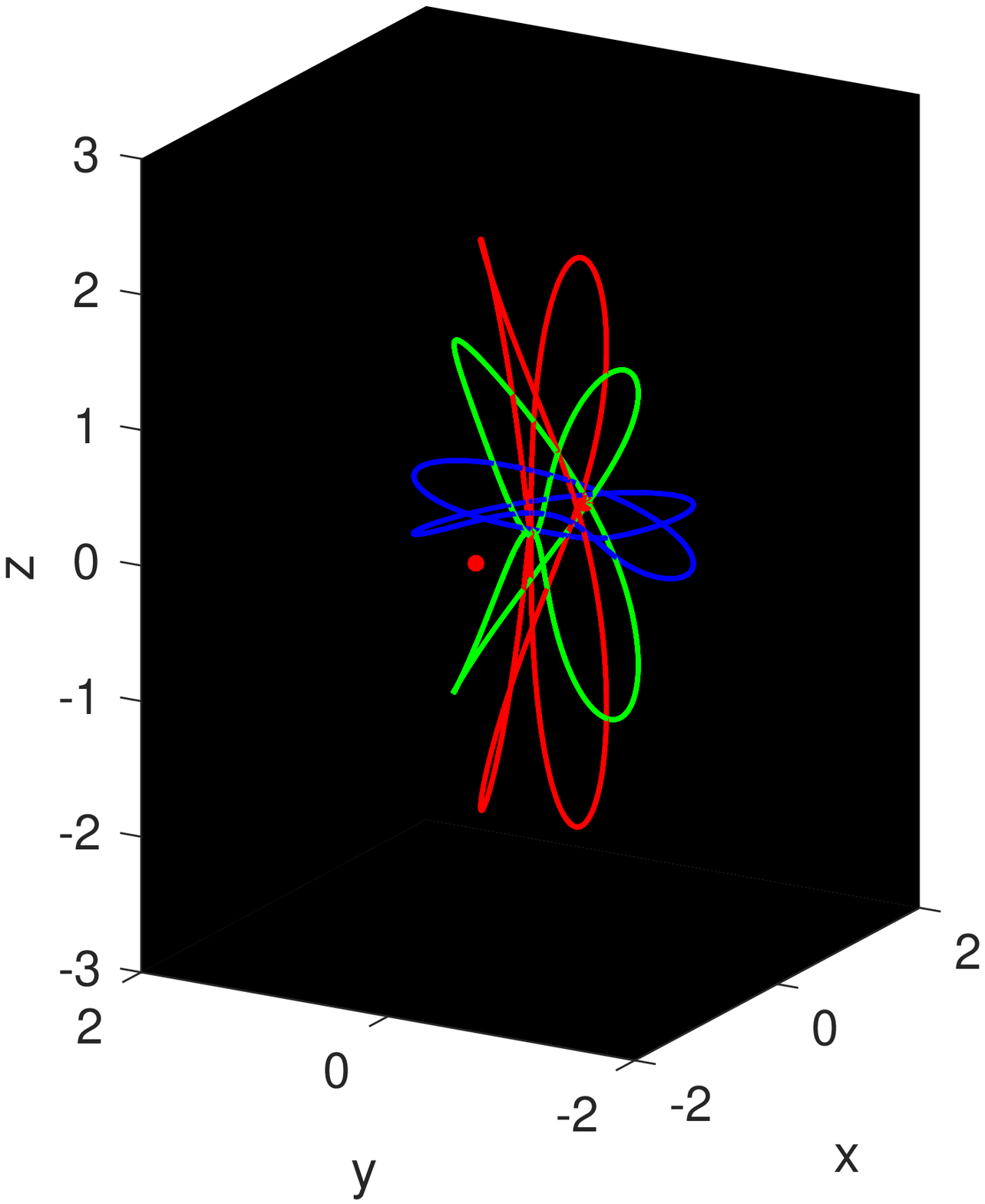}& \includegraphics[width=2.0in,height=2.0in]{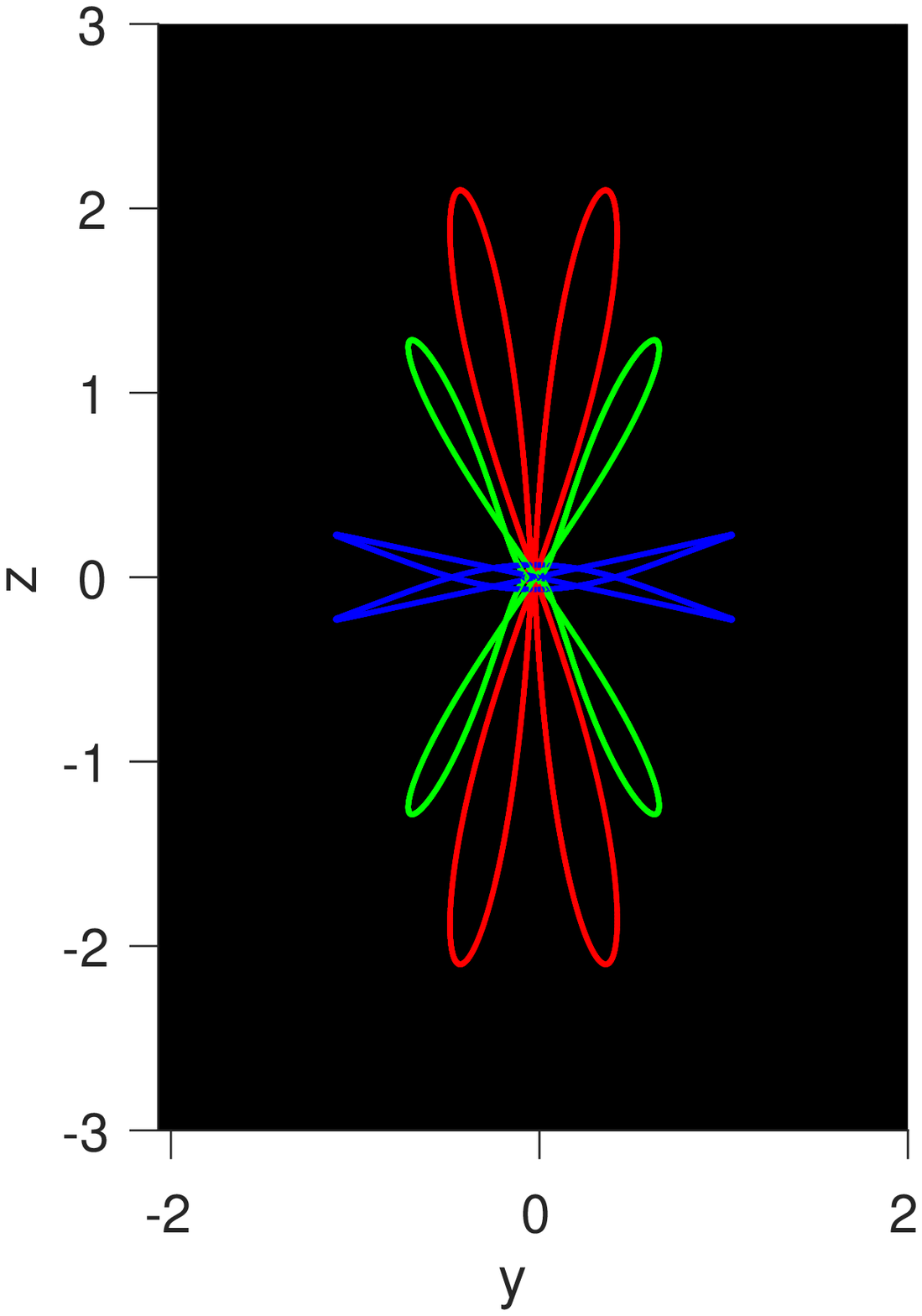}\\
 \includegraphics[width=2.0in,height=2.0in]{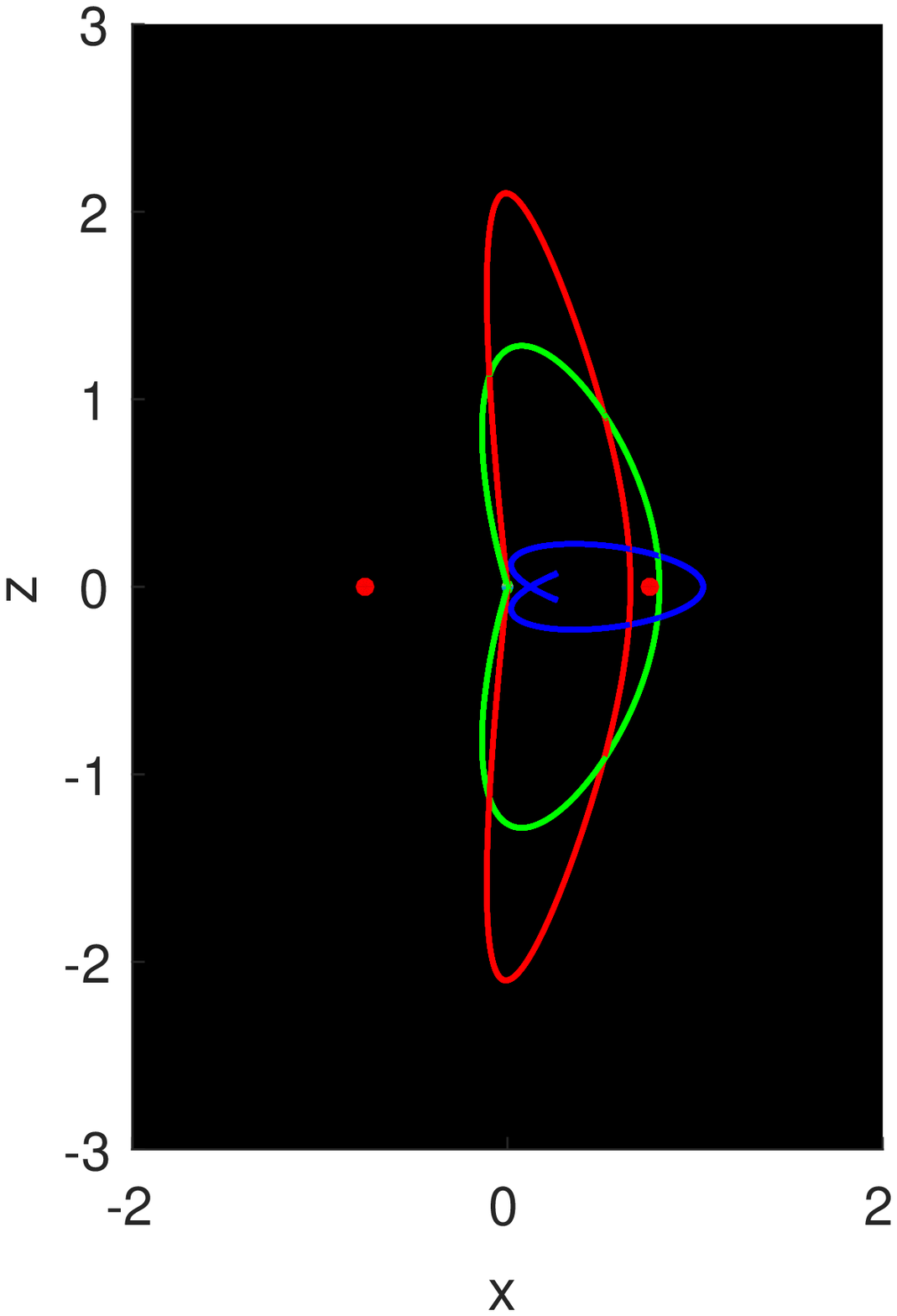}& \includegraphics[width=2.0in,height=2.0in]{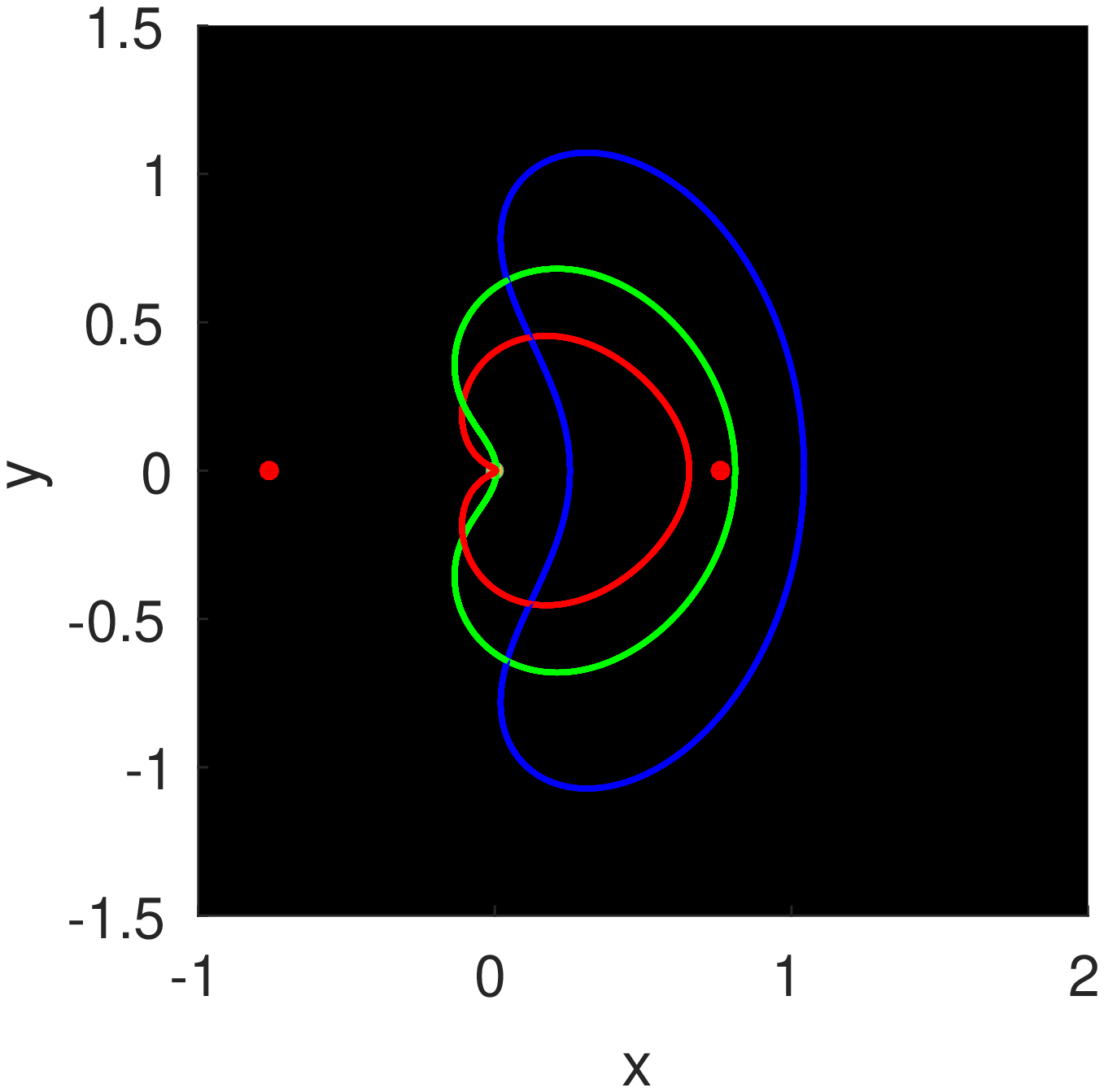}
\end{tabular}
\caption{Different views for periodic orbits within family $a3v$, for $\mu=0.5$. Red dots correspond to equilibrium points.}
\label{figures_a3v_case1}
\end{figure}

\section{Conclusions and perspectives}\label{sec:conclusions}
We have developed a first study of the dynamics of the three-dimensional H4BP focusing on the simplest invariant objects: equilibrium points and symmetric periodic orbits.
\newline

New symmetric periodic orbits around the tertiary, and the ones emanating from the linearized system around $L_1$ and $L_2$, were found. These orbits were computed by means of solving numerically specific boundary value problems associated to the symmetries of the equations of motion.
\newline

The high-precision numerical continuation in the Jacobi constant $C$, for several values of the mass parameter $\mu$, reveals a very rich structure of the basic families of periodic orbits and suggests that there are a plethora of bifurcations among the periodic orbits.
\newline

The conjectures regarding the existence of bifurcations, families tending to ejection-collision orbits, and properties of stability, demand a further study. In a near-future work we will continue the analysis of the H4BP, in order to obtain a deeper understanding of this problem and its connections with other models as it was mentioned in \cite{Chicone}.

\newpage

\section*{Acknowledgments}\label{sec:acknowledgments}
The second author is pleased to acknowledge the support of Asociaci\'on Mexicana de Cultura A.C., the National System of Researchers (SNI), and Conacyt-M\'exico Project A1S10112.

\section*{Declaration of competing interest}

The authors declare that they have no known competing financial interests or personal relationships that could have
appeared to influence the work reported in this paper.

\section*{Funding sources}

This research did not receive any specific grant from funding agencies in the public, commercial, or not-for-profit sectors.


\begin{thebibliography}{99}

\bibitem{Stromgren} Str\"{o}mgren, E.; Connaissance actuelle des orbites dans le probleme des trois Corps. Bull. Astronom. \textbf{9}(2), 87--130 (1933).
\bibitem{HenI} H\'enon, M.: Exploration num\'erique du probl\'eme restreint I. Masses \'egales, Orbites p\'eriodiques. Ann. Astrophysics \textbf{28}, 499--511 (1965)
\bibitem{HenII} H\'enon, M.: Exploration num\'erique du probl\'eme restreint II. Masses \'egales, stabilit\'e des orbites p\'eriodiques. Ann. Astrophysics \textbf{28}, 992--1007 (1965).
\bibitem{Sz} Szebehely V.; Theory of orbits. Academic Press, New York (1967).
\bibitem{PapaI} Baltagiannis, A.N., Papadakis, K.E.: Families of periodic orbits in the restricted four-body problem. Astrophys. Space Sci. \textbf{336}, 357--367 (2011)
\bibitem{BurgosDelgado} Burgos-Garc\'ia, J., Delgado, J.; Periodic orbits in the restricted four-body problem with two equal masses. Astrophys Space Sci \textbf{345}, 247–263 (2013). https://doi.org/10.1007/s10509-012-1118-2
\bibitem{BurgosLessardJames} Burgos-Garc\'ia, J., Lessard, J.,  James, J.D.M. Spatial periodic orbits in the equilateral circular restricted four-body problem: computer-assisted proofs of existence. Celest Mech Dyn Astr \textbf{131}, 2 (2019). https://doi.org/10.1007/s10569-018-9879-8
\bibitem{BurgosGidea} Burgos-Garc\'ia, J., Gidea, M.;  Hill's approximation in a restricted four body problem. Celestial Mechanics and Dynamical Astronomy,. \textbf{122}, 117–141 (2015). https://doi.org/10.1007/s10569-015-9612-9.
\bibitem{Burgos} Burgos-Garc\'ia, J.; Families of periodic orbits in the planar Hill’s four-body problem. Astrophys Space Sci. \textbf{361}, 353 (2016). https://doi.org/10.1007/s10509-016-2943-5.
\bibitem{Belbruno} Belbruno E. and  Miller J.; A Ballistic Lunar Capture Trajectory for the Japanese Spacecraft Hiten. Technical report. IOM 312/90.4-1371-EAB. Jet Propulsion Laboratory, Pasadena, CA, 1990.
\bibitem{Broucke} Broucke, R. A.; Periodic orbits in the restricted three--body problem with earth-moon masses. Technical Report, JPL. (1968).
\bibitem{Koon} Koon, W.S., Lo, M.W., Marsden, J.E., Ross, S.D.; Low energy transfer to the moon. Celestial Mech. Dyn. Astron. \textbf{81}, pp. 63–73 (2001)
\bibitem{Schwarz} Schwarz, R., Dvorak, R., Suli, A., Erdi, B.: Survey of the stability region of hypothetical habitable trojan
planets. Astron. Astrophys. \textbf{474}, 1023 (2007).
\bibitem{Chicone} Chicone, C., Mashhoon, B., Retzloff, D. G.: Chaos in the Hill system.  Helvetica Physica Acta,
\textbf{72(2)} 123--157, (1999).
\bibitem{BurgosDelgadoII} Burgos-Garc\'ia, J., Delgado, J.; On the ``Blue sky
catastrophe'' termination in the restricted four body problem. Celest Mech Dyn Astr \textbf{117}, 113–136 (2013). https://doi.org/10.1007/s10569-013-9498-3.
\bibitem{Leandro} Leandro, E.S.G.: On the central configurations of the planar restricted four-body problem. J. Differ. Equ. \textbf{226(1)}, 323–351 (2006)
\bibitem{KepleyII} Kepley S, Mireles James, J.D.: Homoclinic dynamics in a restricted four-body problem: transverse connections for the saddle-focus equilibrium solution set. Celest Mech Dyn Astr \textbf{131}, 13 (2019).
\bibitem{Birkhoff} Birkhoff, G. D.; Dynamical Systems, American Mathematical Society, Coll. Pub. 9 (1927).
\bibitem{Vogelaere} Vogelaere, R.; in Contributions to the Theory of Non-linear Oscillations, Princeton Univ. Press, vol 4, 53--, (1956).
\bibitem{Jimenez2003} Jim\'enez-Lara, L., Pi\~na, E.: The three-body problem with an inverse square law potential. J. Math. Phys. \textbf{44}, 4078--4089 (2003).
\bibitem{Jimenez1990} Jimenez-Lara, L., Pi\~na, E.: Periodic orbits of an electric charge in a magnetic dipole field. Celest Mech Dyn Astr \textbf{49}, pp. 327-345 (1990).
\bibitem{Chavoya1989} Chavoya-Aceves, O., Pi\~na, E.: Symmetry lines of the dynamics of a heavy rigid body with a fixed point. Il Nuovo Cimento B \textbf{103}, 369--387 (1989)
\bibitem{Lamb} Lamb, J.S.W., Roberts, J.A.G.; Time-reversal symmetry in dynamical systems:
a survey. Physica D. \textbf{112}, pp. 1–39 (1998)
\bibitem{Munoz1} Mu\~nos-Almaraz, F.J., Freire, E., Gal\'an, J., Vanderbauwhede, A.; Continuation of normal doubly symmetric orbits in conservative reversible systems. Celestial Mech. Dyn. Astron. \textbf{97}, pp. 17–47 (2007)
\bibitem{Munoz2} Mu\~nos-Almaraz, F.J., Freire, E., Gal\'an, J., Doedel, E., Vanderbauwhede, A.; Continuation of periodic orbits in conservative and Hamiltonian systems. Physica D. \textbf{181}, pp. 1–38 (2003)
\bibitem{Vanderbauwhede} Vanderbauwhede, A.; Continuation and bifurcation of multi-symmetric
solutions in reversible Hamiltonian systems. Discrete Contin. Dyn. Syst. Ser. A \textbf{33}, pp. 359–363 (2013)
\bibitem{Bengochea1} Bengochea, A., Her\'andez-Gardu\~no, A., P\'erez-Chavela, E.: New families of periodic orbits in the $4$--body problem emanating from a kite configuration. Appl. Math. Comput. \textbf{398}, 125961 (2021).
\bibitem{Bengochea2} Bengochea, A., Gal\'an-Vioque, J., P\'erez-Chavela, E.: Families of Symmetric Exchange Orbits in the Planar
$(1 + 2n)$--Body Problem. Qual. Theory Dyn. Syst. \textbf{34}, 34 (2021)
\bibitem{Kepley} Kepley S, Mireles James, J.D.: Chaotic motions in the restricted four body problem via Devaney’s saddle-focus
homoclinic tangle theorem. J. Differ. Equ. 1–67. https://doi.org/10.1016/j.jde.2018.08.007 (2018)
\bibitem{Scheereslibro} Scheeres D.J.: Orbital Motion in
Strongly Perturbed Environments. Applications to Asteroid, Comet and Planetary
Satellite Orbiters. Springer-Verlag. Berlin Heidelberg (2012).
\bibitem{Zagouras} Zagouras C. G., Markellos V.V .; Axisymmetric Periodic Orbits of the Restricted Problem in Three Dimensions. Astron. Astrophys. \textbf{59}.  79--89 (1977).
\bibitem{BelbrunoII} Belbruno, E., Frauenfelder, U., and van Koert, O.;  A family of periodic orbits in the three-dimensional lunar problem. Celest Mech Dyn Astr. \textbf{131,}  7 (2019).
\bibitem{Knuth} Knuth, D.E.: The art of computer programming. Vol. 2. Addison-Wesley Publishing Co., Reading, Mass.,second edn, (1981). Seminumerical algorithms, Addison-Wesley Series in Computer Science and Information Processing Koch, H., Schenkel, A.
\bibitem{Jorba} Jorba,\'A., Zou, M.: A software package for the numerical integration of ODEs by means of high-order Taylor
methods. Exp. Math. \textbf{14(1)}, 99–117 (2005)
\bibitem{Lessard} Lessard, J.P., Mireles James J.D., Ransford J.; Automatic differentiation for Fourier series and the radii
polynomial approach.  Phys. D,
\textbf{334} 174-186, (2016).
\bibitem{Mireles} Mireles James J.D., Murray M.; Chebyshev–Taylor parameterization of stable/unstable manifolds for periodic
orbits: implementation and applications. Int. J.Bifur.Chaos Appl. Sci. Engrgy, \textbf{27}(14) 1730050, 32 (2017).
\bibitem{BurgosRodriguez} Burgos-Garc\'ia, J., Rodr\'iguez R.S.: Un vistazo al m\'etodo de Taylor con diferenciaci\'on autom\'atica para problemas de valor inicial. Abstraction \& Application,
\textbf{23} 35-45, (2019).
\bibitem{Keller} Keller H. B.: \textit{Lectures on numerical methods in bifurcation problems, volume 79 of Tata Institute of Fundamental Research Lectures  on  Mathematics  and  Physics}. Published for the Tata Institute of Fundamental Research, Bombay. (1986)
\bibitem{Burgossoftware} Burgos-Garc\'ia, J., Software Programs for the H4BP, https://github.com/numericfcfm/Numerics-for-ode, (2021).
\bibitem{Meyer} Meyer, K., Hall, G., Offin, D.; Introduction to Hamiltonian Dynamical Systems and the N-Body Problem, volume 90 of Applied Mathematical Sciences, Springer, 3rd ed, New York (2017).
\bibitem{Henon1974} H\'enon, M.; Vertical Stability of Periodic Orbits in the Restricted Problem. II. Hill's case. Astronomy and Astrophysics, \textbf{30}, 317-21, (1974).
\bibitem{Henon} H\'enon, M.; Numerical exploration of the restricted problem. V. Hill's case: periodic orbits and their stability. Astronomy and Astrophysics, {\bf 1}, 223-238, (1969).
\bibitem{Michalodimitrakis} Michalodimitrakis M.; Hill's problem: Families of three-dimensional periodic orbits (part I). Astrophysics and Space Science, {\bf 68}, 253-268, (1980).
\bibitem{Russell} Russell R.P.; Global Search for Planar and
Three-Dimensional Periodic Orbits Near Europa. The Journal of the Astronautical Sciences, {\bf 54}, 199–-226, (2006).
\bibitem{Pelaez} Lara, M., et al. Dynamic stabilization of L2 periodic orbits using attitude-orbit coupling effects. Journal of Aerospace Engineering. \textbf{4},  pp. 73-81. (2012)
\bibitem{ZagourasII} Zagouras C. G., Kazantzis P.G .; THREE-DIMENSIONAL PERIODIC OSCILLATIONS GENERATING FROM PLANE PERIODIC ONES AROUND THE COLLINEAR LAGRANGIAN POINTS. Astrophysics and Space Science. \textbf{61}.  389--409 (1978).
\bibitem{Gomez} Gomez, G., Mondelo J.M. The dynamics around the collinear equilibrium points of the RTBP. Physica D, {\bf 157}, 283-321, (2001).


\end{thebibliography}
\end{document}